\documentclass[12pt]{article}

\usepackage[utf8]{inputenc} 
\usepackage[T1]{fontenc}    
\usepackage{hyperref}       
\usepackage{url}            
\usepackage{booktabs}       
\usepackage{amsfonts}       
\usepackage{nicefrac}       
\usepackage{microtype}      
\usepackage[pdftex]{graphicx}  
\usepackage{float}
\usepackage{amsthm}
\usepackage{mathtools}
\usepackage{amsfonts}
\usepackage{amssymb}
\floatstyle{ruled}
\usepackage{xcolor}
\usepackage{bbm}
\usepackage{xr}

\usepackage{adjustbox}
\usepackage{wrapfig}
\usepackage{caption}
\usepackage{subcaption}
\usepackage{cprotect}
\usepackage{etoc}
\usepackage{titlesec}
\usepackage{enumitem}

\usepackage{tikz}
\usetikzlibrary{arrows,automata, calc, positioning}
\usepackage{adjustbox}
\usepackage{wrapfig}

\usepackage{authblk}

\usepackage{multicol}
\usepackage{multirow}
\usepackage{booktabs}

\usepackage{color-edits}
\addauthor{vs}{red}
\theoremstyle{definition}
\newtheorem{theorem}{Theorem}
\newtheorem{assumption}{Assumption}

\newtheorem{proposition}{Proposition}
\newtheorem{definition}{Definition}

\newtheorem{corollary}{Corollary}
\newtheorem{lemma}{Lemma}

\newtheorem{algo}{Estimator}

\DeclareMathOperator*{\argmin}{\arg\!\min}
\DeclareMathOperator*{\argmax}{\arg\!\max}

\newcommand{\blind}{1}

\def\spacingset#1{\renewcommand{\baselinestretch}%
{#1}\small\normalsize} \spacingset{1}

\newcommand{\kibitz}[2]{\ifnum\Comments=1{\color{#1}{#2}}\fi}

\newcommand{\E}{\mathbb{E}}

\newcommand{\R}{\mathbb{R}}

\newcommand{\mcR}{{\mathcal R}}

\newcommand{\mcA}{{\mathcal A}}
\newcommand{\mcH}{{\mathcal H}}

\newcommand{\mcX}{{\mathcal X}}

\newcommand{\mcG}{{\mathcal G}}
\newcommand{\mcF}{{\mathcal F}}

\newcommand{\ba}{\begin{array}}
\newcommand{\ea}{\end{array}}
\newcommand{\bs}{\begin{align}\begin{split}\nonumber}
\newcommand{\bsnumber}{\begin{align}\begin{split}}
\newcommand{\es}{\end{split}\end{align}}

\newcommand{\mcZ}{{\mathcal Z}}

\newcommand{\spanF}{\ensuremath{\text{span}}}

\newcommand{\ldot}[2]{\langle #1, #2 \rangle}

\newcommand{\Var}{\ensuremath{\text{Var}}}

\newcommand{\sign}{\ensuremath{\mathtt{sign}}}

\def\balign#1\ealign{\begin{align}#1\end{align}}
\def\balignat#1\ealign{\begin{alignat}#1\end{alignat}}
\def\bitemize#1\eitemize{\begin{itemize}#1\end{itemize}}
\def\benumerate#1\eenumerate{\begin{enumerate}#1\end{enumerate}}
\newenvironment{talign}
 {\csname align\endcsname}
 {\endalign}
\def\balignt#1\ealignt{\begin{talign}#1\end{talign}}%
\newcommand{\splin}{\ensuremath{\text{splin}}}
\newcommand{\nnet}{\ensuremath{\text{nnet}}}
\newcommand{\rkhs}{\ensuremath{\text{rkhs}}}

\title{\bf Adversarial Estimation of Riesz Representers}
  
\author[1]{Victor Chernozhukov}
\author[1]{Whitney K. Newey}
\author[2]{Rahul Singh}
\author[3]{Vasilis Syrgkanis}
\affil[1]{Department of Economics, Massachussetts Institute of Technology, Cambridge, MA, USA 02142}
\affil[2]{Society of Fellows and Department of Economics, Harvard University, Cambridge, MA, USA 02138}
\affil[3]{Department of Management Science and Engineering, Stanford University, Stanford, CA, USA 94305}

\date{Original draft: December 2020. This draft: April 2024.}
\begin{document}
\maketitle

\bigskip
\begin{abstract}
Many causal parameters are linear functionals of an underlying regression. 
The Riesz representer is a key component in the asymptotic variance of a semiparametrically estimated linear functional. 
We propose an adversarial framework to estimate the Riesz representer using general function spaces. 
We prove a nonasymptotic mean square rate in terms of an abstract quantity called the critical radius, then specialize it for neural networks, random forests, and reproducing kernel Hilbert spaces as leading cases. 
Our estimators are highly compatible with targeted and debiased machine learning with sample splitting; our guarantees directly verify general conditions for inference that allow mis-specification. 
We also use our guarantees to prove inference without sample splitting, based on stability or complexity.
Our estimators achieve nominal coverage in highly nonlinear simulations where some previous methods break down. 
They shed new light on the heterogeneous effects of matching grants. 
\end{abstract}

\noindent%

\if1\blind
{
\medskip

\noindent%
{\it Code:} \url{https://colab.research.google.com/github/vsyrgkanis/adversarial_reisz/blob/master/Results.ipynb}
} \fi

\vfill
\newpage

\spacingset{1.5} 

\section{Introduction and related work}\label{sec:intro}

Many parameters traditionally studied in statistics and econometrics are functionals, i.e. scalar summaries, of an underlying regression function \cite{hasminskii1979nonparametric,pfanzagl1982lecture,klaassen1987consistent,robinson1988root,powell1989semiparametric,van1991differentiable,bickel1993efficient}.
For example, $g_0(X)=\E(Y|X)$ may be a regression function in a space $\mcG$, and the parameter of interest $\theta(g_0)$ may be the average policy effect of transporting the covariates according to $t(X)$, so that the functional is
$
\theta:g\mapsto \E[g\{t(X)\}-g(X)]$ \cite{stock1989nonparametric}.
Under regularity conditions, there exists a function $a_0$ called the Riesz representer that represents $\theta$ in the sense that
$\theta(g)=\E\{a_0(X)g(X)\}$ for all $g\in \mcG.$\footnote{The representer $a_0$ exists when $\theta$ is a bounded functional, which is necessary for regular estimation.}
For the average policy effect, it is the density ratio $a_0(X)=f_t(X)/f(X)-1$, where $f_t(X)$ is the density of $t(X)$ and $f(X)$ is the density of $X$. More generally, for any bounded linear functional $\theta(g)=\E\{m(Z;g)\}$, where $g:\mcX\rightarrow \R$ and $\mcX\subset \mcZ$, there exists a Riesz representer $a_0$ such the functional $\theta(g)$ can be evaluated by simply taking the inner product between $a_0$ and $g$.

Estimating the Riesz representer of a linear functional is a critical building block in a variety of tasks. First and foremost, $a_0$ appears in the asymptotic variance of any semiparametrically efficient estimator of the parameter $\theta(g_0)$, so to construct an analytic confidence interval, we require an estimator $\hat{a}$ \cite{newey1994asymptotic}. Second, because $a_0$ appears in the asymptotic variance, $\hat{a}$ can be directly incorporated into estimation of the parameter $\theta(g_0)$ to ensure semiparametric efficiency \cite{robins1995analysis,robins1995semiparametric,van2006targeted,zheng2011cross,belloni2012sparse,luedtke2016statistical,belloni2017program,chernozhukov2018double,chernozhukov2022locally}. Third, $a_0$ may admit a structural interpretation in its own right. In asset pricing, $a_0$ is the stochastic discount factor, and $\hat{a}$ is used to price financial derivatives \cite{hansen1997assessing,ait1998nonparametric,bansal1993no,chen2019deep}.

The Riesz representer may be difficult to estimate. Even for the average policy effect, its closed form involves a density ratio. A recent literature explores the possibility of directly estimating the Riesz representer, without estimating its components or even knowing its functional form, since the Riesz representer is directly identified from data \cite{robins2007comment,avagyan2021high,chernozhukov2018global,hirshberg2019augmented,chernozhukov2018learning,smucler2019unifying}, generalizing what is known about balancing weights \cite{hainmueller2012entropy,imai2014covariate,zubizarreta2015stable,chan2016globally,athey2018approximate,wong2018kernel,zhao2019covariate,kallus2020generalized,hirshberg2019kernel}.
This literature proposes estimators $\hat{a}$ in specific sparse linear or RKHS function spaces, and analyzes them on a case-by-case basis. In prior work, the incorporation of directly estimated $\hat{a}$ into the tasks above appears to require either (i) correct specification, which may be implausible, or (ii) approximation by linear and RKHS function spaces only.

This paper asks: is there a general mean square rate for an estimator $\hat{a}$ constructed over a general machine learning function space that is not Donsker, e.g. a neural network? Moreover, can we use this mean square rate to incorporate $\hat{a}$ into the tasks listed above while simultaneously allowing general function approximation and robustness to mis-specification? Finally, do these theoretical innovations shed new light in a highly influential empirical study where previous work is limited to parametric estimation?

\textbf{Contributions.}
Our first contribution is a general Riesz estimator over nonlinear function spaces with fast estimation rates. Specifically, we propose and analyze an adversarial, direct estimator $\hat{a}$ for the Riesz representer of any mean square continuous linear functional, using any function space $\mcA$ that can approximate $a_0$ well. We prove a high probability, finite sample bound on the mean square error $\|\hat{a}-a_0\|_2$ in terms of an abstract quantity called the critical radius $\delta_n$, which quantifies the complexity of $\mcA$ in a certain sense.  Since the critical radius is a well-known quantity in statistical learning theory, we can appeal to critical radii of machine learning function spaces such as neural networks, random forests, and reproducing kernel Hilbert spaces (RKHSs).\footnote{Hereafter, a ``general'' function space is a possibly non-Donsker space that satisfies a critical radius condition. An ``arbitrary'' function space may not satisfy a critical radius condition.} Unlike previous work on Riesz representers, we provide a unifying approach for general function spaces (and their unions), handling new function spaces such as neural networks and random forests. These new function spaces achieve nominal coverage in highly nonlinear simulations where other function spaces fail.

Our second contribution is to incorporate our direct $\hat{a}$ estimator into downstream tasks such as semiparametric estimation and inference, while simultaneously allowing general function approximation and robustness to mis-specification. We demonstrate that our mean square rate directly verifies the  general conditions for inference in targeted and debiased machine learning with sample splitting. We prove inference based on ``stability-rate robustness'' without sample splitting, which may be of independent interest: a sufficient condition for Gaussian approximation is when the \textit{product} of the estimator stability and the estimation rate vanishes quickly enough. Finally, our mean square rate also verifies conditions for inference based on ``complexity-rate robustness'' without sample splitting. While showing the connection, we clarify that a sufficient condition for Gaussian approximation is when the \textit{product} of the complexity and the estimation rate is $o_p(n^{-1/2})$. 

In practice, adversarial estimators with machine learning function spaces involve computational techniques that may introduce computational error. As our third contribution, we analyze the computational error for some key function spaces used in adversarial estimation of Riesz representers.  
For random forests, we analyze oracle training and prove convergence of an iterative procedure. For the RKHS, we derive a closed form without computational error and propose a Nystr\"om approximation with computational error. In doing so, we attempt to bridge theory with practice for empirical research.

Finally, we provide an empirical contribution by extending the influential analysis of \cite{karlan2007does} from parametric estimation to semiparametric estimation. To our knowledge, semiparametric estimation has not been used in this setting before. 
The substantive economic question is: how much more effective is a matching grant in Republican states compared to Democratic states? The flexibility of machine learning may improve model fit, yet it may come at the cost of statistical power. Our approach appears to improve model fit relative to previous parametric and semiparametric approaches, and this benefit appears to outweigh the cost; we obtain more precise estimates overall.

\textbf{Key connections.}
The Riesz representation theorem implies a continuum of unconditional moment restrictions.
Our central insight is to adapt adversarial techniques to the problem of learning the Riesz representer:
we adversarially enforce the unconditional moment restrictions over a set of test functions \cite{goodfellow2014generative,arjovsky2017wasserstein}. The fundamental advantage of the adversarial approach is its unified analysis over general function classes in terms of the critical radius \cite{koltchinskii2000rademacher,bartlett2005local,negahban2012,lecue2017regularization,lecue2018regularization}. 
Since this paper was circulated on arXiv in 2020, this framework has been extended to Riesz representers of more functionals, e.g. proximal treatment effects \cite{kallus2021causal,ghassami2022minimax}.

A key predecessor is \cite{chernozhukov2018global} which only studies sparse linear function approximation of Riesz representers. We allow for general function spaces, which may improve finite sample performance by reducing approximation error. Our work complements \cite{hirshberg2019augmented}, who use an adversarial approach to semiparametric estimation, without sample splitting, that requires correct specification of $g_0$. 
Our fast $L_2$ rate is compatible with not only their ``complexity-rate'' inference results but also targeted and debiased machine learning with sample splitting,  which allow mis-specification \cite{zheng2011cross,chernozhukov2018double}.
%
Finally, \cite{kaji2020adversarial} propose an adversarial estimator for parametric models, whereas we study semiparametric models.

The Riesz representer's unconditional moment restrictions differ from those of a nonparametric instrumental variable regression \cite{newey2003instrumental,ai2003efficient}.  We complement previous works that adapt modern, adversarial techniques to nonparametric instrumental variable regression, e.g. \cite{dikkala2020minimax}. A key difference is that we prove a mean square rate, which is necessary for targeted and debiased machine learning.

Several subsequent works build on our work and propose other estimators for nonlinear function spaces \cite{singh2021debiased,chernozhukov2021automatic,kallus2021causal,ghassami2022minimax}, but this paper was the first to give $L_2$ estimation rates for the Riesz representer with nonlinear function approximation, allowing neural networks and random forests to be used for direct Riesz estimation. In addition, \cite{chen2022debiased} extend our results for ``stability-rate robustness.'' See Sections~\ref{sec:rate} and~\ref{sec:inference} for detailed comparisons.

When initially circulated, this paper appeared to be the first to:
(i) propose direct Riesz representer estimators compatible with sample splitting over general non-Donsker spaces;
(ii) provide unified $L_2$ rates in terms of critical radius theory;
(iii) prove inference via estimator stability.
None of (i), (ii), or (iii) appear to be contained in previous works. 

Section~\ref{sec:algo} defines our adversarial estimator of the Riesz representer over general function spaces.
Section~\ref{sec:rate} presents our main result: $L_2$ rates for the Riesz estimator over general function spaces. Section~\ref{sec:inference} shows semiparametric inference via sample splitting, estimator stability, or complexity. Section~\ref{sec:compute} studies the computation error that arises in practice. Section~\ref{sec:sim} showcases settings where our flexible estimator achieves nominal coverage while some previous methods do not, and where these gains provide new, rigorous empirical evidence. Section~\ref{sec:conclusion} concludes. 
\section{Adversarial estimation over general function spaces}\label{sec:algo}

\textbf{Class of functionals.} We study linear and mean square continuous functionals of the form $\theta:g\mapsto \mathbb{E}\{m(Z;g)\}$, where $g_0(X)=\E(Y|X)$ and $X\subset Z$. Denote the $L_2$ norm and associated inner product by $\|g\|_2:=\sqrt{\E[g(X)^2]}$ and $\ldot{g}{g'}_2 := \E_{X}[g(X)\, g'(X)]$.

\begin{assumption}[Mean square continuity]\label{ass:strong-smooth}
There exists some constant $0\leq M<\infty$ such that $ \forall f\in \mcF$, $\E\left\{m(Z;f)^2\right\} \leq M\, \|f\|^2_2$.\footnote{A sufficient condition is when the statement holds $\forall f\in L_2$. We define $\mcF\subset L_2$ below.}
\end{assumption}

    In Appendix~\ref{sec:rate_proof}, we verify that a variety of functionals are mean square continuous under standard conditions, including the average policy effect, regression decomposition, average treatment effect, average treatment on the treated, and local average treatment effect.

Mean square continuity implies boundedness of the functional $\theta$, since $|\E\{m(Z;g)\}|\leq \sqrt{\E\{m(Z;g)^2\}} $. By the Riesz representation theorem, any bounded linear functional over a Hilbert space has a Riesz representer $a_0$ in that Hilbert space.  Therefore Assumption~\ref{ass:strong-smooth} implies that the Riesz representer estimation problem is well defined. 

\textbf{Estimator definition.} 
To state our estimator, we introduce the following notation. Let $\E_n[\cdot]$ denote the empirical average and $\|\cdot\|_{2,n}$ the empirical $\ell_2$ norm, i.e. 
$\|g\|_{2,n} :=\sqrt{\E_n[g(X)^2]}.$ For any function space $\mcG$, let $\text{star}(\mcG):=\{r\, g: g\in \mcG, r \in [0, 1]\}$ denote the star hull and let $\partial \mcG:= \{g-g': g, g'\in \mcG\}$ denote the space of differences. 

We propose Riesz representer estimators that use a general function space $\mcA$, equipped with some norm $\|\cdot\|_{\mcA}$. In particular, $\mcA$ is for the minimization in our min-max approach. Given the notation above, we define the class $\mcF$ for adversarial maximization:
$
\mcF:=\text{star}(\partial\mcA):=\{r(a-a'): a, a'\in \mcA, r\in [0,1]\},    $
and we assume that the norm $\|\cdot\|_{\mcA}$ extends naturally to the larger space $\mcF$.
We propose the following estimator.

\begin{algo}[Adversarial Riesz representer]\label{algo:reg-estimator}
    For regularization $(\lambda,\mu)$, define
    \begin{equation*}
    \hat{a} = \argmin_{a\in \mcA} \max_{f\in \mcF} \E_n\{m(Z; f) - a(X)\cdot f(X)\} - \|f\|_{2,n}^2 - \lambda \|f\|_{\mcA}^2 + \mu \|a\|_{\mcA}^2.
\end{equation*}
\end{algo}

\begin{corollary}[Population limit]\label{cor:limit}
Consider the population limit of our criterion where $n\to \infty$ and $\lambda,\mu \to 0$:
$
\max_{f\in \mcF} \E\left[m(Z; f) - a(X)\cdot f(X)\right] - \|f\|_{2}^2.
$
This limit equals $\frac{1}{4} \|a-a_0\|_2^2$.
\end{corollary}

Thus our empirical criterion converges to the mean square error criterion in the population limit, even though an analyst does not have access to unbiased samples from $a_0(X)$.

Our Riesz representer estimator $\hat{a}\in\mcA$ is a function, which may be evaluated at new locations besides the training data. Therefore it may interpolate, which is important for stochastic discount factor analysis in finance; see Appendix~\ref{sec:asset}. Moreover, it may be evaluated on a held out sample, which is important for semiparametric inference with sample splitting that allows for mis-specification; see Section~\ref{sec:inference}. By contrast, a balancing weight estimator defined as a vector $\tilde{a}\in\mathbb{R}^n$ does not allow for interpolation or sample splitting. This is our main departure from previous work on adversarial balancing weights, described in Section~\ref{sec:intro}.

Crucially, the space $\mcA$ does not have to be sparse linear or an RKHS. This is our main departure from previous work on direct Riesz estimation, described in Section~\ref{sec:intro}.\footnote{By direct Riesz estimation, we mean estimation of a function $\hat{a}\in\mcA$ that can interpolate.}

The vanishing norm-based regularization terms in Estimator~\ref{algo:reg-estimator} may be avoided if an analyst knows a bound on $\|a_0\|_{\mcA}$. In such case, one can impose a hard norm constraint on the hypothesis space and optimize over norm-constrained subspaces. By contrast, regularization allows the estimator to adjust to $\|a_0\|_{\mcA}$, without knowledge of it.

Our analysis allows for mis-specification, i.e. $a_0\notin \mcA$. In such case, the estimation error incurs an extra bias of $\epsilon_n:=\|a_*-a_0\|_{2}$ where $a_*:=\argmin_{a\in \mcA} \|a-a_0\|_2$ is the best-in-class approximation. Section~\ref{sec:inference} interprets $\mcA$ as an $L_2$ approximating sequence of function spaces.

\section{Nonasymptotic mean square rate via critical radius}\label{sec:rate}

Our main result is a fast, finite sample $L_2$ rate for the adversarial Riesz representer in terms of the critical radius. The critical radius of a function class is a widely used quantity in statistical learning theory, and it is known for many machine learning function classes. Thus our main result allows us to appeal to critical radii for a family of adversarial Riesz representer estimators over general function classes. We provide new results for function classes previously unused in direct Riesz representer estimation, e.g. neural networks.

\textbf{Critical radius background.} A Rademacher random variable takes values in $\{-1, 1\}$. Let $\varepsilon_{i}$ be independent Rademacher random variables drawn equiprobably. Then the local Rademacher complexity of the function space $\mcF$ over a neighborhood of radius $\delta$ is defined as ${\cal R}(\delta; \mcF):= \E\left\{\sup_{f\in \mcF: \|f\|_2\leq \delta} \frac{1}{n} \sum_{i=1}^n \varepsilon_i f(X_i)\right\}$. 
As an important stepping stone, we prove bounds on $\|\hat{a}-a_0\|_2$ in terms of local Rademacher complexities.

These bounds can be optimized into fast rates by an appropriate choice of the radius $\delta$, called the critical radius. \textcolor{black}{Formally}, the critical radius of a function class $\mcF$ with range in $[-b, b]$ is defined as any solution $\delta_n$ to the inequality $  {\cal R}(\delta; \mcF)\leq \delta^2/b$.\footnote{We focus on bounded functions for simplicity. Future work may adapt our analysis to unbounded functions using moment conditions.} There is a sense in which $\delta_n$ balances bias and variance in the bounds.

The critical radius has been analyzed and derived for a variety of function spaces of interest, such as neural networks, reproducing kernel Hilbert spaces, high-dimensional linear functions, and VC-subgraph classes. The following characterization of the critical radius opens the door to such derivations \cite[Corollary~14.3 and Proposition~14.25]{wainwright2019high}:
the critical radius of any function class $\mcF$, uniformly bounded in $[-b,b]$, is of the same order as any solution to the inequality:
\begin{equation}\label{eqn:metric-entropy-critical}
     \int_{\frac{\delta^2}{2b}}^{\delta} \sqrt{\log\left[N\{\varepsilon; B_n(\delta; \mcF);\ell_2\}\right]} d\varepsilon \leq  \frac{\sqrt{n}\delta^2}{64b}.
\end{equation}
In this expression, $B_n(\delta; \mcF)=\{f\in \mcF: \|f\|_{2,n}\leq \delta\}$ is the ball of radius $\delta$ and $N(\varepsilon; \mcF;\ell_2)$ is the empirical $\ell_2$-covering number at approximation level $\varepsilon$, i.e. the size of the smallest $\varepsilon$-cover of $\mcF$, with respect to the empirical $\ell_2$ metric. Critical radius analysis will handle more function spaces than Donsker analysis; see Appendix~\ref{sec:inference_proof}.

\begin{assumption}[Critical radius for estimation]\label{ass:critical}
    Define the balls of functions $
    \mcF_B:=\{f\in \mcF: \|f\|_{\mcA}^2\leq B\}$ and $
    m\circ \mcF_B:=\{m(\cdot; f): f\in \mcF_B\}
$. Assume that there exists some constant $B\geq 0$ such that the functions in $\mcF_B$ and $m\circ \mcF_B$ have uniformly bounded ranges in $[-b, b]$. Further assume that, for this $B$, $\delta_n$ upper bounds the critical radii of $\mcF_{B}$ and $m\circ \mcF_B$.
\end{assumption}

\textbf{Main result.} Our nonasymptotic main result holds with probability $1-\zeta$. To lighten notation, we summarize the critical radii, approximation error, and low probability event:
$
\bar{\delta}:=\delta_n + \epsilon_n + c_0 \sqrt{\frac{\log(c_1/\zeta)}{n}},
$
where $c_0$ and $c_1$ are universal constants. 

\begin{theorem}[Mean square rate]\label{thm:reg-main-error}
Suppose Assumptions~\ref{ass:strong-smooth} and~\ref{ass:critical} hold. Suppose that the regularization in Estimator~\ref{algo:reg-estimator} satisfies $\mu \geq 6\lambda \geq 12\bar{\delta}^2/B$. Then with probability $1-\zeta$,
$
    \|\hat{a} - a_0\|_2 = O\left(M^2 \bar{\delta} + \mu \|a_*\|_{\mcA}^2/\bar{\delta}\right).
$
Furthermore, if $\mu\leq C \bar{\delta}^2/B$, for some constant $C$, the bound simplifies to $O\left[\bar{\delta}\,\max\left\{M^2, \frac{\|a_*\|_{\mcA}^2}{B}\right\}\right]$.
\end{theorem}

\begin{corollary}[Weaker metric rate]\label{cor:weaker_metric}
    Consider the weaker metric $\|\cdot\|_{\mcF}$ defined as $\|a\|_{\mcF}^2 = \sup_{f\in \mcF}\, \ldot{a}{f}_2 - \frac{1}{4}\|f\|_2^2 \leq \|a\|_2^2$. Then under the conditions of Theorem~\ref{thm:reg-main-error}, 
    $
    \|\hat{a}-a_0\|_{\mcF} =O\left[\bar{\delta}' \max\left\{M^2, \|a_*\|_{\mcA}^2/B\right\}\right]
    $
    where the definition of $\bar{\delta}'$ replaces $\epsilon_n$ with $\epsilon'_n = \inf_{a\in \mcA}\|a-a_0\|_{\mcF}$.
\end{corollary}

\begin{corollary}[Mean square rate without norm regularization]\label{cor:main-error}
    Suppose the conditions of Theorem~\ref{thm:reg-main-error} hold, and that the function classes $\mcF$ and $\mcG$ are already norm constrained with uniformly bounded ranges. Then taking $\lambda=\mu=0$ in Estimator~\ref{algo:reg-estimator},
    $
     \|\hat{a} - a_0\|_2 = O\left(M^2 \bar{\delta}''\right)
    $
    where we define $\bar{\delta}''$ by replacing $\delta_n$ with $\delta_n''$, which bounds the critical radii of  $\mcF$ and $m\circ \mcF$.
\end{corollary}

    In Appendix~\ref{sec:sparse}, we consider another version of Estimator~\ref{algo:reg-estimator} without the $\|f\|^2_{2,n}$ term. We prove a fast rate on $\|\hat{a}-a_0\|_2$ as long as $\mcF$ can be decomposed into the union of symmetric function spaces, i.e. $\mcF=\cup_{i=1}^d \mcF^i$. 


We bound the population norm $\|\hat{a}-a_0\|_{2}$ in terms of critical radii, for function spaces beyond the sparse linear case and the RKHS. The population norm may be viewed as a generalization error, and it is necessary for downstream analysis where the function $\hat{a}$ interpolates or is evaluated on a held-out sample. 
We complement previous work on $L_2$ rates by providing a unified analysis across function spaces, including new ones. See Appendix~\ref{sec:sparse} for formal comparisons in the sparse linear special case, where strong results are known.

We study the Riesz representer, whereas \cite{dikkala2020minimax} study the nonparametric instrumental variable regression. Theorem~\ref{thm:reg-main-error} and \cite[Theorem 1]{dikkala2020minimax} use adversarial techniques, yet there is a crucial difference in the nature of the result: we bound the mean square error, whereas \cite[Theorem 1]{dikkala2020minimax} bounds a \textit{projected} mean square error. The mean square error rate is essential to verify the conditions of targeted and debiased machine learning; a projected mean square error rate is weaker and insufficient.

\begin{corollary}[Union of hypothesis spaces]\label{cor:union}
    Suppose that $\mcF=\cup_{i=1}^d \mcF^i$ and that the critical radius of each $\mcF^i$ is $\delta_n^i$ in the sense of Assumption~\ref{ass:critical}. 
 Suppose that for some  $\mcF^i$, $\delta^i_n \gtrsim \sqrt{\frac{\log(d)}{n}}$. 
 Then the critical radius of $\mcF$ is $\delta_n=O\left\{\max_{i}\delta_n^i+\sqrt{\frac{\log(d)}{n}}\right\}$, and the conclusions of the above results continue to hold.
\end{corollary}

    By Corollary~\ref{cor:union}, our results allow an analyst to estimate the Riesz representer with a union of flexible function spaces, which is an important practical advantage. This is a specific strength of the critical radius approach and a contribution to the direct Riesz estimation literature, which appears to have been studied one function space at a time.\footnote{Estimator~\ref{algo:reg-estimator} takes $\mcF$ as given, allowing for unions. Future work may study adaptive selection of $\mcF$.}

\textbf{Special cases: Neural network, random forest, RKHS.}
As a leading example, suppose that the function class $\mcA$ can be expressed as a rectified linear unit (ReLU) activation neural network with depth $L$ and width $W$, denoted as $\mcA_{\nnet(L,W)}$. Functions in $\mcF$ can be expressed as neural networks with depth $L+1$ and width $2W$. 

\begin{corollary}[Neural network Riesz representer rate]\label{cor:nn}
   Take $\mcA=\mcA_{\nnet(L,W)}$. Suppose that $m\circ \mcF$ is representable as a neural network with depth $O(L)$ and width $O(W)$. Finally, suppose that the covariates are such that functions in $\mcF$ and $m\circ \mcF$ are uniformly bounded in $[-b,b]$. Then with probability $1-\zeta$,
   $
        \|\hat{a}-a_0\|_{2} = O\big\{\min_{a\in A_{\nnet(L,W)}} \|a-a_0\|_2 + \sqrt{\frac{L\, W\, \log(W)\,\log(b)\, \log(n)}{n}} + \sqrt{\frac{\log(1/\zeta)}{n}}\big\}.
   $
\end{corollary}

The second term is the critical radius. By the $L_1$ covering number for VC classes \cite{haussler1995sphere} as well as the bounds of \cite[Theorem 14.1]{anthony2009neural} and \cite[Theorem~6]{bartlett2019nearly}, the critical radii of $\mcF$ and $m\circ \mcF$ are $\delta_n=O\left\{\sqrt{\frac{L\, W\, \log(W)\,\log(b)\, \log(n)}{n}}\right\}.$
See \cite[Proof of Example~3]{foster2019orthogonal} for a detailed derivation.

    If $a_0$ is representable as a ReLU neural network, then the first term vanishes and we achieve an almost parametric rate. If $a_0$ is representable as a nonparametric Holder function, then one may appeal to approximation results for ReLU activation neural networks \cite{yarotsky2017error,yarotsky2018optimal}. Such results typically require that the depth and the width of the neural network grow as some function of the approximation error $\epsilon_n$, leading to errors of the form $O\left[\epsilon_n + \sqrt{\frac{L(\epsilon_n)\, W(\epsilon_n)\, \log\{W(\epsilon_n)\}\,\log(b)\, \log(n)}{n}} + \sqrt{\frac{\log(1/\zeta)}{n}}\right]$. Optimally balancing $\epsilon_n$ leads to almost tight nonparametric rates.

    Corollary~\ref{cor:nn} for the Riesz representer is the same order as \cite[Theorem 2]{farrell2018DeepNeural} for nonparametric regression. In this sense, it appears relatively sharp.

Next, consider the oracle trained random forest estimator described in Section~\ref{sec:compute}. Denote by $\mcA_{\text{base}(d)}$ the base space, with VC dimension $d$, in which each tree of the forest is estimated. Denote by $\mcA_{\text{rf}(d)}=\{\sum_{j}w_j \tilde{a}_j: \tilde{a}_j\in \mcA_{\text{base}(d)}\}$ the linear span of the base space.

\begin{corollary}[Random forest Riesz representer rate]\label{cor:rf}
     Take $\mcA=\mcA_{\text{rf}(d)}$. \textcolor{black}{Suppose that}
     $m\circ \mcF\in \mcA_{\text{rf}(d)}$. Finally, suppose that the covariates are such that functions in $\mcF$ and $m\circ \mcF$ are uniformly bounded in $[-b,b]$. Then under after $T$ iterations of oracle training and under the regularity conditions described in Proposition~\ref{prop:oracle},  with probability $1-\zeta$,
   $
        \|\hat{a}-a_0\|_{2} = O\big\{\min_{a\in A_{\text{rf}(d)}} \|a-a_0\|_2 +\frac{\log(T)}{T}+b\sqrt{\frac{Td\log(n)}{n}} + \sqrt{\frac{\log(1/\zeta)}{n}}\big\}.
   $
\end{corollary}

    The second term is the computational error from Proposition~\ref{prop:oracle} below. The third term is the critical radius, which follows from the complexity of oracle training and from analysis of VC spaces \cite{shalev2014understanding}. We defer further discussion to Section~\ref{sec:compute}.

    Suppose that the base estimator is a binary decision tree with small depth. This example satisfies the requirement of $\mcA_{\text{base}(d)}$ \cite{Mansour2000}, and we use it in practice. 
%
    If $T=O\{(n/d)^{1/3}\}$, the bound simplifies to $ \|\hat{a}-a_0\|_{2} = O\left\{b\frac{\log(n)d^{1/3}}{n^{1/3}} + \sqrt{\frac{\log(1/\zeta)}{n}}\right\}.$

Denote by $\mcA_{\rkhs(k)}$ the RKHS with kernel $k$, so that $\|\cdot\|_{\mcA_{\rkhs(k)}}$ is the RKHS norm. Define the empirical kernel matrix $K\in\R^{n\times n}$ where $K_{ij}=k(x_i, x_j)/n$. Let $(\eta_j)_{j=1}^n$ be the eigenvalues of $K$. For a possibly different kernel $\tilde{k}$, we define analogous objects with tildes. We arrive at the following result using \cite[Corollary~13.18]{wainwright2019high}.

\begin{corollary}[RKHS Riesz representer rate]\label{cor:rkhs}
    Take $\mcA=\mcA_{\rkhs(k)}$. \textcolor{black}{Suppose that} 
    $m\circ \mcF\in \mcA_{\rkhs(\tilde{k})}$. Suppose that there exists some $B\geq 0$ such that functions in $\mcF_B$ and $m\circ \mcF_B$ are uniformly bounded in $[-b, b]$. Let $\delta_n$ be any solution to the inequalities $B\sqrt{\frac{2}{n}}\sqrt{\sum_{j=1}^\infty \max\{\eta_j, \delta^2\}}\leq \delta^2$ and $   B\sqrt{\frac{2}{n}}\sqrt{\sum_{j=1}^\infty \max\{\tilde{\eta}_j, \delta^2\}}\leq \delta^2$. Then with probability $1-\zeta$,
    $
    \|\hat{a}-a_0\|_{2} =O\left[\min_{a\in A_{\rkhs(k)}} \|a-a_0\|_2 +\|a_0\|_{\mcA}\left\{\delta_n + \sqrt{\frac{\log(1/\zeta)}{n}}\right\}\right].
    $
\end{corollary}

    Our estimator does not need to know the RKHS norm $\|a_0\|_{\mcA}$. Instead it automatically adjusts to the unknown RKHS norm. 
%
    The bound $\delta_n$ is based on empirical eigenvalues. These empirical quantities can be used as a data-adaptive diagnostic.

    For particular kernels, a more explicit bound can be derived as a function of the eigendecay. For example, the Gaussian kernel has an exponential eigendecay. \cite[Example~13.21]{wainwright2019high}  derives $\delta_n=O\left\{b\sqrt{\frac{\log(n)}{n}}\right\}$, thus leading to almost parametric rates: $\|\hat{a}-a_0\|_2= O\left[\min_{a\in A_{\rkhs(k)}} \|a-a_0\|_2 +\|a_0\|_{\mcA} \left\{b\sqrt{\frac{\log(n)}{n}}+ \sqrt{\frac{\log(1/\zeta)}{n}}\right\}\right]$.

    See Appendix~\ref{sec:sparse} for sparse linear function spaces and further comparisons. 

\section{Semiparametric inference}\label{sec:inference}

So far, we have analyzed a machine learning estimator $\hat{a}$ for $a_0$, the Riesz representer to the mean square continuous functional $\theta:g\mapsto \E\{m(Z;g)\}$. A well known use of $\hat{a}$ is to construct a consistent, asymptotically normal, and semiparametrically efficient estimator for a parameter $\theta_0:=\theta(g_0)\in\R$. For example, when the parameter $\theta_0$ is the average policy effect, we have $g_0(X)=\E(Y|X)$, $\theta(g)=\E[g\{t(X)\}-g(X)]$, and $a_0(X)=f_t(X)/f(X)-1$. 

In this section, we use our main result to prove inference for three estimators of $\theta_0$: (i) targeted machine learning, (ii) debiased machine learning, and (iii) the doubly robust estimator without sample splitting. We directly verify known conditions for (i) and (ii) via sample splitting \cite{bickel1982adaptive,schick1986asymptotically,klaassen1987consistent}. For (iii), we prove new inference results via estimator stability, and clarify inference results via estimator complexity. 

\textbf{Estimator definition.} In what follows, let $ m_{a}(Z; g) := m(Z; g) + a(X) \{Y - g(X)\}$.

\begin{algo}[Targeted machine learning \cite{zheng2011cross,chernozhukov2018learning}]\label{algo:tml}
Partition $n$ observations into $K=\Theta(1)$ folds $P_1, \ldots, P_K$. For each fold, estimate $\hat{a}_k, \hat{g}_k$ based on the observations outside of fold $P_k$. Finally construct the estimate $\tilde{\theta} = \frac{1}{n} \sum_{k=1}^K \sum_{i\in P_k} m(Z_i; \tilde{g}_k)$ where $\tilde{g}_k(x)=\hat{g}_k(x)+\frac{\sum_{i\in P_k} \hat{a}_k(X_i)\{Y_i-\hat{\gamma}_k(X_i)\}}{\sum_{i\in P_k} \hat{\alpha}_k(X_i)^2}\hat{a}_k(x)$.
\end{algo}

\begin{algo}[Debiased machine learning \cite{levit1976efficiency,hasminskii1979nonparametric,chernozhukov2018double}]\label{algo:dml}
Partition $n$ observations into $K=\Theta(1)$ folds $P_1, \ldots, P_K$. For each fold, estimate $\hat{a}_k, \hat{g}_k$ based on the observations outside of fold $P_k$. Finally construct the estimate $\check{\theta} = \frac{1}{n} \sum_{k=1}^K \sum_{i\in P_k} m_{\hat{a}_k}(Z_i; \hat{g}_k)$.
\end{algo}

\begin{algo}[Doubly robust estimator \cite{robins1995analysis,robins1995semiparametric,chernozhukov2022locally}]\label{algo:debias-nocross}
    Estimate $(\hat{g},\hat{a})$ using all observations. Set $\hat{\theta}= \E_n\left\{m_{\hat{a}}(Z; \hat{g})\right\}$.
\end{algo}

\textbf{Normality via sample splitting.}
We use a weak and well known condition for nuisance estimators $\hat{g}_k$ and $\hat{a}_k$: the mixed bias $\E[\{\hat{a}_k(X)-a_0(X)\}\, \{\hat{g}_k(X) - g_0(X)\}]$ vanishes quickly.

\begin{assumption}[Mixed bias condition]\label{ass:main-cond}
    Suppose that $\forall k \in [K]$: $\sqrt{n}\, \E[\{\hat{a}_k(X)-a_0(X)\}\, \{\hat{g}_k(X) - g_0(X)\}] \rightarrow_p 0$.\footnote{Without sample splitting, $K=1$ and this condition remains well defined.}
\end{assumption}

    By Cauchy-Schwarz inequality, Assumption~\ref{ass:main-cond} is implied by $\sqrt{n}\|\hat{a}_k-a_0\|_2 \|\hat{g}_k - g_0\|_2 \to_p 0$. The latter is the celebrated double rate robustness condition, also called the product rate condition, whereby either $\hat{g}_k$ or $\hat{a}_k$ may have a relatively slow estimation rate, as long as the other has a sufficiently fast estimation rate. 

    The mixed bias condition is weaker than double rate robustness: $\hat{a}_k$ only needs to approximately satisfy Riesz representation for test functions of the form $f=\hat{g}_k-g_0$. Hence if $\|\hat{g}_k-g_0\|_2\leq r_n$, then it suffices for $\hat{a}_k$ to be a local Riesz representer around $\hat{g}_k$ rather than a global Riesz representer for all of $\mcG$. In Appendix~\ref{sec:inference_proof}, we prove that Riesz estimation may become much simpler if the only aim is to satisfy Assumption~\ref{ass:main-cond}.

Assumption~\ref{ass:main-cond} leaves limited room for mis-specification. In Appendix~\ref{sec:inference_proof}, we allow for inconsistent nuisance estimation: the probability limit of $\hat{g}_k$ may not be $g_0$, or the probability limit of $\hat{a}_k$ may not be $a_0$, as long as the other nuisance is correctly specified and converges at the parametric rate \cite{Benkeser2017}. Below, we focus on the thought experiment where, for a fixed $n$, the best in class approximations are $(g_*,a_*)$ which may not coincide with $(g_0,a_0)$. In the limit, the function spaces become rich enough to include $(g_0,a_0)$.

\begin{corollary}[Normality via sample splitting]\label{cor:debias}
    Suppose Assumptions~\ref{ass:strong-smooth} and~\ref{ass:main-cond} hold. Further assume
  (i) boundedness: $Y$, $g(X)$, and $a(X)$ are bounded almost surely, for all $g\in \mcG$ and $a\in \mcA$; 
        (ii) individual rates:  $\|\hat{a}_k-a_*\|_2\stackrel{L^2}{\to} 0$  and $\|\hat{g}-g_*\|_2 \stackrel{L^2}{\to} 0$, where $g_*$ or $a_*$ may not necessarily equal $g_0$ or $a_0$.
    Then and $ \sqrt{n}\sigma_*^{-1}\left(\tilde{\theta} - \theta_0\right) \to_d N\left(0, 1\right)$ and $ \sqrt{n}\sigma_*^{-1}\left(\check{\theta} - \theta_0\right) \to_d N\left(0, 1\right)$, where $\sigma_*^2 := \Var\{m_{a_*}(Z; g_*)\}$ may be a sequence indexed by the sample size..
\end{corollary}

    Boundedness in Corollary~\ref{cor:debias} can be relaxed to bounded fourth moments of $Y, g(X), a(X)$, as long as we strengthen the individual rates  to $\|\hat{a}_k-a_*\|_4\to_p 0$ and $ \|\hat{g}_k-g_*\|_4 \to_p 0$.

    A special case takes $g_*=g_0$ and $a_*=a_0$. More generally, we consider the possibility that $\|g_*-g_0\|_2\leq \epsilon_n$ and $\|a_*-a_0\|_2\leq \epsilon_n$, where $r_n \ll \epsilon_n \ll 1$. For example, consider the thought experiment where $(\mcG,\mcA)$ are sequences of function spaces that approximate the nuisances $(g_0,a_0)$ increasingly well as the sample size increases. Then by Cauchy-Schwarz and triangle inequalities, a sufficient condition for Assumption~\ref{ass:main-cond} is that $\sqrt{n}(r_n^2+\epsilon_n^2)\rightarrow0$. In other words, for a fixed sample size, $(\mcG,\mcA)$ may not include $(g_0,a_0)$; it suffices that they do so in the limit, and that the product of their approximation errors $\epsilon_n^2$ vanishes quickly enough. In this thought experiment, $\sigma_*^2$ is a sequence indexed by the sample size as well.

\textbf{Normality via estimator stability.}
Next, we turn to Estimator~\ref{algo:debias-nocross}, which does not split the sample. Sample splitting may come at a finite sample cost since it reduces the effective sample size, as shown in simulations in Appendix~\ref{sec:sim_detail}. Our theoretical contribution in this section is to prove semiparametric inference of machine learning estimators without sample splitting, the Donsker condition, or even the critical radius bound. Of independent interest, we characterize ``stability-rate robustness'' as a sufficient condition for inference.

\begin{assumption}[Estimator stability for inference]\label{ass:stability}
    Let $\hat{h}:=(\hat{a}, \hat{g})$ and let $\hat{h}^{-i}$ be the estimated function if sample $i$ were removed from the training set. Assume $\hat{h}$ is symmetric across samples and satisfies $\E_Z\{\|\hat{h}(Z) - \hat{h}^{-i}(Z)\|_{\max}^2\} \leq \beta_n$.
\end{assumption}

    \cite{Kale11cross-validationand} propose Assumption~\ref{ass:stability} in order to derive improved bounds on cross validation. It is formally called $\beta_n$ mean square stability, which is weaker than the well studied uniform stability \cite{Bousquet2002StabilityAG}. See \cite{elisseeff2003leave,Celisse2016,pmlr-v98-abou-moustafa19a} for further discussion.

\begin{theorem}[Normality via estimator stability]\label{thm:debias-nocross-stability}
     Suppose Assumptions~\ref{ass:strong-smooth},~\ref{ass:main-cond}, and~\ref{ass:stability} hold. Further assume
  (i) boundedness: $Y$, $g(X)$, and $a(X)$ are bounded almost surely, for all $g\in \mcG$ and $a\in \mcA$; 
        (ii) individual rates:  $\E(\|\hat{a}-a_*\|^2_2)=o_p(r^2_n)$  and $\E(\|\hat{g}-g_*\|^2_2) = o_p(r^2_n)$, where $g_*$ or $a_*$ may not necessarily equal $g_0$ or $a_0$;
        (iii) joint rates: $r^2_{n-1}+n\beta_{n-1}r_{n-2}\to 0$.
    Then $ \sqrt{n}\sigma_*^{-1}\left(\hat{\theta} - \theta_0\right) \to_d N\left(0, 1\right)$ where $\sigma_*^2 := \Var\{m_{a_*}(Z; g_*)\}$.
\end{theorem}

Sub-bagging means using as an estimator the average of several base estimators, where each base estimator is calculated from a subsample of size $s<n$. The sub-bagged estimator is stable with $\beta_n=\frac{s}{n}$ \cite{elisseeff2003leave}. If the base estimator's bias decays as some function $\textsc{bias}(s)$, then  sub-bagged estimators typically achieve $r_n = \sqrt{\frac{s}{n}} + \textsc{bias}(s)$  \cite{Athey2016,Khosravi2019,Syrgkanis2020}. For our results, it suffices that the \textit{product} of stability and the rate vanishes quickly,  i.e. $n\beta_n r_n = \sqrt{\frac{s^3}{n}} + s\, \textsc{bias}(s) \to 0$. If $s=o(n^{1/3})$ and $\textsc{bias}(s)=o(1/s)$, our joint rate condition holds.

Consider a high dimensional setting with $p\gg n$. Suppose that only $r \ll n$ variables are $\mu$ strictly relevant, i.e. each decreases explained variance by least $\mu>0$. \cite{Syrgkanis2020} show that the bias of a deep Breiman tree trained on $s$ observations decays as $\exp(-s)$ in this setting. A deep Breiman forest, where each tree is trained on $s=O\left(\frac{2^r \log(p)}{\mu}\right)=o(n^{1/3})$ samples drawn without replacement, achieves $r_n = O\left(\sqrt{\frac{s 2^r}{n}}\right)$. Thus sub-bagged deep Breiman random forests satisfy the conditions of Theorem~\ref{thm:debias-nocross-stability} in a sparse, high-dimensional setting. In subsequent work, \cite{chen2022debiased} accommodate more types of random forests by refining our ``stability-rate robustness''.

\textbf{Normality via critical radius.} Finally, we study Estimator~\ref{algo:debias-nocross} and clarify a ``complexity-rate robustness'' sufficient condition for inference. Consider the following critical radius assumption for inference, slightly abusing notation by recycling the symbols $\delta_n$ and $\bar{\delta}$.

\begin{assumption}[Critical radius for inference]\label{ass:critical2}
   Assume that, with high probability, $\hat{g} \in \hat{\mcG}\subseteq \mcG$ and $\hat{a} \in \hat{\mcA}\subseteq \mcA$. Moreover, assume that there exists some constant $B\geq 0$ such that the functions in $(\hat{\mcG}-g_*)_{B}$, $\{m\circ (\hat{\mcG}-g_*)\}_B$, and $(\hat{\mcA}-a_*)_B$ have uniformly bounded ranges in $[-b, b]$, and such that with high probability $\|\hat{g}-g_*\|_{\hat{\mcG}} \leq B$ and $\|\hat{a}-a_*\|_{\hat{\mcA}} \leq B$. Further assume that, for this $B$, $\delta_n$ upper bounds the critical radii of $(\hat{\mcG}-g_*)_{B}$, $\{m\circ (\hat{\mcG}-g_*)\}_B$, and $(\hat{\mcA}-a_*)_B$. Assume $\delta_n$ is lower bounded by $\sqrt{\frac{\log\log(n)}{n}}$, and that $|a_*(X)|$ and $|Y-g_*(X)|$ are bounded almost surely.\footnote{The lower bound is a weak regularity condition for concentration in nonlinear settings via Bernstein style arguments \cite{wainwright2019high,foster2019orthogonal}.}  
\end{assumption}

    In the special case that $\hat{\mcG}=\mcG$ and $\hat{\mcA}=\mcA$, Assumption~\ref{ass:critical2} simplifies to a bound on the critical radii of $\mcG_{B}$, $m\circ \mcG_{B}$, and $\mcA_{B}$. More generally, we allow the possibility that only the critical radii of \textit{subsets} of function spaces, where the estimators are known to belong, are well behaved. This nuance is helpful for sparse linear settings, where $\hat{\mcG}$ is a restricted cone that is much simpler than $\mcG$.
%
As before, to lighten notation, we define the following summary of the critical radii: $\bar{\delta}=\delta_n + c_0 \sqrt{\frac{\log(c_1\, n)}{n}}$, where $c_0$ and $c_1$ are universal constants.

\begin{theorem}[Normality via critical radius]\label{thm:debias-nocross}
    Suppose Assumptions~\ref{ass:strong-smooth},~\ref{ass:main-cond}, and~\ref{ass:critical2} hold. Further assume
  (i) boundedness: $Y$, $g(X)$, and $a(X)$ are bounded almost surely, for all $g\in \mcG$ and $a\in \mcA$; 
        (ii) individual rates:  $\|\hat{a}-a_*\|_2=o_p(r_n)$  and $\|\hat{g}-g_*\|_2 = o_p(r_n)$, where $g_*$ or $a_*$ may not necessarily equal $g_0$ or $a_0$;
        (iii) joint rates: $\sqrt{n}\left(\bar{\delta}\, r_n + \bar{\delta}^2\right)\to 0$.
    Then $ \sqrt{n}\sigma_*^{-1}\left(\hat{\theta} - \theta_0\right) \to_d N\left(0, 1\right)$ where $\sigma_*^2 := \Var\{m_{a_*}(Z; g_*)\}$.
\end{theorem}

 Without sample splitting, inferential theory often requires the Donsker condition or slowly increasing entropy. 
 Theorem~\ref{thm:debias-nocross} replaces such conditions with $\sqrt{n}\left(\bar{\delta}\, r_n + \bar{\delta}^2\right)\to 0$, a permissive  complexity bound
     in terms of the critical radius that allows for machine learning. It provides a rather sharp characterization of an important trade-off.
    In particular, we interpret $\sqrt{n}\left(\bar{\delta}\, r_n + \bar{\delta}^2\right)\to 0$ as ``complexity-rate robustness''.
    For general function spaces, $\bar{\delta}$ may vanish slowly, as long as $r_n$ vanishes quickly enough to compensate. This condition excludes certain function spaces, e.g. those for which the integral in~\eqref{eqn:metric-entropy-critical} diverges as $\delta \downarrow 0$.

Theorem~\ref{thm:debias-nocross} refines and extends \cite[eq. 19]{hirshberg2019augmented}. ``Complexity-rate robustness'' is a simple heuristic, which appears not to have been explicitly stated before. It is compatible with any $\hat{a}$ estimator satisfying its conditions, rather than a specific choice of balancing weights defined as a vector in $\mathbb{R}^n$. Moreover, it tolerates some mis-specification.

Appendix~\ref{sec:sparse} compares complexity-rate robustness with double rate robustness in the sparse setting. Complexity-rate robustness says that if both nuisances are moderately sparse, then sample splitting can be eliminated, improving the effective sample size. Double rate robustness says that one nuisance may be quite dense while the other is quite sparse, if we use sample splitting. Appendix~\ref{sec:sparse} also shows how complexity-rate robustness recovers known sufficient conditions in the lasso literature; in this sense, it appears relatively sharp.
\section{Analysis of computational error}\label{sec:compute}

After studying $\hat{a}$ in Section~\ref{sec:rate} and $(\tilde{\theta},\check{\theta},\hat{\theta})$ in Section~\ref{sec:inference}, we attempt to bridge theory with practice by analyzing the computational error for some key function spaces used in adversarial estimation. For random forests, we prove convergence
of an iterative procedure. For the RKHS, we derive a closed form without computational
error. For neural networks, we describe existing results in Appendix~\ref{sec:compute_proof}. Appendix~\ref{sec:compute_proof} also discusses how to choose the regularization hyperparameter values $(\lambda,\mu)$ in accordance with Theorem~\ref{thm:reg-main-error}.
This section and Appendix~\ref{sec:compute_proof} aim to provide practical guidance for empirical researchers.

\textbf{Oracle training for random forest.}
Consider Corollary~\ref{cor:rf}, which uses random forest function spaces. We analyze an optimization procedure called oracle training to handle this case. It may be viewed as a particular criterion for fitting the random forest. More generally, it is an iterative optimization procedure based on zero sum game theory for when $\mcA$ is a non-differentiable function space, hence gradient based methods do not apply.

In this exposition, we study a variation of Estimator~\ref{algo:reg-estimator} with $\lambda=\mu=0$, i.e. without norm-based regularization. Define $\ell(a,f):=\E_n\left[m(Z;f) - a(X)\cdot f(X) - f(X)^2\right]$. We view $\ell(a,f)$ as the payoff of a zero sum game where the players are $a$ and $f$. The game is linear in $a$ and concave in $f$, so it can be solved when $f$ plays a no-regret algorithm at each period $t\in \{1,..., T\}$ and $a$ plays the best response to each choice of $f$.

\begin{proposition}[Oracle training converges]\label{prop:oracle}
Suppose that $g\mapsto \E_n[m(Z; g)]$ has operator norm $1\leq M_n<\infty$, and that $\mcF$ is convex. Suppose that at each period $t\in \{1, \ldots, T\}$, the players follow $f_t = \argmax_{f\in \mcF} \ell(\bar{a}_{t-1}, f)$ and $a_t =\argmin_{a\in \mcA} \ell(a, f_t)$, where $\bar{a}_{t}:= \frac{1}{t} \sum_{\tau=1}^t a_{\tau}$. Then for $T=\Theta\left(\frac{M_n\, \log(1/\epsilon)}{\epsilon}\right)$, $\bar{a}_T$ is an $\epsilon$-approximate solution to Estimator~\ref{algo:reg-estimator} with $\lambda=\mu=0$.
\end{proposition}

    Since $\mcF$ is convex, a no-regret strategy for $f$ is follow-the-leader. At each period, the player maximizes the empirical past payoff $\ell(\bar{a}_{t-1}, f)$. This loss may be viewed as a modification of the Riesz loss. In particular, construct a new functional that is the original functional minus $hf$ then estimate its Riesz representer.

 For any fixed $f_t$, the best response for $a$ is to minimize $\ell(a,f_t)$. Since $\ell(a,f)$ is linear in $a$, this is the same as maximizing $\E_n[a(X) \cdot f_t(X)]$. In other words, $a$ wants to match the sign of $f$. In summary, the best response for $a$ is equivalent to a weighted classification oracle, where the label is $\sign\{f(X_i)\}$ and the weight is $ |f(X_i)|$.

    Since $\ell(a,f)$ is linear in $a$, each $a_t$ is supported on only $t$ elements $\tilde{a}_1,...,\tilde{a}_t$ in the base space. In Corollary~\ref{cor:rf}, we assume that each element of the base space has VC dimension at most $d$. Hence each $a_t$ has VC dimension at most $dt$ and therefore $\bar{a}_T$ has VC dimension at most $dT$ \cite{shalev2014understanding}. Thus the entropy integral~\eqref{eqn:metric-entropy-critical} is of order $\sqrt{\frac{Td \log(n)}{n}}$. Finally, since the $f$ player problem reduces to a modification of the Riesz problem, this bound applies to both of the induced function spaces in oracle training.

\textbf{Closed form for RKHS.}
Consider the setting of Corollary~\ref{cor:rkhs}, which uses RKHSs. We prove that Estimator~\ref{algo:reg-estimator} has a closed form solution without any computational error, and derive its formula.  Our results extend the classic representation arguments of \cite{kimeldorf1971some,scholkopf2001generalized}. We use backward induction, first analyzing the best response of an adversarial maximizer, which we denote by $\hat{f}_a$, as a function of $a$. Then we derive the minimizer $\hat{a}$ that anticipates this best response.

Formally, let $\mathcal{F}=\mathcal{H}$ and $\mcA=\mathcal{H}$, where $\mathcal{H}$ is the RKHS with kernel $k$. In what follows, we denote the usual empirical kernel matrix by $K^{(1)}\in\R^{n\times n}$, with entries given by $K^{(1)}_{ij}=k(X_i,X_j)$, and the usual evaluation vector by $K_{xX}^{(1)}\in\R^{1\times n}$, with  entries given by $[K_{xX}^{(1)}]_{j}=k(x,X_j)$. To express our estimator, we introduce additional kernel matrices $K^{(2)},K^{(3)},K^{(4)}\in\R^{n\times n}$ and additional evaluation vectors $K^{(2)}_{xX},K^{(3)}_{xX},K^{(4)}_{xX}\in\mathbb{R}^{1\times n}$. For readability, we reserve details on how to compute these additional matrices and vectors for Appendix~\ref{sec:compute_proof}. At a high level, these additional objects apply the functional $\theta:g\mapsto \mathbb{E}[m(g;Z)]$ to the kernel $k$ and data $(X_i)_{i=1}^n$ in various ways. For any symmetric matrix $A$, let $A^{-}$ denote its pseudo-inverse. If $A$ is invertible then $A^{-}=A^{-1}$.

\begin{proposition}[Closed form of maximizer]\label{prop:closed1}
For a potential minimizer $a$, the adversarial maximizer $\hat{f}_a$ has a closed form solution with a coefficient vector $\hat{\gamma}_a\in \R^{2n}$. More formally,
$\hat{f}_a(x)=\begin{bmatrix}K^{(1)}_{xX} & K^{(2)}_{xX} \end{bmatrix}\hat{\gamma}_a$ and  $m(x,\hat{f}_a)=\begin{bmatrix}K^{(3)}_{xX} & K^{(4)}_{xX} \end{bmatrix} \hat{\gamma}_a$.
The coefficent vector is explicitly given by $\hat{\gamma}_a=\frac{1}{2}\Delta^{-}\left[V -U \mathbf{a}\right]$ where
\begin{align*}
U :=~& \begin{bmatrix}K^{(1)} \\ K^{(3)} \end{bmatrix} \in \mathbb{R}^{2n \times n}, &
\Delta:=~& U U^{\top} + n\lambda \begin{bmatrix} K^{(1)} & K^{(2)} \\ K^{(3)}  & K^{(4)} \end{bmatrix} \in\mathbb{R}^{2n\times 2n}, &
V := \begin{bmatrix}K^{(2)} \\ K^{(4)} \end{bmatrix} \mathbf{1}_{n} \in \mathbb{R}^{2n};
\end{align*}
$\mathbf{a}\in\mathbb{R}^n$ is defined such that $\mathbf{a}_i=a(x_i)$, and $\mathbf{1}_{n}\in\R^n$ is the vector of ones. 
\end{proposition}

\begin{proposition}[Closed form of minimizer]\label{prop:closed2}
The minimizer $\hat{a}$ has a closed form solution with a coefficient vector $\hat{\beta}\in\R^n$. More formally, $\hat{a}(x)=K^{(1)}_{xX} \hat{\beta}$. The coefficient vector is explicity given by $\hat{\beta}=\left\{A^{\top}  \Delta^{-} A +4n\mu\cdot K^{(1)}\right\}^-A^{\top}\Delta^{-}V$ where $A:= U K^{(1)}$.
\end{proposition}

    Combining Propositions~\ref{prop:closed1} and~\ref{prop:closed2}, it is possible to compute $\hat{\mathbf{a}}$, and hence $\hat{f}_{\hat{a}}(x)$ and $m(x,\hat{f}_{\hat{a}})$. Therefore it is possible to compute the optimized loss in Estimator~\ref{algo:reg-estimator}. While Theorem~\ref{thm:reg-main-error} provides theoretical guidance on how to choose $(\lambda,\mu)$, the optimized loss provides a practical way to choose $(\lambda,\mu)$.

    Kernel balancing weights may be viewed as Riesz representer estimators for the average treatment effect functional $\theta:g\mapsto \mathbb{E}[g(1,W)-g(0,W)]$ evaluated on training data, where $X=(D,W)$, $D$ is the treatment, and $W$ is the covariate. Various estimators have been proposed, e.g. \cite{wong2018kernel,zhao2019covariate,kallus2020generalized,hirshberg2019kernel}, which are vectors in $\mathbb{R}^n$. 
    Our results situate kernel balancing weights within a unified framework for semiparametric inference across general function spaces. The loss of Estimator~\ref{algo:reg-estimator}, the closed form of Proposition~\ref{prop:closed2}, and the norm of our guarantee in Theorem~\ref{thm:reg-main-error} depart from and complement these works.  

    The closed form expressions above involve inverting kernel matrices that scale with $n$. To reduce the computational burden, we derive a Nystr\"om approximation in Appendix~\ref{sec:compute_proof}.

\section{Simulated and real data analysis}\label{sec:sim}

\textbf{Adversarial estimation with general function spaces may improve coverage.} We demonstrate how our proposal, Estimator~\ref{algo:reg-estimator}, compares favorably with alternative estimators in an average treatment effect (ATE) coverage simulation.
Let $X=(D,W)$, where $D$ is the treatment and $W$ are the covariates.
In highly nonlinear simulations with $n=1000$ and $dim(W)=10$, our estimators may achieve nominal coverage where some previous methods break down. In high dimensional simulations with $n=100$ and $dim(W)=100$, our estimators may have lower bias and shorter confidence intervals. These gains seem to accrue from, simultaneously, using flexible function spaces for Riesz estimation and directly estimating the Riesz representer. These simulation designs are challenging, and only particular variations of our estimator work well. Appendix~\ref{sec:sim_detail} gives details.

We implement five variations of Estimator~\ref{algo:reg-estimator}, denoted by $\hat{a}$ in Section~\ref{sec:algo}. The five variations of $\hat{a}$ are (i) sparse linear, (ii) RKHS, (iii) RKHS with Nystr\"om approximation, (iv) random forest, or (v) neural network. Note that (i) echoes lasso Riesz representers, and (ii-iii) extend kernel balancing weights to allow for sample splitting. However, (iv) and (v) are altogether new function spaces. For comparison, we implement two propensity score estimators for the Riesz representer: (vi) logistic and (vii) random forest.

We use these Riesz estimators $\hat{a}$ together with a  boosted regression estimator $\hat{g}$ within the ATE Estimators~\ref{algo:dml} and~\ref{algo:debias-nocross}, denoted by $\check{\theta}$ and $\hat{\theta}$ in Section~\ref{sec:inference}. The former uses sample splitting while the latter does not. We report the coverage, bias, and interval length.
We leave table entries blank when previous work does not provide theoretical justification.

\begin{table}[h]
    \centering
    \begin{subtable}{0.45\textwidth}
        \centering
        \scalebox{0.7}{
    \begin{tabular}{|c|cc|ccccc|}
    \hline 
    Est. & \multicolumn{2}{|c|}{Prop. score} &  \multicolumn{5}{|c|}{Adversarial} \\
 Space   & logistic & R.F. & sparse & \underline{RKHS} & Nystr\"om & \underline{R.F.} & N.N. \\  
    \hline 
        coverage & 83 & 76 &  74 & 95 & 83 & 91 & 69  \\
        bias & -12 & -1 &  -7 & -4 & -3 & 0 & -8  \\
        length & 54 & 29 &  33 & 53 & 32 & 39 & 35 \\
        \hline 
    \end{tabular}
    }
    \caption{Sample splitting}
    \label{tab:nonlinear}
        \end{subtable}
    \hfill 
      \begin{subtable}{0.45\textwidth}
        \centering
         \scalebox{0.7}{
    \begin{tabular}{|c|cc|ccccc|}
    \hline 
    Est. & \multicolumn{2}{|c|}{Prop. score} & \multicolumn{5}{|c|}{Adversarial} \\
 Space   & logistic & R.F. &  sparse & \underline{RKHS} & Nystr\"om & \underline{R.F.} & N.N. \\  
    \hline 
        coverage & 79 & - &  73 & 88 & 79 & 90 & 72  \\
        bias & -11 & - &  -7 & -4 & -4 & -2 & -8  \\
        length & 47 & - & 30 & 36 & 30 & 37 & 33 \\
        \hline 
    \end{tabular}
    }
    \caption{No sample splitting}
    \label{tab:nonlinear_no}
    \end{subtable}
\vspace{-10pt}
    \caption{Nonlinear design $\{n=1000,dim(W)=10\}$. Values are multiplied by $10^2$.}
    \vspace{-10pt}
\end{table}

Tables~\ref{tab:nonlinear} and~\ref{tab:nonlinear_no} present results from 100 simulations of the highly nonlinear design, with and without sample splitting, respectively. We find that our adversarial estimators (ii) and (iv) achieve near nominal coverage with sample splitting (95\% and 91\%) and slightly undercover without sample splitting (88\% and 90\%). For readability, we underline these estimators. None of the propensity score methods (vi-vii) achieve nominal coverage, with or without sample splitting, and they undercover more severely (at best 83\% and 76\%).

\begin{table}
    \centering
    \begin{subtable}{0.45\textwidth}
        \centering
        \scalebox{0.7}{
    \begin{tabular}{|c|cc|ccccc|}
    \hline 
    Est. & \multicolumn{2}{|c|}{Prop. score} &  \multicolumn{5}{|c|}{Adversarial} \\
 Space   & logistic & R.F. & \underline{sparse} & RKHS & Nystr\"om & R.F. & N.N. \\  
   \hline 
        coverage & 92 & 91 &  93 & 89 & 88 & 84 & 3  \\
        bias & 44 & 16 &  8 & 6 & 6 & -7 & -35  \\
        length & 214 & 93 &  88 & 72 & 73 & 126 & 6 \\
        \hline 
    \end{tabular}
    }
    \caption{Sample splitting}
    \label{tab:high}
        \end{subtable}
    \hfill 
      \begin{subtable}{0.45\textwidth}
        \centering
         \scalebox{0.7}{
    \begin{tabular}{|c|cc|ccccc|}
    \hline 
    Est. & \multicolumn{2}{|c|}{Prop. score} & \multicolumn{5}{|c|}{Adversarial} \\
 Space   & logistic & R.F. &  \underline{sparse} & RKHS & Nystr\"om & R.F. & N.N. \\  
     \hline 
        coverage & 79 & - &  88 & 88 & 87 & 91 & 3  \\
        bias & 2 & - &  1 & 8 & 9 & 26 & -33  \\
        length & 64 & - &  75 & 76 & 78 & 184 & 3 \\
        \hline 
    \end{tabular}
    }
    \caption{No sample splitting}
    \label{tab:high_no}
    \end{subtable}
    \vspace{-10pt}
    \caption{High dimensional design $\{n=100,dim(W)=100\}$. Values are multiplied by $10^2$.}
    \vspace{-10pt}
\end{table}

Tables~\ref{tab:high} and~\ref{tab:high_no} present analogous results from the high dimensional design. We find that our adversarial estimator (i) achieves close to nominal coverage with sample splitting (93\%) and slightly undercovers without sample splitting (88\%). For readability, we underline this estimator. The propensity score estimators (vi-vii) also achieve close to nominal coverage (92\%, 91\%). Comparing the instances of near nominal coverage, we see that our estimator (i) has half of the bias (0.08 versus 0.44 and 0.16) and shorter confidence intervals (0.88 versus 2.14 and 0.93) compared to the propensity score estimators (vi-vii).

Appendix~\ref{sec:sim_detail} presents additional simulation results for a more basic design with $n\in\{100,200,500,1000,2000\}$ and $dim(W)=10$. We find that every variation of our estimator achieves nominal coverage, with or without sample splitting. Moreover, in a basic setting, eliminating sample splitting may improve precision---a simple point with practical consequences for applied statistics, which underscores the importance of Theorems~\ref{thm:debias-nocross-stability} and~\ref{thm:debias-nocross}.

In summary, in simple designs, every variation of our estimator works well, and variations without sample splitting have better precision; in the highly nonlinear design, some variations with sample splitting achieve nominal coverage where some previous methods break down; in the high dimensional design, a variation with sample splitting achieves nominal coverage with smaller bias and length than some previous methods. 

The different simulation designs illustrate two virtues of our framework. First, our main result in Section~\ref{sec:rate} applies to variations (i-v) of Estimator~\ref{algo:reg-estimator} in a unified manner. Second, Estimator~\ref{algo:reg-estimator} is compatible with Estimators~\ref{algo:tml},~\ref{algo:dml}, and~\ref{algo:debias-nocross}, which have different finite sample performance in different settings; sample splitting may help or hinder performance.

\textbf{Heterogeneous effects by political environment.}
Finally, we extend the highly influential analysis of \cite{karlan2007does} from parametric estimation to semiparametric estimation. To our knowledge, semiparametric estimation has not been used in this setting before; we provide an empirical contribution. We implement our estimators (i-v) and propensity score estimators (vi-vii) on real data. Compared to the parametric results of \cite{karlan2007does}, which do not allow general nonlinearities, our flexible approach relaxes functional form restrictions and improves precision. Compared to semiparametric results obtained with some previous methods, our approach may improve precision by allowing for several machine learning function classes and reducing approximation error.

\begin{figure}[h]
\begin{centering}
     \begin{subfigure}[b]{0.48\textwidth}
         \centering
         \includegraphics[width=0.9\textwidth]{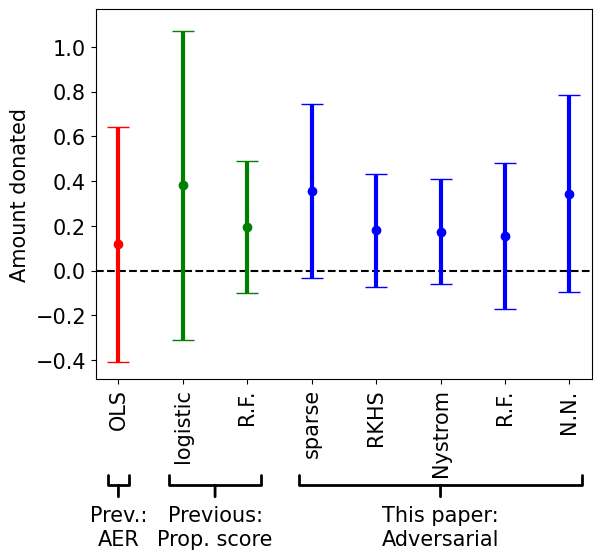}
         \caption{Sample splitting 
         }
     \end{subfigure}
     \hfill
     \begin{subfigure}[b]{0.48\textwidth}
         \centering
         \includegraphics[width=0.9\textwidth]{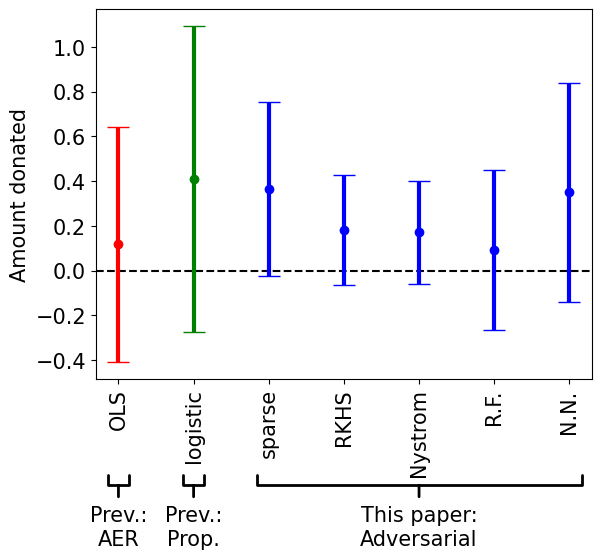}
         \caption{No sample splitting 
         }
     \end{subfigure}
\par
\vspace{-10pt}
\caption{\label{fig:dollars}
{
{Heterogeneous effects on dollars donated \{$n=25859$, $dim(W)=15$\}}. 
}
}
\vspace{-10pt}
\end{centering}
\end{figure}

\begin{figure}[h]
\begin{centering}
     \begin{subfigure}[b]{0.48\textwidth}
         \centering
         \includegraphics[width=0.9\textwidth]{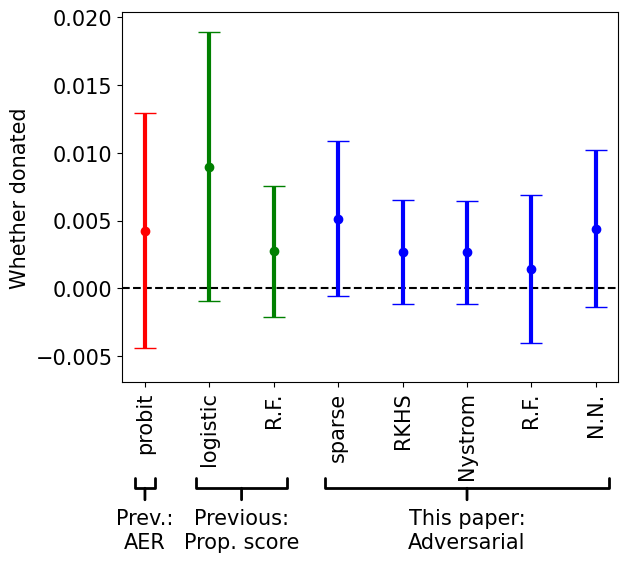}
         \caption{Sample splitting 
         }
     \end{subfigure}
     \hfill
     \begin{subfigure}[b]{0.48\textwidth}
         \centering
         \includegraphics[width=0.9\textwidth]{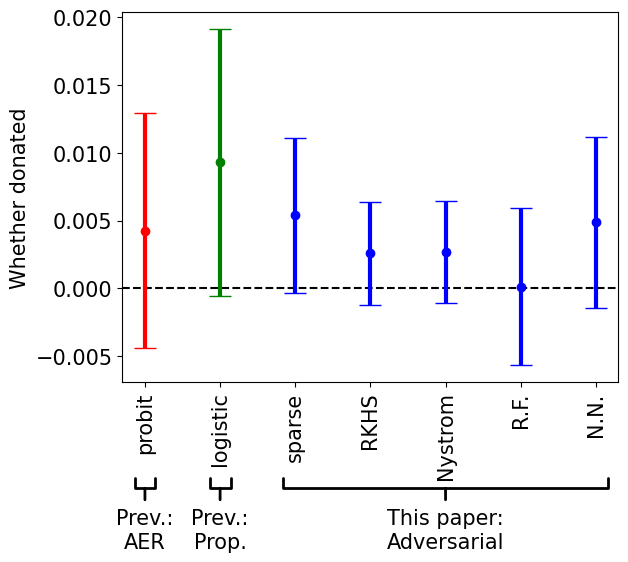}
         \caption{No sample splitting 
         }
     \end{subfigure}
\par
\vspace{-10pt}
\caption{\label{fig:whether}
{
{Heterogeneous effects on whether donated \{$n=25859$, $dim(W)=15$\}}. 
}
}
\vspace{-10pt}
\end{centering}
\end{figure}

We study the heterogeneous effects, by political environment, of a matching grant on charitable giving in a large scale natural field experiment. We follow the variable definitions of \cite[Table 6]{karlan2007does}. In Figure~\ref{fig:dollars}, the outcome $Y\in\mathbb{R}$ is dollars donated, while in Figure~\ref{fig:whether}, $Y\in\{0,1\}$ indicates whether the household donated. The treatment $D\in\{0,1\}$ indicates whether the household received a 1:1 matching grant as part of a direct mail solicitation. The covariates include political environment, previous contributions, and demographics. Altogether, $dim(W)=15$ and $n=25859$ in our sample; see Appendix~\ref{sec:sim_detail}.

A central finding of \cite{karlan2007does} is that ``the matching grant treatment was ineffective in [Democratic] states, yet quite effective in [Republican] states. The nonlinearity is striking...'', which motivates us to formalize a semiparametric estimand. The authors arrive at this conclusion with nonlinear but parametric estimation, focusing on the coefficient of the interaction between the binary treatment and a binary covariate $W_1$ indicating whether the household is in a state that voted for George W. Bush in the 2004 presidential election.

We generalize the interaction coefficient. Denote the regression $g_0\{D,W_1,W_{2:dim(W)}\}=\mathbb{E}\{Y|D,W_1,W_{2:dim(W)}\}$. We consider the parameter $\theta(g_0)$ where $\theta(g)=\mathbb{E}([g\{1,1,W_{2:dim(W)}\}-g\{0,1,W_{2:dim(W)}\}]-[g\{1,0,W_{2:dim(W)}\}-g\{0,0,W_{2:dim(W)}\}])$. Intuitively, this parameter asks: how much more effective was the matching grant in Republican states compared to Democratic states? Its Riesz representer $a_0\{D,W_1,W_{2:dim(W)}\}$ involves products and differences of the inverse propensity scores $[\mathbb{P}\{D=1|W_1,W_{2:dim(W)}\}]^{-1}$ and  $[\mathbb{P}\{W_1|W_{2:dim(W)}\}]^{-1}$, motivating our direct adversarial estimation approach. See Appendix~\ref{sec:sim_detail} for derivations.

Figures~\ref{fig:dollars} and~\ref{fig:whether} visualize point estimates and 95\% confidence intervals for this semiparametric estimand. The former considers effects on dollars donated while the latter considers effects on whether households donated at all. Each figure presents results with and without sample splitting, using propensity score estimators (vi-vii) in green and our proposed adversarial estimators (i-v) in blue. As before, we leave entries blank when previous work does not provide theoretical justification. For comparison, we also present the parametric point estimate and confidence interval of \cite[Table 6, Column 9]{karlan2007does} in red.\footnote{The estimates in red slightly differ from \cite{karlan2007does} since we drop observations receiving 2:1 and 3:1 matching grants and report the functional rather than the probit coefficient; see Appendix~\ref{sec:sim_detail}.}

Our estimates are stable across different estimation procedures and show that heterogeneity in the effect is positive, confirming earlier findings. The results are qualitatively consistent with \cite{karlan2007does}, who use a simpler parametric specification motivated by economic reasoning. Compared to the original parametric findings in red, which are not statistically significant, our findings in blue relax parametric assumptions and improve precision: our preferred confidence interval is more than 50\% shorter in Figures~\ref{fig:dollars} and~\ref{fig:whether}. Compared to the semiparametric methods in green, our preferred confidence interval is more than 20\% shorter in Figures~\ref{fig:dollars} and~\ref{fig:whether}. Our approach appears to improve model fit relative to some previous parametric and semiparametric approaches, and has a unified justification. 
\section{Discussion}\label{sec:conclusion}

This paper uses critical radius theory over general function spaces to provide unified $L_2$ rates for nonparametric estimation of Riesz representers. The results in this paper depart from previous work by allowing for approximation of the Riesz representer by neural networks and random forests. Our results are compatible with targeted and debiased machine learning inference results with sample splitting that allow for mis-specification, as well as ``stability-rate'' and ``complexity-rate''  inference results without sample splitting. Simulations demonstrate that our flexible method may achieve nominal coverage when less flexible methods break down. Our method provides rigorous empirical evidence on heterogeneous effects of matching grants by political environment.


\bibliographystyle{hapalike_mod}

\spacingset{1}
{

}
\spacingset{1.5}


\appendix

\section{Proof of main results}\label{sec:main}

\subsection{Adversarial estimation and fast rate}

For convenience, throughout this appendix we  use the notation
\begin{align*}
    \Psi(a, f) =~& \E[m(Z; f) - a(X)\cdot f(X)],\quad 
    \Psi_n(a, f) =~ \frac{1}{n}\sum_{i=1}^{n} \left(m(Z_i; f) - a(X_i)\cdot f(X_i)\right)\\
    \Psi_n^\lambda(a, f) =~& \Psi_n(a, f) - \|f\|_{2,n}^2 - \lambda \|f\|_{\mcA}^2,\quad 
    \Psi^{\lambda}(a, f) =~ \Psi(a, f) - \frac{1}{4} \|f\|_2^2 - \frac{\lambda}{2}\|f\|_{\mcA}^2.
\end{align*}
Thus
$
    \hat{a} := \argmin_{a\in \mcA} \sup_{f\in \mcF} \Psi_n^{\lambda}(a, f) + \mu \|a\|_{\mcA}^2.
$
We also lighten notation, setting $\delta=\bar{\delta}$.

\begin{proof}[Proof of Theorem~\ref{thm:reg-main-error}]
    We proceed in steps.
    \begin{enumerate}
        \item Relating empirical and population regularization.
By \cite[Theorem 14.1]{wainwright2019high}, with probability $1-\zeta$,
$
    \forall f\in \mcF_{B}: \left|\|f\|_{n,2}^2 - \|f\|_2^2 \right| \leq \frac{1}{2} \|f\|_2^2 + \delta^2 
$
for our choice of $\delta:=\delta_n + c_0 \sqrt{\frac{\log(c_1/\zeta)}{n}}$, where $\delta_n$ upper bounds the critical radius of $\mcF_{B}$ and $c_0, c_1$ are universal constants. Moreover, for any $f$, with $\|f\|_{\mcA}^2\geq B$, we can consider the function $f \sqrt{B}/\|f\|_{\mcA}$, which also belongs to $\mcF_{B}$, since $\mcF$ is star-convex. Thus we can apply the above lemma to this re-scaled function and multiply both sides by $\|f\|_{\mcA}^2/B$:
\begin{equation*}
    \forall f\in \mcF \text{ such that } \|f\|_{\mcA}^2 \geq B: \left|\|f\|_{n,2}^2 - \|f\|_2^2 \right| \leq \frac{1}{2} \|f\|_2^2 + \delta^2\frac{\|f\|_{\mcA}^2}{B}.
\end{equation*}
Thus overall, we have
\begin{equation}\label{eqn:reg-pop-emp}
    \forall f\in \mcF: \left|\|f\|_{n,2}^2 - \|f\|_2^2 \right| \leq \frac{1}{2} \|f\|_2^2 + \delta^2 \max\left\{1, \frac{\|f\|_{\mcA}^2}{B}\right\}.
\end{equation}
Hence with probability $1-\zeta$,
\begin{align*}
    \forall f\in \mcF: \lambda \|f\|_{\mcA}^2 + \|f\|_{2,n}^2 \geq~& \lambda \|f\|_{\mcA}^2 + \frac{1}{2}\|f\|_{2}^2 - \delta^2 \max\left\{1, \frac{\|f\|_{\mcA}^2}{B}\right\}\\
    \geq~& \left(\lambda - \frac{\delta^2}{B}\right)\|f\|_{\mcA}^2 + \frac{1}{2}\|f\|_{2}^2 - \delta^2.
\end{align*}
Assuming that $\lambda \geq \frac{2\delta^2}{B}$, the latter is at least
\begin{equation*}
    \forall f\in \mcF: \lambda \|f\|_{\mcA}^2 + \|f\|_{2,n}^2 \geq \frac{\lambda}{2} \|f\|_{\mcA}^2 + \frac{1}{2}\|f\|_{2}^2 - \delta^2.
\end{equation*}

        \item Upper bounding centered empirical sup-loss.
        We now argue that
$$\sup_{f \in \mcF} (\Psi_n(\hat{a}, f) - \Psi_n(a_*, f)) = \sup_{f \in \mcF} \E_n[(a_*(X)-\hat{a}(X))\, f(X)]$$ 
is small. By the definition of $\hat{a}$,
\begin{equation}\label{eqn:reg-optimality}
\sup_{f\in \mcF} \Psi_n^{\lambda}(\hat{a}, f) \leq \sup_{f\in \mcF} \Psi_n^{\lambda}(a_*, f) + \mu \left(\|a_*\|_{\mcA}^2 - \|\hat{a}\|_{\mcA}^2\right).
\end{equation}

By \cite[Lemma~14]{foster2019orthogonal}, the fact that $m(Z; f) - a_*(X) f(X)$ is $(1+b)$-Lipschitz with respect to the vector $(m(Z;f), f(z))$ (since $a_*(X)\in [-b, b]$) and by our choice of $\delta:=\delta_n + c_0 \sqrt{\frac{\log(c_1/\zeta)}{n}}$, where $\delta_n$ is an upper bound on the critical radii of $\mcF_{B}$ and $m \circ \mcF_B$, with probability $1-\zeta$
\begin{align*}
 \forall f\in \mcF_{B}: \left|\Psi_n(a_*, f) - \Psi(a_*, f)\right| &=O\left(\delta \left(\|f\|_2 + \sqrt{\E[m(Z;f)^2]}\right) + \delta^2\right) \\
 &=O\left(\delta\, M\, \|f\|_2 + \delta^2\right).
\end{align*}
We have invoked Assumption~\ref{ass:strong-smooth}. Thus, if $\|f\|_{\mcA} \geq \sqrt{B}$, we can apply the latter inequality for the function $f \sqrt{B}/\|f\|_{\mcA}$, which falls in $\mcF_{B}$, and then multiply both sides by $\|f\|_{\mcA}/\sqrt{B}$ (invoking the linearity of the operator $\Psi_n(a, f)$ with respect to $f$):
\begin{align}\label{eqn:reg-concentration}
    \forall f \in \mcF: \left|\Psi_n(a_*, f) - \Psi(a_*, f)\right| = O\left(\delta\, M\, \|f\|_2 + \delta^2 \max\left\{1, \frac{\|f\|_{\mcA}}{\sqrt{B}}\right\}\right).
\end{align}

By~\eqref{eqn:reg-optimality} and \eqref{eqn:reg-concentration}, we have that with probability $1-2\zeta$, for some universal constant $C$,
\begin{align*}
   &\sup_{f\in \mcF} \Psi_n^{\lambda}(a_*, f) 
    = \sup_{f\in \mcF} \left(\Psi_n(a_*, f) - \|f\|_{2,n}^2 - \lambda \|f\|_{\mcA}^2\right)\\
    \leq~& \sup_{f\in \mcF} \left(\Psi(a_*, f) + C \delta^2 + \frac{C\delta^2}{\sqrt{B}}\|f\|_{\mcA} + C M \delta \|f\|_2 - \|f\|_{2,n}^2 - \lambda \|f\|_{\mcA}^2\right)\\
    \leq~& \sup_{f\in \mcF} \left(\Psi(a_*, f) + C \delta^2 + \frac{C\delta^2}{\sqrt{B}}\|f\|_{\mcA}  + C M\delta \|f\|_2 - \frac{1}{2} \|f\|_2^2 - \frac{\lambda}{2} \|f\|_{\mcA}^2 + \delta^2\right)\\
    \leq~& \sup_{f\in \mcF} \Psi^{\lambda/2}(a_*, f) + O\left(\delta^2\right)
    + \sup_{f\in \mcF} \left( \frac{C\delta^2}{\sqrt{B}}\|f\|_{\mcA} - \frac{\lambda}{4}\|f\|_{\mcA}^2\right) + \sup_{f\in \mcF} \left(C M\delta \|f\|_2 - \frac{1}{4}\|f\|_{2}^2\right).
\end{align*}
Moreover, observe that for any norm $\|\cdot\|$ and any constants $a,b>0$,
$
    \sup_{f\in \mcF} \left(a\|f\| - b \|f\|^2\right) \leq \frac{a^2}{4b}.
$
Thus if we assume that $\lambda\geq 2\delta^2/B$, we have
\begin{align*}
\sup_{f\in \mcF} \left( \frac{C\delta^2}{\sqrt{B}}\|f\|_{\mcA} - \frac{\lambda}{4}\|f\|_{\mcA}^2\right) \leq~& \frac{C^2 \delta^4}{B \lambda} \leq \frac{C^2}{2} \delta^2,\quad 
    \sup_{f\in \mcF} \left(C M\delta \|f\|_2 - \frac{1}{4}\|f\|_{2}^2\right) \leq~& C^2 M^2 \delta^2.
\end{align*}
Hence we have
$
    \sup_{f\in \mcF} \Psi_n^{\lambda}(a_*, f)\leq \sup_{f\in \mcF} \Psi^{\lambda/2}(a_*, f) + O\left(M^2\, \delta^2\right).
$
Moreover
\begin{align*}
\sup_{f\in \mcF} \Psi_n^{\lambda}(\hat{a}, f) =~& \sup_{f \in \mcF} \left(\Psi_n(\hat{a}, f) - \Psi_n(a_*, f) + \Psi_n(a_*, f) - \|f\|_{2,n}^2 - \lambda \|f\|_{\mcA}^2\right)\\ 
\geq~& \sup_{f \in \mcF} \left(\Psi_n(\hat{a}, f) - \Psi_n(a_*, f) - 2 \|f\|_{2,n}^2 - 2\, \lambda \|f\|_{\mcA}^2\right)  \\
&\quad + \inf_{f\in \mcF} \left(\Psi_n(a_*, f) + \|f\|_{2,n}^2 + \lambda \|f\|_{\mcA}^2\right).
\end{align*}
Since $\Psi_n(a, f)$ is a linear operator of $f$ and $\mcF$ is a symmetric class, we have
\begin{align*}
    &\inf_{f\in \mcF} \left(\Psi_n(a_*, f) + \|f\|_{2,n}^2 + \lambda \|f\|_{\mcA}^2\right) =~ \inf_{f\in \mcF} \left(\Psi_n(a_*, -f) + \|f\|_{2,n}^2 + \lambda \|f\|_{\mcA}^2\right)\\
    &=~ \inf_{f\in \mcF} \left(- \Psi_n(a_*, f) + \|f\|_{2,n}^2 + \lambda \|f\|_{\mcA}^2\right)
    =~ - \sup_{f\in \mcF} \left(\Psi_n(a_*, f) - \|f\|_{2,n}^2 - \lambda \|f\|_{\mcA}^2\right) \\
    &=~ -\sup_{f\in \mcF} \Psi_n^{\lambda}(a_*, f).
\end{align*}
Combining this expression with~\eqref{eqn:reg-optimality} yields
\begin{align*}
    &\sup_{f \in \mcF} \left(\Psi_n(\hat{a}, f) - \Psi_n(a_*, f) - \|f\|_{2,n}^2 - \lambda \|f\|_{\mcA}^2\right) 
    \leq 2\, \sup_{f\in \mcF} \Psi_n^{\lambda}(a_*, f) + \mu \left(\|a_*\|_{\mcA}^2 - \|\hat{a}\|_{\mcA}^2\right)\\
    \leq~& 2\,\sup_{f\in \mcF} \Psi^{\lambda/2}(a_*, f) + \mu \left(\|a_*\|_{\mcA}^2 - \|\hat{a}\|_{\mcA}^2\right) + O\left(M^2\, \delta^2\right).
\end{align*}

        \item Lower bounding centered empirical sup-loss.
First observe that
$
    \Psi_n(a, f) - \Psi_n(a_*, f) = \E_n[(a_*(X)-a(X)) f(X)].
$
Let $\Delta=a_* - \hat{a}$. Suppose that $\|\Delta\|_2\geq \delta$ and let $r = \frac{\delta}{2\|\Delta\|_2}\in [0, 1/2]$. Since $\Delta\in \mcF$ and $\mcF$ is star-convex, we also have that $r \Delta\in \mcF$. Thus
\begin{align*}
    &\sup_{f \in \mcF} \left(\Psi_n(\hat{a}, f) - \Psi_n(a_*, f) - \|f\|_{2,n}^2 - \lambda \|f\|_{\mcA}^2\right) \\
    \geq~& 
    \Psi_n(\hat{a}, r\Delta) - \Psi_n(a_*, r\Delta) - r^2 \|\Delta\|_{2,n}^2 - \lambda r^2 \|\Delta\|_{\mcA}^2\\
    =~& r \E_n\left[(a_*(X) - \hat{a}(X))^2\right] - r^2 \|\Delta\|_{2,n}^2 - \lambda r^2 \|\Delta\|_{\mcA}^2\\
    =~& r \|\Delta\|_{2,n}^2 - r^2 \|\Delta\|_{2,n}^2 - \lambda r^2 \|\Delta\|_{\mcA}^2
    \geq~ r \|\Delta\|_{2,n}^2 - r^2 \|\Delta\|_{2,n}^2 - \lambda \|\Delta\|_{\mcA}^2.
\end{align*}

Moreover, since $\delta_n$ upper bounds the critical radius of $\mcF_{B}$ and by~\eqref{eqn:reg-pop-emp}
\begin{align*}
    r^2 \|\Delta\|_{2,n}^2
    \leq~& r^2 \left(2\|\Delta\|_2^2 + \delta^2 + \delta^2 \frac{\|\Delta\|_{\mcA}^2}{B}\right)
    \leq 2\delta^2 + \delta^2 \frac{\|\Delta\|_{\mcA}^2}{B} \leq 2\delta^2 + \lambda \|\Delta\|_{\mcA}^2.
\end{align*}
Thus
\begin{align*}
    \sup_{f \in \mcF} \left(\Psi_n(\hat{a}, f) - \Psi_n(a_*, f) - \|f\|_{2,n}^2 - \lambda \|f\|_{\mcA}^2\right) \geq~& r \|\Delta\|_{2,n}^2 - 2\delta^2 - 2\lambda \|\Delta\|_{\mcA}^2.
\end{align*}
Furthermore, since $\delta_n$ upper bounds the critical radius of $\mcF_{B}$ and by~\eqref{eqn:reg-pop-emp},
\begin{align*}
    \|\Delta\|_{2,n}^2 \geq \frac{1}{2} \|\Delta\|_2^2 - \frac{\delta^2}{2B}\|\Delta\|_{\mcA}^2 - \delta^2 \geq \frac{1}{2} \|\Delta\|_2^2 - \lambda\|\Delta\|_{\mcA}^2 - \delta^2.
\end{align*}
Thus we have
\begin{align*}
    \sup_{f \in \mcF} \left(\Psi_n(\hat{a}, f) - \Psi_n(a_*, f) - \|f\|_{2,n}^2 - \lambda \|f\|_{\mcA}^2\right) \geq~& \frac{r}{2} \|\Delta\|_{2}^2 - 3\delta^2 - 3\lambda \|\Delta\|_{\mcA}^2\\
    \geq~& \frac{\delta}{4} \|\Delta\|_2  - 3\delta^2 - 3\lambda \|\Delta\|_{\mcA}^2.
\end{align*}

        \item Combining upper and lower bound.
Combining the upper and lower bound on the centered population sup-loss we get that with probability $1-3\zeta$, either $\|\Delta\|_2\leq \delta$ or
\begin{align*}
    \frac{\delta}{4} \|\Delta\|_2 \leq~& O\left(M^2\, \delta^2\right) + 2\,\sup_{f\in \mcF} \Psi^{\lambda/2}(a_*, f) + 3\lambda \|\Delta\|_{\mcA}^2 + \mu \left(\|a_*\|_{\mcA}^2 - \|\hat{a}\|_{\mcA}^2\right).
\end{align*}
We now control the last part. Since $\mu \geq 6\lambda$,
\begin{align*}
    3\lambda \|\Delta\|_{\mcA}^2 + \mu \left(\|a_*\|_{\mcA}^2 - \|\hat{a}\|_{\mcA}^2\right) \leq~&  6\lambda \left(\|a_*\|_{\mcA}^2 + \|\hat{a}\|_{\mcA}^2\right) + \mu \left(\|a_*\|_{\mcA}^2 - \|\hat{a}\|_{\mcA}^2\right) \leq 2 \mu \|a_*\|_{\mcA}^2.
\end{align*}
We then conclude that
\begin{equation*}
    \frac{\delta}{4} \|\Delta\|_2 = O\left(M^2\, \delta^2\right) + 2\,\sup_{f\in \mcF} \Psi^{\lambda/2}(a_*, f) + 2 \mu \|a_*\|_{\mcA}^2.
\end{equation*}
Dividing by $\delta/4$ gives
\begin{equation*}
     \|\Delta\|_2 = O\left(M^2\, \delta\right) + \frac{8}{\delta}\,\sup_{f\in \mcF} \Psi^{\lambda/2}(a_*, f) + 8 \frac{\mu}{\delta} \|a_*\|_{\mcA}^2.
\end{equation*}
Thus either $\|\Delta\|_2\leq \delta$ or the latter inequality holds. In both cases, the latter holds.
        
        \item Upper bounding population sup-loss at minimum.
      By Riesz representation,
\begin{align*}
    \sup_{f\in \mcF} \Psi^{\lambda/2}(a_*, f) =~& \sup_{f\in \mcF} \E[ (a_0(X)-a_*(X))\, f(z) ] - \frac{1}{4} \|f\|_2^2 - \frac{\lambda}{4}\|f\|_{\mcA}^2\\
    \leq~& \sup_{f\in \mcF} \E[ (a_0(X) - a_*(X))\, f(z) ] - \frac{1}{4} \|f\|_{2}^2
    = \|a_0 - a_*\|_2^2.
\end{align*}

        \item Concluding.
In summary,
\begin{equation*}
    \|\hat{a}-a_*\|_2 = O\left(M^2\, \delta\right) + \frac{8}{\delta}\, \|a_*-a_0\|_2^2 + 8 \frac{\mu}{\delta} \|a_*\|_{\mcA}^2.
\end{equation*}
By the triangle inequality,
\begin{equation*}
    \|\hat{a}-a_0\|_2 = O\left(M^2\, \delta\right) + \frac{8}{\delta}\, \|a_*-a_0\|_2^2 + \|a_* - a_0\|_2 + 8 \frac{\mu}{\delta} \|a_*\|_{\mcA}^2.
\end{equation*}

Choosing $a_* = \argmin_{a\in \mcA} \|a-a_0\|_2$ and using the fact that $\delta \geq \epsilon_n$ gives
\begin{equation*}
    \|\hat{a}-a_0\|_2 = O\left(M^2 \delta + \|a_*-a_0\|_2 + \frac{\mu}{\delta} \|a_*\|_{\mcA}^2\right) = O\left(M^2 \delta + \frac{\mu}{\delta} \|a_*\|_{\mcA}^2\right). \qedhere
\end{equation*}
    \end{enumerate}
\end{proof}

\subsection{Semiparametric inference}

\begin{proof}[Proof of Theorem~\ref{thm:debias-nocross}]
    We proceed in steps.

    \begin{enumerate}
        \item Decomposition. Observe that $\theta_0 = \E[m_{a}(Z; g_0)]$ for all $a$. Moreover,
\begin{align*}
    \hat{\theta}-\theta_0 =~& \E_n[m_{\hat{a}}(Z; \hat{g})] - \E[m_{\hat{a}}(Z; \hat{g})] + \E[m_{\hat{a}}(Z; \hat{g})] - \E[m_{\hat{a}}(Z; g_0)]\\
    =~& \E_n[m_{\hat{a}}(Z; \hat{g})] - \E[m_{\hat{a}}(Z; \hat{g})]  + \E[(a_0(X) - \hat{a}(X))\, (\hat{g}(X) - g_0(X))].
\end{align*}
By Assumption~\ref{ass:main-cond}
we have that
\begin{align*}
    \sqrt{n}\left(\hat{\theta}-\theta_0\right) =~& \sqrt{n} \underbrace{\E_n[m_{\hat{a}}(Z; \hat{g})] - \E[m_{\hat{a}}(Z; \hat{g})]}_{A} + o_p(1).
\end{align*}

If $\|\hat{a}-a_*\|_2 \to_p 0$ and $\|\hat{g}-g_*\|_2\to_p 0$ then we can further decompose $A$ as
\begin{align*}
    A =~& \E_n[m_{a_*}(Z; g_*)] - \E[m_{a_*}(Z; g_*)] + \{\E_n\left[m_{\hat{a}}(Z; \hat{g}) - m_{a_*}(Z; g_*)\right] - \E[m_{\hat{a}}(Z; \hat{g}) - m_{a_*}(Z; g_*)]\}.
\end{align*}
The latter two terms in $A$ form an empirical process.
\item Critical radius theory. We derive a concentration inequality similar to~\eqref{eqn:reg-concentration} for the empirical process.  Let $\delta_{n,\zeta}=\delta_n + c_0 \sqrt{\frac{\log(c_1/\zeta)}{n}}$. 
Recall \cite[Lemma 14]{foster2019orthogonal}, which holds for any loss $\ell$ that is Lipschitz in $f(x)$, and when $\|f\|_{\infty}\leq 1$: fix $f'\in\mcF$, then
with probability $1-\zeta$, $\forall f\in\mcF$
$$
 |(\mathbb{E}_n-\mathbb{E})[\{\ell\{f(x),z\}-\ell\{f'(x),z\}]|=O(\delta_{n,\zeta} \|f-f'\|_2+\delta_{n,\zeta}^2).
$$
In order to place critical radius assumptions on the centered function spaces, we take $f'(x)=(0,0,0)$ and we take $f(x)=\{m(z;g-g_*),g(x)-g_*(x),a(x)-a_*(x)\}$. Notice that the loss is Lipschitz since
\begin{align*}
    &m_{a}(Z; g) - m_{a_*}(Z; g_*) 
    =m(Z;g)+a(X)\{Y-g(X)\}-[m(Z;g_*)+a_*(X)\{Y-g_*(X)\}] \\
    &=m(Z;g-g_*) +a(X)\{Y-g_*(X)+g_*(X)-g(X)\}-a_*(X)\{Y-g_*(X)\} \\
    &=m(Z;g-g_*)+\{a(X)-a_*(X)\})\{Y-g_*(X)\}+a(X)\{g_*(X)-g(X)\}\\
    &=m(Z;g-g_*)+\{a(X)-a_*(X)\})\{Y-g_*(X)\}+\{a(X)-a_*(X)\}\{g_*(X)-g(X)\}\\
    &\quad +a_*(X)\{g_*(X)-g(X)\}
\end{align*}
so its derivative is, in absolute value, $(1,|a(X)-a_*(X)|+|a_*(X)|,|Y-g_*(X)|+|g(X)-g_*(X)|)$.\footnote{It is also nonlinear in $f$, so we place the regularity condition of a lower bound on $\delta_n$ in Assumption~\ref{ass:critical2}.} Here, $|a_*(X)|,|Y-g_*(X)|$ are bounded by hypothesis, while $(\hat{\mcG}-g_*)_{B}$, $(\hat{\mcA}-a_*)_B$ have uniformly bounded ranges in $[-b, b]$.
Hence for all $g-g_* \in (\hat{\mcG}-g_*)_{B}$ and for all $(a-a_*)\in (\hat{\mcA}-a_*)_B$, with probability $1-\zeta$, invoking Assumption~\ref{ass:strong-smooth},
\begin{align*}
&\left|\E_n\left[m_{a}(Z; g) - m_{a_*}(Z; g_*)\right] - \E[m_{a}(Z; a) - m_{a_*}(Z; g_*)]\right|\\
&=O\left[\delta_{n,\zeta}\left\{ \|m\circ(g-g_*)\|_2 + \|g-g_*\|_2 + \|a-a_*\|_2 \right\}+\delta_{n,\zeta}^2\right] \\
&=O\left[\delta_{n,\zeta}\left\{ M\|g-g_*\|_2 + \|a-a_*\|_2 \right\}+\delta_{n,\zeta}^2\right].
\end{align*}
Next consider $\|g-g_*\|_{\hat{\mcG}} \geq \sqrt{B}$ and $\|a-a_*\|_{\hat{\mcA}}\geq \sqrt{B}$. We apply the previous result for $(g-g_*)\sqrt{B}/\|g-g_*\|_{\hat{\mcG}}$ and $(a-a_*)\sqrt{B}/\|a-a_*\|_{\hat{\mcA}}$, then multiply both sides by $\|g-g_*\|_{\hat{\mcG}}/\sqrt{B}$ and $\|a-a_*\|_{\hat{\mcA}}/\sqrt{B}$. Hence for all $g-g_*$ and for all $a-a_*$, with probability $1-\zeta$,
\begin{align*}
&\left|\E_n\left[m_{a}(Z; g) - m_{a_*}(Z; g_*)\right] - \E[m_{a}(Z; a) - m_{a_*}(Z; g_*)]\right|\\
&=O\left[\delta_{n,\zeta}\left\{ M\|g-g_*\|_2\|a-a_*\|_{\hat{\mcA}} + \|a-a_*\|_2\|g-g_*\|_{\hat{\mcG}} \right\}+\delta_{n,\zeta}^2\|g-g_*\|_{\hat{\mcG}}\|a-a_*\|_{\hat{\mcA}}\right].
\end{align*}

\item Bounding the empirical process. 
Applying this concentration inequality, with probability $1-\zeta$
\begin{multline*}
     \left|\E_n\left[m_{\hat{a}}(Z; \hat{g}) - m_{a_*}(Z; g_*)\right] - \E[m_{\hat{a}}(Z; \hat{g}) - m_{a_*}(Z; g_*)]\right| \\
     \leq O\left(\delta_{n,\zeta} M \left(\|\hat{a}-a_*\|_2\, \|\hat{g}-g_*\|_{\hat{\mcG}} + \|\hat{g}-g_*\|_2\, \|\hat{a}-a_*\|_{\hat{\mcA}}\right) + \delta_{n,\zeta}^2\, \|\hat{g}-g_*\|_{\hat{\mcG}} \, \|\hat{a}-a_*\|_{\hat{\mcA}}\right).
\end{multline*}
Let $\bar{\delta} = \delta_n + c_0 \sqrt{\frac{c_1 n}{n}}$. Then 
\begin{multline*}
    \left|\E_n\left[m_{\hat{a}}(Z; \hat{g}) - m_{a_*}(Z; g_*)\right] - \E[m_{\hat{a}}(Z; \hat{g}) - m_{a_*}(Z; g_*)]\right| \\
     = O_p\left(\bar{\delta} M \left(\|\hat{a}-a_*\|_2\, \|\hat{g}-g_*\|_{\hat{\mcG}} + \|\hat{g}-g_*\|_2\, \|\hat{a}-a_*\|_{\hat{\mcA}}\right) + \bar{\delta}^2\, \|\hat{g}-g_*\|_{\hat{\mcG}} \, \|\hat{a}-a_*\|_{\hat{\mcA}}\right).
\end{multline*}
If $\|\hat{a} - a_*\|_2, \|\hat{g}-g_*\|_2= O_p(r_n)$ and $\|\hat{a}-a_*\|_{\hat{\mcA}}, \|\hat{g}-g_*\|_{\hat{\mcG}} = O_p(1)$, then
\begin{equation*}
    \left|\E_n\left[m_{\hat{a}}(Z; \hat{g}) - m_{a_*}(Z; g_*)\right] - \E[m_{\hat{a}}(Z; \hat{g}) - m_{a_*}(Z; g_*)]\right| 
     = O_p\left(M\, \bar{\delta} r_n  + \bar{\delta}^2\right).
\end{equation*}
Thus as long as $\sqrt{n}\left(\bar{\delta} r_n + \bar{\delta}^2\right) \to 0$, we have that
\begin{equation*}
    \sqrt{n}\left|\E_n\left[m_{\hat{a}}(Z; \hat{g}) - m_{a_*}(Z; g_*)\right] - \E[m_{\hat{a}}(Z; \hat{g}) - m_{a_*}(Z; g_*)]\right| 
     = o_p(1).
\end{equation*}

\item Collecting results. We conclude that
$
    \sqrt{n}\left(\hat{\theta} - \theta_0\right) = \sqrt{n} \left(\E_n[m_{a_*}(Z; g_*)] - \E[m_{a_*}(Z; g_*)]\right) + o_p(1).
$
By the central limit theorem, the final expression is asymptotically normal with asymptotic variance $\sigma_*^2 =\Var(m_{a_*}(Z; g_*))$. \qedhere
    \end{enumerate}
\end{proof}

\begin{proof}[Proof of Theorem~\ref{thm:debias-nocross-stability}]
    Define the notation
    $
    h:=(g,a)$ and $V(Z; h):=m_{a}(Z; g) - m_{a_*}(Z; g_*) - \E[m_{a}(Z; g) - m_{a_*}(Z; g_*)].
    $
    We argue that
$
    \sqrt{n}\, \E_n\left[V(Z; \hat{h})\right] = o_p(1).
$
Then remainder of the proof is identical to the proof of Theorem~\ref{thm:debias-nocross}. For the desired property to hold, it suffices to show that $n\, \E\left[\E_n\left\{V(Z; \hat{h})\right\}^2\right] \to 0$.
First we rewrite the differences 
\begin{align*}
    V(Z; h) - V(Z;h') 
    =~& m(Z; g-g') + (a(X) - a'(X))\, Y - a(X) g(X) + a'(X) g'(X)\\
    & - \left(\ldot{a_0}{g-g'}_2 - \ldot{a}{g}_2 + \ldot{a'}{g'}_2 + \ldot{a-a'}{g_0}_2\right).
\end{align*}
By Assumption~\ref{ass:strong-smooth} and boundedness,
$
\E\left[\left(V(Z; h) - V(Z; h')\right)^2\right] \leq c_0 \E\left[\|h(X) - h'(X)\|_{\infty}^2\right]
$
for some constant $c_0$. Moreover, since, for every $x, y$: $x^2 \leq y^2 + |x|\, |x-y| + |y|\, |x-y|$,
\begin{align*}
    &\E\left[\E_n[V(Z; \hat{h})]^2\right] 
    = \frac{1}{n^2} \sum_{i, j}\E\left[V(Z_i; \hat{h})  V(Z_j;\hat{h})\right]\\
    \leq~& \frac{1}{n^2}\sum_{i, j} \left(\E\left[V(Z_i;\hat{h}^{-i,j})  V(Z_j;\hat{h}^{-i,j})\right] + 2\, \E\left[\left|V(Z_i;\hat{h}^{-i,j})\right|\, \left|V(Z_j;\hat{h}^{-i,j}) - V(Z_j; \hat{h})\right|\right]\right)\\
    \leq~& \frac{1}{n^2}\sum_{i, j} \left(\E\left[V(Z_i;\hat{h}^{-i,j})  V(Z_j;\hat{h}^{-i,j})\right] + 2\, \sqrt{\E\left[V(Z_i;\hat{h}^{-i,j})^2\right]}\, \sqrt{\E\left[\left(V(Z_j;\hat{h}^{-i,j}) - V(Z_j; \hat{h})\right)^2\right]}\right)\\
    \leq~& \frac{1}{n^2}\sum_{i, j} \left(\E\left[V(Z_i;\hat{h}^{-i,j})  V(Z_j;\hat{h}^{-i,j})\right] + 2\, c_0\,  \sqrt{\E\left[V(Z_i;\hat{h}^{-i,j})^2\right]}\, \sqrt{\E\left[\|\hat{h}^{-i,j}(X_j) - \hat{h}(X_j)\|_{\infty}^2\right]}\right)\\
    \leq~& \frac{1}{n^2}\sum_{i, j} \left(\E\left[V(Z_i;\hat{h}^{-i,j})  V(Z_j;\hat{h}^{-i,j})\right] + 8\, c_0\, \beta_{n-1} \sqrt{\E\left[V(Z_i;\hat{h}^{-i,j})^2\right]}\right).
\end{align*}
For every $i\neq j$ we have
\begin{align*}
    \E[V(Z_i; \hat{h}^{-i,j})  V(Z_j; \hat{h}^{-i,j})] =~& \E\left[\E\left[V(Z_i; \hat{h}^{-i})  V(Z_j; \hat{h}^{-j}) \mid \hat{h}^{-i,j}\right]\right]\\ 
    =~& \E\left[\E\left[V(Z_i;\hat{h}^{-i,j}) \mid \hat{h}^{-i,j}\right]  \E\left[V(Z_j; \hat{h}^{-i,j}) \mid \hat{h}^{-i,j}\right]\right] = 0;\\
\sqrt{\E[V(Z; \hat{h}^{-i,j})^2]} \leq~& O\left(\sqrt{\E\left[\|\hat{a}^{-i,j} - a_*\|_2^2 + \|\hat{g}^{-i,j} - g_*\|_2^2\right]}\right) = O(r_{n-2})\\
\E[V(Z; \hat{h}^{-i})^2] \leq~&  O\left(\E\left[\|\hat{a}^{-i} - a_*\|_2^2 + \|\hat{g}^{-i} - g_*\|_2^2\right]\right) = O(r_{n-1}^2).
\end{align*}
Thus we get that
$
    n\, \E\left[\E_n[V(Z; \hat{h})]^2\right] = \frac{1}{n} \sum_{i=1}^n \E[V(Z_i; \hat{h}^{-i})^2]  + O\left(\beta_{n-1} r_{n-2}\right)= O\left(r_{n-1}^2 + n\,\beta_{n-1} r_{n-2}\right).
$
In summary, it suffices to assume that
$
r_{n-1}^2 + n\,\beta_{n-1} r_{n-2}\to 0
$.
\end{proof}
\section{Application: Asset pricing}\label{sec:asset}

\subsection{Fundamental asset pricing equation}

A variety of financial models deliver the same fundamental asset pricing equation.  We briefly summarize how this equation involves a Riesz representer called the stochastic discount factor (SDF), which plays a central role in the theory of asset pricing. To do so, we introduce some additional notation. Define $\mathbb{E}_{t,i}(\cdot)=\mathbb{E}(\cdot|I_t,I_{t,i})$ where $I_t$ are macroeconomic conditioning variables that are not asset specific, e.g. inflation rates and market return, and $I_{t,i}$ are asset-specific characteristics, e.g. the size or book-to-market ratio of firm $i$ at time $t$.
Let $N$ be the number of firms.

The no-arbitrage condition is equivalent to the existence of strictly positive SDF $M_{t+1}$ such that for any asset,
$
\mathbb{E}_{t,i}(M_{t+1}Q_{t+1,i})=P_{t,i},
$
where $P_{t,i}$ is the price of asset $i$ at time $t$, $Q_{t+1,i}$ is payoff of asset $i$ at time $t+1$, and $M_{t+1}$ is the market-wide SDF at time $t+1$. This equation can also be parametrized in terms of returns. If an investor pays one dollar for an asset $i$ today, the gross rate of return $R_{t+1,i}$ is how many dollars the investor receives tomorrow; formally, the price is $P_{t,i}=1$ and the payoff is $Q_{t+1,i}=R_{t+1,i}$. If an investor borrows a dollar today at the interest rate $R_{t+1}^f$ and buys an asset $i$ that gives the gross rate of return $R_{t+1,i}$ tomorrow, then, from the perspective of the investor who paid nothing out-of-pocket, the price is $P_{t,i}=0$ while the payoff is the excess rate of return $R_{t+1,i}^e:=R_{t+1,i}-R_{t+1}^f$, leading to equation $
\mathbb{E}_{t,i}(M_{t+1}R^e_{t+1,i})=0
$ for any excess return $R^e_{t+1,i}$. 

Another helpful expression involves the tangency portfolio and its portfolio weights. If the fundamental asset pricing equation holds, then $M_{t+1}=1-\sum_{j=1}^N w_0(I_t,I_{t,j})R^e_{t+1,j}$. Here, $w_0(I_t,I_{t,j})$ are called the portfolio weights in the tangency portfolio $\sum_{j=1}^N w_0(I_t,I_{t,j})R^e_{t+1,j}$, and we have assumed that they are summarized by the mapping $w_0$. By substitution and the law of iterated expectations, we conclude that, for any function $g_{t,i}$,
\begin{equation}\label{eq:asset} 
\mathbb{E}\{g_{t,i}(I_t,I_{t,i})R^e_{t+1,i}\}=\mathbb{E}\left\{g_{t,i}(I_t,I_{t,i})R^e_{t+1,i}\sum_{j=1}^N w_0(I_t,I_{t,j})R^e_{t+1,j}\right\}.
\end{equation}

\subsection{Matching symbols}

We clarify the sense in which the SDF estimation problem generalizes the problem we study in the main text. Observe that the right hand side of~\eqref{eq:asset} pools information across $j\in\{1,...,N\}$. Fix $t\in \{1,...,T-1\}$. Stacking~\eqref{eq:asset} across $i\in \{1,...,N\}$ gives a system of $N$ equations. We view the left hand side as a $\R^{N}$-valued functional of the $\R^{N}$-valued function $g_t=(g_{t,1},...,g_{t,N})$. In particular, the $i$-th component of $\theta(g_t)=\mathbb{E}\{m(Z_t;g_t)\}$ is $\mathbb{E}\{m_i(Z_t;g_t)\}=\mathbb{E}\{g_{t,i}(I_t,I_{t,i})R^e_{t+1,i}\}$, where $Z_t=(I_t,I_{t,1},...,I_{t,N},R^e_{t+1,j})$. We view the right hand side as an elementwise product between $g_t$ and $a_0$. In particular, $a_0$ is an $\R^{N}$-valued function as well whose $i$-th component is $R^e_{t+1,i}\sum_{j=1}^N w_0(I_t,I_{t,j})R^e_{t+1,j}$. In summary, the SDF problem is a generalization to vector-valued functions and their vector-valued functionals.

Next we define the key function spaces. We view the portfolio weights $w_0$ as a natural way to parametrize $a_0$, and in fact as a way to define $\mcA=\mcA_{\text{sdf}}$ in our estimation procedure. Since each $g_{t,i}$ only takes as arguments $(I_t,I_{t,i})$, one may consider $\mcG=\mcG_{\text{sdf}}$ as the space of $\R^N$-valued functions whose $i$-th components only depend on $(I_t,I_{t,i})$. As before, we may define $\mcF$ from $\mcG$ accordingly. Since $a_0$ is pinned down by $w_0$ and the latter has a structural interpretation, we may report $\hat{w}\in \mathcal{W}_{\text{sdf}}$ instead of $\hat{a}\in \mcA_{\text{sdf}}$.

For an estimator to be well-defined as the optimizer of a criterion, that criterion must be a scalar-valued function. As we see above, the key objects in the SDF problem are vector-valued. Therefore to derive an estimator, we must posit a way to aggregate the vector objective in $\R^N$ into a scalar objective in $\R$. A natural choice is the empirical average $\frac{1}{N}\sum_{i=1}^n (\cdot)$. Implementing this choice, we arrive at our SDF estimator.

\begin{algo}[Adversarial SDF]\label{algo:sdf-estimator}
    For regularization $(\lambda,\mu)$, define
    \begin{equation*}
    \hat{a} = \argmin_{a\in \mathcal{A}_{\text{sdf}}} \max_{f\in \mcF} \E_t\E_i \{m_i(Z_t; f_t) - a_i(Z_t)\cdot f_i(Z_t)-f_i(Z_t)^2\} - \lambda \|f\|_{\mcA}^2 + \mu \|a\|_{\mcA}^2
\end{equation*}
where $\E_t(\cdot)=\frac{1}{T-1}\sum_{t=1}^{T-1}(\cdot)$ and $\E_i(\cdot)=\frac{1}{N}\sum_{i=1}^{N}(\cdot)$.
\end{algo}

\subsection{Detailed comparisons}

Our contribution most directly relates to previous works that propose adversarial \cite{hansen1997assessing}, neural network \cite{bansal1993no,chen2009land}, and adversarial neural network \cite{chen2019deep} approaches to estimate the SDF. Previous work either provides formal guarantees or allows for flexible, non-Donsker function spaces. We complement these prior works by proposing an adversarial estimator that simultaneously achieves both. More generally, our formal guarantees provide theoretical justification for the impressive empirical performance of machine learning in asset pricing. See e.g. \cite{chen2019deep} for a recent review of machine learning in asset pricing and \cite{christensen2017nonparametric} for a recent review of SDF estimation under Donsker assumptions.

    Our objective for the SDF is
    $$\min_{a\in \mathcal{A}_{\text{sdf}}} \max_{f\in \mcF} \E_t\E_i \{m_i(Z_t; f_t) - a_i(Z_t)\cdot f_i(Z_t)-f_i(Z_t)^2\} - \lambda \|f\|_{\mcA}^2 + \mu \|a\|_{\mcA}^2.$$
   For comparison, \cite{chen2019deep} essentially propose the following objective for the SDF:
   $$\min_{a\in \mathcal{A}_{\text{sdf}}} \max_{f\in \mcF} \E_i [\E_t\{m_i(Z_t; f_t) - a_i(Z_t)\cdot f_i(Z_t)\}]^2.$$
   The main difference is that we provide theoretical guarantees. 
   In the main text, we discuss how our results apply to variants of our estimator without regularization, using additional assumptions.

\section{Sparse linear case}\label{sec:sparse}

\subsection{Complexity-rate robustness versus double rate robustness}

How does ``complexity-rate robustness'' compare to ``double rate robustness'' in a concrete example?    
    Consider estimating $\hat{g}$ and $\hat{a}$ with $\ell_1$ constrained linear functions in $p$ dimensions. Suppose that $(g_*,a_*)$ are sparse linear functions with at most $s\ll n$ nonzero coefficients. Under a restricted eigenvalue condition on the covariates, a lasso estimator $\hat{g}$ and a sparse linear adversarial estimator $\hat{a}$ satisfy  $\bar{\delta}=O\left\{\sqrt{\frac{s\log(p)}{n}}\log(n)\right\}$ and  $r_n=O\left\{\sqrt{\frac{s\log(p)}{n}}\right\}$. Thus $s=o\left[\sqrt{n}/\{\log(p)\log(n)\}\right]$ suffices for Theorem~\ref{thm:debias-nocross}. See formal statements below.
    
    For comparison, for the sample splitting estimator, a sufficient condition is double rate robustness, which amounts to double sparsity robustness: $\sqrt{s_g\, s_a} = o\left\{\sqrt{n}/\log(p)\right\}$, where $(s_g,s_a)$ are the numbers of nonzero coefficients of $(g_*,a_*)$.\footnote{More generally, for the sample splitting estimator, the well known sufficient condition is $\sqrt{n}(r_n^ar_n^g+\epsilon_n^a\epsilon_n^g)\rightarrow0$, where $\|\hat{a}-a_*\|_2=o_p(r^a_n)$,  $\|a_*-a_0\|_2\leq \epsilon^a_n$, $\|\hat{g}-g_*\|_2=o_p(r^g_n)$, and $\|g_*-g_0\|_2\leq \epsilon^g_n$.} 
    
    Theorem~\ref{thm:debias-nocross}'s requirement is slightly stronger. It disallows a setting in which one nuisance is quite dense while the other is quite sparse. Still, if both nuisances are moderately sparse, then sample splitting can be eliminated, improving the effective sample size.

     Theorem~\ref{thm:debias-nocross}'s sufficient condition for the sparse linear setting, $s=o\left[\sqrt{n}/\{\log(p)\log(n)\}\right]$, recovers the sufficient conditions of \cite[eq. 5.5]{belloni2014pivotal} and \cite[Condition 3(iv)]{belloni2014uniform} up to logarithmic factors. In this sense it appears relatively sharp. We clarify a general complexity-rate robustness condition, implicit in \cite[eq. 19]{hirshberg2019augmented}, that applies to a broad range of settings beyond the sparse linear case.

On the one hand, eliminating sample splitting may increase the effective sample size. On the other, it may increase ``own observation'' bias \cite{newey2018cross}. Section~\ref{sec:sim} and Appendix~\ref{sec:sim_detail} show that the former phenomenon may outweigh the latter in some cases. Future research may formalize this trade-off in finite samples.

\subsection{Fast rate}

We now demonstrate how our general results specialize to sparse linear function spaces. We consider exactly sparse functions, then approximately sparse functions. For the former, define the class of $s$-sparse linear function classes in $p$ dimensions with bounded coefficients, i.e., 
$
\mcA_{\splin} :=\{x \to \ldot{\theta}{x}: \|\theta\|_{0} \leq s, \|\theta\|_{\infty}\leq b\}. 
$
Observe that $\mcF$ is also the class of $s$-sparse linear functions, with bounded coefficients in $[-2b,2b]$.

\begin{corollary}[Exactly sparse linear Riesz representer rate]\label{cor:sparse-linear}
Take $\mcA=\mcA_{\splin}$. Suppose that the covariate has a bounded $\ell_1$-norm almost surely. 
      Then the estimator presented in Corollary~\ref{cor:main-error} satisfies, with probability $1-\zeta$,
    \begin{equation*}
        \|\hat{a}-a_0\|_{2} \leq O\left\{\min_{a\in A_{\splin}} \|a-a_0\|_2 + \sqrt{\frac{s\log(p\, b)\,\log(n)}{n}} + \sqrt{\frac{\log(1/\zeta)}{n}}\right\}.
    \end{equation*}
\end{corollary}

Towards an additional result, we introduce the following notation:
\begin{equation*}
\spanF_{\kappa}(\mcF) := \left\{\sum_{i=1}^p w_i f_i: f_i\in \mcF, \|w\|_1\leq \kappa, p\leq\infty\right\}.
\end{equation*}

\begin{proposition}[Mean square rate without $\ell_{2}$ regularization via symmetry]\label{prop:reg-main-error-2}
Consider a set of test functions $\mcF=\cup_{i=1}^d \mcF^i$, that decomposes as a union of $d$ symmetric test function spaces $\mcF^i$. Suppose that $\mcA$ is star-convex. Consider the adversarial estimator
\begin{equation*}
    \check{a} = 
    \argmin_{a\in \mcA} \,\,\,\,\sup_{f\in \mcF}\,\, \E_n\{m(Z; f) - a(X)\cdot f(X)\}+ \lambda \|a\|_{\mcA}.
\end{equation*}
Let $m\circ \mcF^i = \{m(\cdot; f): f\in \mcF^i\}$ and
$$
\delta_{n,\zeta}:=2\max_{i} \left\{\mcR(\mcF^i) + \mcR(m\circ \mcF^i)\right\} + c_0 \sqrt{\frac{\log(c_1\, d/\zeta)}{n}},
$$
for some universal constants $c_0, c_1$ and $B_{n,\lambda,\zeta}:= \left(\|a_0\|_{\mcA} + \delta_{n,\zeta}/\lambda\right)^2$. Then with probability $1-\zeta$, we have $\|\check{a}\|_{\mcA} \leq \|a_0\|_{\mcA} + \delta_{n,\zeta}/\lambda$.
Moreover, suppose $\lambda\geq \delta_{n,\zeta}$ and
\begin{equation*}
\textstyle{\forall a\in \mcA_{B_{n,\lambda,\zeta}} \text{ with } \|a-a_0\|_2\geq \delta_{n,\zeta}: \frac{a - a_0}{\|a-a_0\|_2} \in \spanF_{\kappa}(\mcF)}.
\end{equation*}
Then with probability $1-\zeta$,
\begin{equation*}
    \|\check{a}-a_0\|_2 \leq \kappa \left[ 2\left(\|a_0\|_{\mcA}+1\right\} \mcR(\mcA_1) + \delta_{n,\zeta} + \lambda \left\{\|a_0\|_{\mcA}-\|\hat{a}\|_{\mcA}\right\}\right].
\end{equation*}
\end{proposition}

Proposition~\ref{prop:reg-main-error-2}, allows us to prove a similar guarantee that relaxes the hard sparsity constraint in $\mcA$. Instead, we consider $\ell_1$-bounded high dimensional linear function classes and we impose a restricted eigenvalue condition which is typical for such relaxations.

\begin{corollary}[Approximately sparse linear Riesz representer rate]\label{cor:sparse-linear-reg-ell1}
Suppose that $a_0(x)=\ldot{\theta_0}{x}$ with $\|\theta_0\|_0\leq s$, $\|\theta_0\|_1\leq B$, $\|\theta_0\|_{\infty}\leq 1$ and $\|x\|_{\infty} \leq 1$.  Let $\delta_{n,\zeta}=c_0 \sqrt{\frac{\log(c_1 p/\zeta)}{n}}$ for appropriate universal constants $c_0, c_1$. Suppose that the covariance matrix $V=\E[XX']$ satisfies the restricted eigenvalue condition:
\begin{equation*}
    \forall \nu\in \R^p \text{ such that } \|\nu_{S^c}\|_1 \leq \|\nu_S\|_1 + \delta_{n,\zeta}/\lambda: \nu^\top  V\nu \geq \gamma \|\nu\|_2^2.
\end{equation*}
Let $\mcA = \{x\to \ldot{\theta}{x}: \theta \in \R^p\}$, $\|\ldot{\theta}{\cdot}\|_{\mcA}=\|\theta\|_1$, and $\mcF=\{x \to  \xi x_i: i\in [p], \xi\in \{-1, 1\}\}$. Then the estimator in Proposition~\ref{prop:reg-main-error-2}, with $\delta_{n,\zeta}\leq \lambda\leq \frac{\gamma}{8s}$, produces an estimate $\hat{a}(\cdot)=\ldot{\hat{\theta}}{\cdot}$ which satisfies, with probability $1-\zeta$, that $\|\hat{\theta}\|_1 \leq \|\theta_0\|_1 + \delta_{n,\zeta}/\lambda$ and that
\begin{equation*}
  \|\hat{a}-a_0\|_2 \leq O\left[\max\left\{1, \frac{1}{\lambda}\frac{\gamma}{s}\right\} \sqrt{\frac{s}{\gamma}} \left\{(\|\theta_0\|_1 + 1)\sqrt{\frac{\log(p)}{n}} + \sqrt{\frac{\log(p/\zeta)}{n}}\right\}\right].
\end{equation*}
\end{corollary}

If the unrestricted minimum eigenvalue of $V$ is at least $\gamma$, then the restricted eigenvalue condition always holds. We only require a condition on the population covariance matrix $V$ and not on the empirical covariance matrix.

    The Dantzig selector objective of \cite{chernozhukov2018global} is essentially
    $$
\tilde{a}'=\argmin_{a\in\mcA} \|a\|_{\mcA} \text{ such that } \|\mathbb{E}_n \{m(Z;f)-a(X)\cdot f(X)\}\|_{\infty}\leq \lambda 
    $$
where $f:\mcX\rightarrow\mathbb{R}^p$ is a pre-specified dictionary of basis functions.
    The lasso objective of \cite{chernozhukov2018learning} is essentially
    $$
    \tilde{a}''=\argmin_{a\in\mcA} \mathbb{E}_n\left\{-2m(Z;a)+a(X)^2\right\}+2\lambda\|a\|_{\mcA}.
    $$
    The constraint in the former may be viewed as the first order condition of the latter when $a(X)=\theta'f(X)$. By contrast, our objective in Proposition~\ref{prop:reg-main-error-2} is
    $$
    \check{a} = 
    \argmin_{a\in \mcA} \,\,\,\,\sup_{f\in \mcF}\,\, \E_n\{m(Z; f) - a(X)\cdot f(X)\}+ \lambda \|a\|_{\mcA},
    $$
    which may be viewed as directly regularizing the lasso first order condition. A key point of departure from previous work in sparse linear settings is the adversarial choice of $f\in \mcF$.

\subsection{Semiparametric inference}

We provide formal results 
for the sparse linear case. We  appeal to the generality of Assumption~\ref{ass:critical2}, whereby it suffices to study the complexity of $(\hat{\mcG}-g_*)_{B}$, $\{m\circ (\hat{\mcG}-g_*)\}_B$, and $(\hat{\mcA}-a_*)_B$ whose elements are bounded in $[-b,b]$.

We introduce the following general notation, which we will subsequently specialize for each of the function spaces of interest. Let $\mcA = \{\theta'\phi(X)\}$ with $\phi(X)\in \R^p$ and $\|\phi(X)\|_{\infty}\leq 1$ almost surely. Throughout, we assume $\E[\phi(X)\phi(X)']\succeq \gamma I$. Let $\|\theta'\phi(X)\|_{\mcA} = \|\theta\|_1$. Let $a_*=\theta_*'\phi(X)$ and suppose that $a_*$ is $s$-sparse. Let $\nu = \theta-\theta_*$. Let $T$ denote the set of coordinate indices of the support of $\theta_*$, and let $T^c$ denote its complement. 

\begin{definition}[Restricted cone]
    For any $\mu,\kappa\geq0$, we define the restricted cone
    $$
     (\mcA-a_*)_B = \{\nu'\phi(X): \|\nu\|_1^2 \leq B, \|\nu_{T^c}\|_1 \leq \mu \|\nu_T\|_1 + \kappa\}.
    $$
    Its elements are bounded in $[-b,b]$ with $b=\sqrt{B}$.
\end{definition}

\begin{corollary}[Sparse linear hypothesis spaces are restricted cones]\label{cor:cone}
    Consider the sparse linear hypothesis space of the form 
    $$
     (\tilde{\mcA}-a_*)_{\tilde{B}} = \{\nu'\phi(X): \|\theta\|_1^2 \leq\tilde{B}\},\quad \tilde{B}=(\kappa+\|\theta_*\|_1)^2,\quad \kappa\rightarrow0.
     $$
This space is a special case of the restricted cone with $B=(\kappa+2\|\theta_*\|_1)^2$, $\mu=1$, and $\kappa\rightarrow0$.
\end{corollary}

\begin{proposition}[Critical radius of restricted cone]\label{prop:cone}
    The critical radius of $ (\mcA-a_*)_B$ is 
    $$
\delta_n=c_1  \cdot \log(n) \log(B) B^{1/2} \cdot \max\left\{ \mu \sqrt{\frac{s\log(p)}{n\gamma}}, \kappa^{1/2} \left(\frac{\log(p)}{n}\right)^{1/4}\right\}.
$$    
    where $c_1$ is a universal constant.
\end{proposition}

\begin{corollary}[Critical radius of $(\hat{\mcG}-g_*)_B$]\label{cor:g_cr}
    For $\hat{g}$, consider the constrained lasso with a hard constraint of $B$ that is sufficiently large. Let $g_*$ be the minimizer of the population square loss over the high dimensional linear function space. Then $\hat{\mcG}$ is the restricted cone with $\mu$ a small constant and  $\kappa=0$. If $g_*$ is $s$ sparse and $\E[\phi(X)\phi(X)']\succeq \gamma I$ then the critical radius of $(\hat{\mcG}-g_*)_B$ is 
    $
\delta_n=c_1  \cdot \log(n) \log(B) B^{1/2} \cdot  \mu \sqrt{\frac{s\log(p)}{n\gamma}}.
$
\end{corollary}

\begin{corollary}[Critical radius of $\{m\circ(\hat{\mcG}-g_*)\}_B$]\label{cor:m_cr}
Suppose the conditions of Corollary~\ref{cor:g_cr} hold, replacing $\E[\phi(X)\phi(X)']\succeq \gamma I$ with $\E[\psi(X)\psi(X)']\succeq \gamma I$ where $\psi(X)=m\circ \phi(X)$. Then the critical radius of $\{m\circ(\hat{\mcG}-g_*)\}_B$ is
 $
\delta_n=c_1  \cdot \log(n) \log(B) B^{1/2} \cdot  \mu \sqrt{\frac{s\log(p)}{n\gamma}}.
$
\end{corollary}

\begin{corollary}[Critical radius of $(\hat{\mcA}-a_*)_B$]\label{cor:a_cr}
    For $\hat{a}$, consider the sparse linear adversarial estimator of Corollary~\ref{cor:sparse-linear-reg-ell1} with $\lambda=\gamma/8s$. Then $\hat{\mcA}$ is the restricted cone with $\mu=1$ and $\kappa=c_2 \frac{s}{\gamma} \sqrt{\frac{\log(p/\zeta)}{n}}$ where $c_2$ is a universal constant. If $a_*$ is $s$ sparse, $s=o\left\{\sqrt{\frac{n}{\log(p)}}\right\}$, and $\E[\phi(X)\phi(X)']\succeq \gamma I$ then the critical radius of $(\hat{\mcA}-a_*)_B$ is
    $\delta_n=c_3  \cdot \log(n) \log(\|\theta_*\|_1) \|\theta_*\|_1 \cdot \sqrt{\frac{s\log(p/\zeta)}{n\gamma}}$ where $c_3$ is a universal constant.
\end{corollary}

    The rates for $\hat{g}$ and $\hat{a}$ are $r_n=O\left\{\sqrt{\frac{s\log(p)}{n}}\right\}$. Corollaries~\ref{cor:g_cr},~\ref{cor:m_cr}, and~\ref{cor:a_cr} imply $\bar{\delta}=O\left\{\sqrt{\frac{s\log(p)}{n}}\log(n)\right\}$. Hence our complexity-rate robustness condition $\sqrt{n}\bar{\delta}\, r_n\to 0$ holds when $s=o\{\frac{\sqrt{n}}{\log(p)\log(n)}\}$.

\subsection{Computational analysis}

In the case of sparse linear functions, the estimator in Proposition~\ref{prop:reg-main-error-2} optimizes
\begin{equation*}\label{eqn:minimax-ell1}
    \min_{\theta\in \R^p: \|\theta\|_1\leq B}\,\, \max_{i\in [2p]} \E_n\left[m(Z; f_i) - f_i(X)\, \ldot{\theta}{X} \right] + \lambda \|\theta\|_1,
\end{equation*}
where $f_i(X) = X_i$ for $i\in \{1, \ldots, p\}$ and $f_i(X)=-X_i$ for $i\in \{p+1, \ldots, 2p\}$.
One optimization approach is sub-gradient descent, which yields an $\epsilon$-approximate solution after $O\left(p/\epsilon^2\right)$ steps. In what follows, we present and analyze an alternative optimization approach that improves the number of steps to $O\left\{\log(p)/\epsilon\right\}$. 

In this approach, we view the problem as a zero sum game with players $\theta$ and $i$, and use simultaneous gradient descent. The minimizer uses a strategy called Optimistic-Follow-the-Regularized-Leader (OFTRL) with an entropic regularizer. The maximizer uses a strategy called Optimistic Hedge (OH) over probability distributions on the finite set of test functions, analogous to \cite[Proposition 13]{dikkala2020minimax}.

To present the optimization routine, we rewrite the problem so that the maximizer optimizes over distributions in the $2p$-dimensional simplex, i.e.
\begin{equation*}
    \min_{\theta\in \R^p: \|\theta\|_1\leq B}\,\, \max_{w\in \R_{\geq 0}^{2p}:\|w\|_1=1} \E_n\left[m(Z; \ldot{w}{f}) - \ldot{w}{f}(X)\, \ldot{\theta}{X} \right] + \lambda \|\theta\|_1,
\end{equation*}
 where $f=(f_1, \ldots, f_{2p})$ denotes a $2p$ vector-valued function. To avoid the non-smoothness of the $\ell_1$ penalty, we introduce the augmented vector $V=(X; -X)$. We further rewrite the problem so that the minimizer optimizes over the positive orthant of a $2p$-dimensional vector $\rho=(\rho^+; \rho^-)$, with an $\ell_1$ bounded norm. Matching symbols, $\theta=\rho^+ - \rho^-$, and
\begin{equation*}
    \min_{\rho\in \R_{\geq 0}^{2p}: \|\rho\|_1\leq B}\,\, \max_{w\in \R_{\geq 0}^{2p}:\|w\|_1=1} \E_n\left[m(Z; \ldot{w}{f}) - \ldot{w}{V}\, \ldot{\rho}{V} \right] + \lambda \sum_{i=1}^{2p} \rho_i
\end{equation*}
where $\ldot{w}{f}(X) = \ldot{w}{V}$. Finally, for a matrix $A$, define the notation  $\|A\|_{\infty}=\max_{i, j} |A_{ij}|$.

\begin{proposition}[Sparse linear optimization converges]\label{prop:sparse-optimization-ell1}
Consider the procedure that for $t=1, \ldots, T$ sets
\begin{align*}
    \tilde{\rho}_{i, t+1} =~& \tilde{\rho}_{i, t} e^{- 2\frac{\eta}{B}\, \left\{- \E_n(V_i\, \ldot{V}{w_t}) + \lambda\right\} + \frac{\eta}{B}\, \left\{- \E_n(V_i\, \ldot{V}{w_{t-1}}) + \lambda\right\}} &
    \rho_{t+1} =~& \tilde{\rho}_{t+1}\, \min\left\{1, \frac{B}{\|\tilde{\rho}_{t+1}\|_1}\right\}\\
    \tilde{w}_{i, t+1} =~& w_{i, t} e^{2\, \eta\, \E_n\{m(Z; f_i) - V_i \ldot{V}{\rho_t}\} - \eta\, \E_n\{m(Z; f_i) - V_i \ldot{V}{\rho_{t-1}}\}} & w_{t+1} =~& \frac{\tilde{w}_{t+1}}{\|\tilde{w}_{t+1}\|_1}
\end{align*}
with $\tilde{\rho}_{i,-1}=\tilde{\rho}_{i,0}=1/e$ and $\tilde{w}_{i,-1}=\tilde{w}_{i,0}=1/(2p)$. Return $\bar{\rho}=\frac{1}{T} \sum_{t=1}^T \rho_t$. Then for 
\begin{equation*}
\eta=\frac{1}{4\|\E_n(VV^\top)\|_{\infty}},\quad T=16\|\E_n(VV^\top)\|_{\infty} \frac{4B^2 \log(B\vee 1) + (B+1) \log(2p)}{\epsilon},
\end{equation*}
where $\eta$ is the step size and $T$ is the number of steps, the parameter $\bar{\theta}=\bar{\rho}^{+} - \bar{\rho}^-$ is an $\epsilon$-approximate solution for the estimator in Proposition~\ref{prop:reg-main-error-2}.
\end{proposition}
\section{Fast rate details}\label{sec:rate_proof}

\subsection{Additional results: Mean square continuity and critical radius}

\begin{proposition}[Mean square continuity]\label{prop:continuity}
    Under simple conditions, the average treatment effect, average policy effect, policy effect from transporting covariates, cross effect, regression decomposition, average treatment on the treated, and local average treatment effect are mean square continuous.
\end{proposition}

\begin{lemma}[Comparing critical radii]\label{lemma:bounded}
    Consider the space $\mcF$, and define $\mcF^b=\{f/b:f\in\mcF\}$. Let $\delta_n$ be the critical radius of $\mcF^b$. Then the critical radius of $\mcF$ is bounded by $b\delta_n$.
\end{lemma}


\subsection{Proof of corollaries}

\begin{proof}[Proof of Corollary~\ref{cor:limit}]
By Riesz representation and the first order condition,
    \begin{align*}
    \max_{f\in \mcF} \E\left[m(Z; f) - a(X)\cdot f(X)\right] - \|f\|_{2}^2 =~& \max_{f\in \mcF} \E\left[\{a_0(X) - a(X)\}\cdot f(X) -  f(X)^2\right]\\
    =~& \frac{1}{4}\E\left[\{a_0(X) - a(X)\}^2\right] = \frac{1}{4} \|a-a_0\|_2^2.
\end{align*}
\end{proof}

\begin{proof}[Proof of Corollary~\ref{cor:weaker_metric}]
    The metric $\|\cdot\|_{\mcF}$ satisfies the triangle inequality:
\begin{equation*}
    \|a + b\|_{\mcF} \leq \sqrt{\sup_{f\in \mcF} \ldot{a}{f} - \frac{1}{4} \|f\|_2^2 + \sup_{f\in \mcF} \ldot{b}{f} - \frac{1}{4} \|f\|_2^2} \leq \|a\|_{\mcF} + \|b\|_{\mcF}.
\end{equation*}
It is also positive definite i.e. $\|0\|_{\mcF}=0$, but not necessarily homogeneous, i.e. $\|\lambda a\|_{\mcF}= |\lambda| \|a\|_{\mcF}$ for $\lambda\in \R$. Observe that $\|\cdot\|_{\mcF}$ satisfies
\begin{equation*}
    \|a\|_{\mcF}^2 = \inf_{f\in \mcF} \frac{1}{4} \|f\|_2^2 - \ldot{a}{f}_2 + \|a\|_2^2 - \|a\|_2^2 = \inf_{f\in \mcF} \left\|a - f/2\right\|_2^2 - \|a\|_2^2 \leq \inf_{f \in \mcF} \|a - f\|_2^2 - \|a\|_2^2
\end{equation*}
where in the last inequality we use the fact that $\mcF$ is star-convex. Thus it is at most the projection of $a$ on $\mcF$. Hence, it suffices that $\mcA$ approximates $a_0$ in this weak sense: for some $a_*\in \mcA$ the projection of $a_*-a_0$ on $\mcF$ is at most $\epsilon_n$. Any component of $a_0$ that is orthogonal to $\mcF$ can be ignored, since if we decompose $a_0 = a_0^{\perp} + a_0^{\parallel}$, then $\sup_{f\in \mcF} \ldot{a_0^{\perp}}{f}_2=0$ and hence $\|a_0 - a_*\|_{\mcF} = \|a_0^{\parallel} - a_*\|_{\mcF}$.
\end{proof}

\begin{proof}[Proof of Corollary~\ref{cor:main-error}]
The assumption of norm constraints implies $\|f\|_{\mcA}\leq U$ for all $f\in \mcF$, and hence $\|a\|_{\mcA} \leq U$ for all $a\in \mcA$. The result follows from Theorem~\ref{thm:reg-main-error}, taking $B\to \infty$, $\mcF_B \to \mcF$, and $m\circ \mcF_B\to m\circ \mcF$.
\end{proof}

\begin{proof}[Proof of Corollary~\ref{cor:union}]
   Recall~\eqref{eqn:metric-entropy-critical}:
$
    \frac{64}{\sqrt{n}} \int_{\frac{\delta^2}{2b}}^{\delta} \sqrt{\log\left[N_n\{\epsilon; B_n(\delta; \mcF)\}\right]} d\epsilon \leq \frac{\delta^2}{b}.
$
By definition of $B_n(\delta; \mcF)$ as the ball of radius $\delta$,
\begin{align*}
    B_n(\delta; \mcF) 
    &= \{f\in \mcF: \|f\|_{2,n}\leq \delta\} 
    = \cup_{i=1}^d \{f\in \mcF^i:\|f\|_{2,n}\leq \delta\}
    =\cup_{i=1}^d B_n(\delta; \mcF^i).
\end{align*}
Note that $N_n\{\epsilon; B^1\cup B^2\}\leq N_n\{\epsilon; B^1\}+N_n\{\epsilon; B^2\}$ due to its definition as the covering number at approximation level $\epsilon$, i.e. the size of the smallest $\epsilon$-cover. Therefore
\begin{align*}
    \log\left[N_n\{\epsilon; B_n(\delta; \mcF) \}\right]
    &=\log\left[N_n\{\epsilon; \cup_{i=1}^d B_n(\delta; \mcF^i)\}\right] \\
    &\leq \log\left[\sum_{i=1}^d N_n\{\epsilon; B_n(\delta; \mcF^i)\}\right] \\
    &\leq \log(d) + \max_{i} \log(N_n\{\epsilon; B_n(\delta; \mcF^i)\}).
\end{align*}
Consider the integral of the former term. Integrating from $0$ to $\delta$ amounts to muliplying by $\delta$. Therefore the former term contributes an additive quantity of the order
$
    \frac{64\delta }{\sqrt{n}}\sqrt{\log(d)} \leq \frac{\delta^2}{b}.
$
If $\delta \geq   2\cdot 64 b\sqrt{\frac{\log(d)}{n}}$, then the inital term is at most $\delta^2/(2b)$ . If $\delta\geq 2 \max_i \delta_n^i$, then the integral of the latter term will also be at most $\delta^2/(2b)$. Taking the sum delivers the result noting that Assumption~\ref{ass:critical} takes $b=1$.
\end{proof}

\subsection{Proof of additional results}

\begin{proof}[Proof of Proposition~\ref{prop:continuity}]
 Recall the definition of mean square continuity: $\exists M\geq 0$ such that
$
\forall f\in \mcF: \sqrt{\E\left[m(Z;f)^2\right]} \leq M\, \|f\|_2.
$
We verify mean square continuity for several important functionals. 
\begin{enumerate}
    \item Average treatment effect (ATE): $\theta_0=\mathbb{E}[g_0(1,W)-g_0(0,W)]$.
    To lighten notation, let $\pi_0(w):=\mathbb{P}(D=1|W=w)$ be the propensity score.
    Assume $\pi_0(w)\in \left(\frac{1}{M},1-\frac{1}{M}\right)$ for $M\in(1,\infty)$. Then
    \begin{align*}
        \E[g(1, W) - g(0, W)]^2&\leq~ 2\E[g(1,W)^2 + g(0, W)^2]\\
&\leq~ 2M \E\left[\pi_0(W)\, g(1,W)^2 + [1-\pi_0(W)]\, g(0, W)^2\right] \\
&= 2M \mathbb{E}[g(X)]^2.
    \end{align*}
    
    \item Average policy effect: $\theta_0=\int g_0(x)d\mu(x)$ where $\mu(x)=F_1(x)-F_0(x)$.
    Denote the densities corresponding to distributions $(F,F_1,F_0)$ by $(f,f_1,f_0)$. Assume $\frac{f_1(x)}{f(x)}\leq \sqrt{M}$ and $\frac{f_0(x)}{f(x)}\leq \sqrt{M}$ for $M\in[0,\infty)$. In this example, $m(Z;g)=m(g)$. Then
    \begin{align*}
        \mathbb{E}[m(Z;g)]^2&=\{m(g)\}^2 
        =\left\{\int g(x)d\mu(x)\right\}^2 
        =\left\{\mathbb{E}\left[ g(X)\left\{\frac{f_1(X)}{f(X)}-\frac{f_0(X)}{f(X)}\right\} \right]\right\}^2 \\
        &\leq \left\{2 \sqrt{M} \mathbb{E}|g(X)|\right\}^2 
        \leq 4M \mathbb{E}[g(X)]^2.
    \end{align*}
    
    \item Policy effect from transporting covariates: $\theta_0=\mathbb{E}[g_0(t(X))-g_0(X)]$.
    Denote the density of $t(X)$ by $f_t(x)$. Assume $\frac{f_t(x)}{f(x)}\leq M$ for $M\in[0,\infty)$. Then
    \begin{align*}
        \mathbb{E}[g(t(X))-g(X)]^2 
        &\leq 2\mathbb{E}[g(t(X))^2+g(X)^2] 
        = 2\mathbb{E}\left[g(X)^2\left\{\frac{f_t(X)}{f(X)}-1\right\}\right] \\
        &\leq 2 (M+1) \mathbb{E}[g(X)]^2.
    \end{align*}
    \item Cross effect: $\theta_0=\mathbb{E}[Dg_0(0,W)]$.
    Assume $\pi_0(w)<1-\frac{1}{M}$ for some $M\in(1,\infty)$. Then
    \begin{align*}
        \mathbb{E}[Dg(0,W)]^2&\leq \mathbb{E}[g(0,W)]^2 
        \leq M \mathbb{E}[\{1-\pi_0(W)\}g(0,W)^2] 
        \leq M \mathbb{E}[g(X)]^2.
    \end{align*}
    
    \item Regression decomposition: $\mathbb{E}[Y|D=1]-\mathbb{E}[Y|D=0]=\theta_0^{response}+\theta_0^{composition}$
    where
    \begin{align*}
        \theta_0^{response}&=\mathbb{E}[g_0(1,W)|D=1]-\mathbb{E}[g_0(0,W)|D=1] \\
        \theta_0^{composition}&=\mathbb{E}[g_0(0,W)|D=1]-\mathbb{E}[g_0(0,W)|D=0].
    \end{align*}
    Assume $\pi_0(w)<1-\frac{1}{M}$ for some $M\in(1,\infty)$. Then re-write the target parameters in terms of the cross effect:
      \begin{align*}
        \theta_0^{response}&=\frac{\mathbb{E}[DY]-\mathbb{E}[Dg_0(0,W)]}{\mathbb{E}[D]} \\
        \theta_0^{composition}&=
        \frac{\mathbb{E}[D\gamma_0(0,W)]}{\mathbb{E}[D]}-\frac{\mathbb{E}[(1-D)Y]}{\mathbb{E}[1-D]}.
    \end{align*}
   We reduce the regression decomposition into cross effects and population means.
    
    \item Average treatment on the treated (ATT): $\theta_0=\mathbb{E}[g_0(1,W)|D=1]-\mathbb{E}[g_0(0,W)|D=1]$.
    Assume $\pi_0(w)<1-\frac{1}{M}$ for some $M\in(1,\infty)$. Then re-write the target parameters in terms of the cross effect and population means:
      \begin{align*}
        \theta_0&=\frac{\mathbb{E}[DY]-\mathbb{E}[Dg_0(0,W)]}{\mathbb{E}[D]}.
    \end{align*}
    
    \item Local average treatment effect (LATE): $\theta_0=\frac{\mathbb{E}[g_0(1,W)-g_0(0,W)]}{\mathbb{E}[h_0(1,W)-h_0(0,W)]}$.
    The result follows from the view of LATE as a ratio of two ATEs.
\end{enumerate}   
\end{proof}

\begin{proof}[Proof of Lemma~\ref{lemma:bounded}]
   The Rademacher complexity is
   \begin{align*}
       {\cal R}(\delta; \mcF^b)
       &= \E\left\{\sup_{f\in \mcF^b: \|f\|_2\leq \delta} \frac{1}{n} \sum_{i=1}^n \epsilon_i f(X_i)\right\} 
       =\E\left\{\sup_{f\in \mcF: \|f\|_2/b\leq \delta} \frac{1}{n} \sum_{i=1}^n \epsilon_i f(X_i)/b\right\} \\
       &=b^{-1}\E\left\{\sup_{f\in \mcF: \|f\|_2\leq b\delta} \frac{1}{n} \sum_{i=1}^n \epsilon_i f(X_i)\right\} 
       = \frac{1}{b}{\cal R}(b\delta; \mcF).
   \end{align*}
   By hypothesis, ${\cal R}(\delta_n; \mcF^b)\leq \delta_n^2$. Combining these results,
   $$
   \frac{1}{b}{\cal R}(b\delta_n; \mcF)= {\cal R}(\delta_n; \mcF^b) \leq \delta_n^2
   \iff 
{\cal R}(b\delta_n; \mcF)\leq b\delta_n^2=\frac{(b\delta_n)^2}{b}.
   $$
   Hence, taking $\delta_n'=b\delta_n$,
   $
   {\cal R}(\delta'_n; \mcF)\leq b\delta_n^2=\frac{(\delta'_n)^2}{b},
   $
   so the critical radius of $\mcF$ is upper bounded by $\delta_n'$.
   \end{proof}
\section{Semiparametric inference details}\label{sec:inference_proof}

\subsection{Comparison to the Donsker condition}

The critical radius condition, in place of the Donsker condition, implies semiparametric inference in Theorem~\ref{thm:debias-nocross}. We now clarify the similarities and differences between these concepts. A standard decomposition (restated in Appendix~\ref{sec:main}) shows that
$$
\sqrt{n}(\hat{\theta}-\theta_0)=\nu_n(g_0,a_0)+\{\nu_n(\hat{g},\hat{a})-\nu_n(g_0,a_0)\}+\sqrt{n}\mathbb{E}[\{a_0(X)-\hat{a}(X)\}\{\hat{g}(X)-g_0(X)\}]
$$
where $\nu_n(g,a)=\sqrt{n}(\mathbb{E}_n-\mathbb{E})[m(Z;g)+a(X)\{Y-g(X)\}]$ is an empirical process indexed by nuisance functions. A central limit theorem gives $\nu_n(g_0,a_0) \overset{d}{\rightarrow}N(0,\sigma^2)$, and a product rate condition gives $\sqrt{n}\mathbb{E}[\{a_0(X)-\hat{a}(X)\}\{\hat{g}(X)-g_0(X)\}] \overset{p}{\rightarrow}0$. What remains to show is that $\{\nu_n(\hat{g},\hat{a})-\nu_n(g_0,a_0)\} \overset{p}{\rightarrow}0$, for which we refine previous arguments.

The Donsker condition involves stochastic equicontinuity, which means that $\|g_1-g_2\|_2\leq r_n$ and $\|a_1-a_2\|_2\leq r_n$, where $r_n\downarrow 0$, imply $\{\nu_n(g_1,a_1)-\nu_n(g_2,a_2)\} \overset{p}{\rightarrow}0$, so that the desired result holds \cite{andrews1994empirical}. A sufficient condition is Pollard's condition \cite{pollard1982central}, which we now state for the function class $\mcF$, uniformly bounded in $[-b,b]$:
\begin{equation}\label{eqn:pollard}
     \int_{0}^{\infty} \sup_{Q\in \mathcal{Q}}  \sqrt{\log\left[N\{\varepsilon b; \mcF;L_2(Q)\}\right]} d\varepsilon <\infty
\end{equation}
where $\mathcal{Q}$ is the the set of all finite discrete probability measures. A common analytic approach is to verify~\eqref{eqn:pollard} for simple function spaces \cite{pakes1989simulation,andrews1994asymptotics,pakes1995limit,ai2003efficient,chen2003estimation}. 
While smooth Sobolev spaces satisfy this condition, even rearranged Sobolev spaces do not \cite[Section 4.5]{chernozhukov2018global}, nor do $L_q$ balls for $q\in(0,1]$ \cite[Lemma 2]{raskutti2011minimax}.

By contrast, the less strict condition~\eqref{eqn:metric-entropy-critical} gives the critical radius $\delta_n$. When $\|\hat{g}-g_0\|_2\leq r_n$ and $\|\hat{a}-a_0\|_2\leq r_n$, Theorem~\ref{thm:debias-nocross} clarifies that $\{\nu_n(\hat{g},\hat{a})-\nu_n(g_0,a_0)\} =O_p\{\sqrt{n}(\delta_n r_n+\delta_n^2)\}$, simplifying and generalizing \cite[eq. 19]{hirshberg2019augmented}.\footnote{We also allow limited mis-specification.} 
We call this condition complexity-rate robustness
since higher complexity can be compensated by better estimation rates. The Donsker approach does not have such a condition. In summary, the critical radius approach gives a weaker sufficient condition than the Donsker approach for $\{\nu_n(\hat{g},\hat{a})-\nu_n(g_0,a_0)\} \overset{p}{\rightarrow}0$. It recovers known results in the sparse linear special case; see Appendix~\ref{sec:sparse}.

\subsection{Additional results: Local Riesz, mis-specification}

To lighten notation, when stating the local Riesz estimator results, we suppress indexing by the fold $k$, similar to Section~\ref{sec:algo}.

\begin{assumption}[Critical radius for local Riesz representation]\label{ass:critical3}
 Define the set
$
    \mcF(r_n) = \{f \in \text{star}\left(\partial(\mcG - \hat{g})\right): \|f\|_2\leq r_n \}.
$
  Assume $\delta_n$ upper bounds the critical radius of 
  $
\{Z \to m(Z; f) - a(X)\, f(X): f\in \mcF(r_n), a\in \mcA\}.
$
To lighten notation, let $\bar{\delta}:=\delta_n + c_0\sqrt{\frac{\log(c_1/\zeta)}{n}}$. 
\end{assumption}

\begin{proposition}[Local Riesz representation]\label{prop:local}
   Suppose Assumptions~\ref{ass:strong-smooth} and~\ref{ass:critical3} hold, and that $\|\hat{g}-g_0\|_2\leq r_n$ with probability $1-\zeta$. Consider the estimator
$$
    \hat{a}'=\arg\inf_{a \in \mcA}\,\, \sup_{f \in \mcF(r_n)} \frac{1}{n}\sum_{i=1}^{n} \left(m(Z_i; f) - a(X_i)\cdot f(X_i)\right).
$$
   Then with probability $1-\zeta$,
   $$
   \E[(\hat{a}'(X)-a_0(X))\, \left(\hat{g}(X) - g_0(X)\right)] \leq O\left(M\, r_n \bar{\delta} + \bar{\delta}^2 + r_n\, \inf_{a\in \mcA}\|a_0 - a\|_2\right).
   $$
   If in addition $r_n\, \bar{\delta} = o(n^{-1/2})$, $\bar{\delta} = o(n^{-1/4})$, and $r_n\inf_{a\in \mcA}\|a_0 - a\|_2 = o(n^{-1/2})$, then Assumption~\ref{ass:main-cond} holds.
\end{proposition}

    If $a_0 \in \mcA$ and both $\mcA$ and $\mcG$ are VC-subgraph classes with constant VC dimension, then  $\bar{\delta}=O\left(\sqrt{\frac{\log(n/\zeta)}{n}}\right)$. For the conclusions of Proposition~\ref{prop:local} to hold, it suffices that $r_n = o(1)$, i.e. that $\hat{g}$ is consistent in $L_2$.

   For the local Riesz results above, we require $\mcA$ to have a small approximation error to $a_0$ with respect to the weaker norm $
    \|a_0 - a\|_{\mcF} = \sup_{f\in \mcF(1)}  \ldot{a_0 - a}{f}.
$
Thus $a$ does not need to match the component of $a_0$ that is orthogonal to the subspace $\mcF$. 

For example, assume that $\mcF$ lies in the space spanned by top $K$ eigenfunctions of a reproducing kernel Hilbert space. Then it suffices to consider as $\mcA$ the space spanned by those functions too, and $\inf_{a\in \mcA} \|a_0 - a\|_{\mcF}=0$. 

As another example, if $\mcG$ is a finite dimensional linear function space and $g_0\in \mcG$, then it suffices to consider $\mcA$ that is also finite dimensional linear, even if the true $a_0$ does not lie in that sub-space. The conditions of Corollary~\ref{cor:debias} will be satisfied, even if $\hat{a}$ will never be consistent with respect to $a_0$.

\begin{assumption}[Mixed bias condition: Inconsistent nuisance]\label{ass:main-cond_inconsistent}
    Suppose that $\forall k \in [K]$: $\sqrt{n}\, \E[\{\hat{a}_k(X)-a_*(X)\}\, \{\hat{g}_k(X) - g_*(X)\}] \rightarrow_p 0$, where $g_*$ or $a_*$ may not necessarily equal $g_0$ or $a_0$.
\end{assumption}

\begin{proposition}[Normality with inconsistent nuisance]\label{prop:inconsistent}
Suppose Assumptions~\ref{ass:strong-smooth} and~\ref{ass:main-cond_inconsistent} hold. Further assume
  (i) boundedness: $Y$, $g(X)$, and $a(X)$ are bounded almost surely, for all $g\in \mcG$ and $a\in \mcA$; 
        (ii) individual rates:  $\|\hat{a}_k-a_*\|_2\stackrel{L^2}{\to} 0$  and $\|\hat{g}-g_*\|_2 \stackrel{L^2}{\to} 0$, where $g_*$ or $a_*$ may not necessarily equal $g_0$ or $a_0$.
Finally assume that $\hat{g}_k$ admits an asymptotically linear representation around the truth $g_0$, i.e.
        \begin{equation*}
    \sqrt{|P_k|}\left(\hat{g}_k(X) - g_0(X)\right) = \frac{1}{\sqrt{|P_k|}} \sum_{i\in P_k} \psi(X, Z_i; g_0) + o_p(1),\quad \E[\psi(X, Z_i; g_0)\mid X]=0.
\end{equation*}
    Then $ \sqrt{n}\sigma_*^{-1}\left(\check{\theta} - \theta_0\right) \to_d N\left(0, 1\right)$ where $
    \sigma_*^2 :=\Var_{Z_i}(m_{a_*}(Z_i; g_*) + \E_X\left[\{a_0(X) - a_*(X)\}\,\psi(X, Z_i; g_0)\right]).
$
Similarly, if $\hat{a}_k$ has an asymptotically linear representation around the truth, then the statement above holds with
$
    \sigma_*^2 :=\Var_{Z_i}(m_{a_*}(Z_i; g_*) + \E_X\left[\psi(X, Z_i; a_0)\, \{g_0(X) - g_*(X)\}\right]).
$
\end{proposition}

\subsection{Proof of corollaries}

\begin{proof}[Proof of Corollary~\ref{cor:debias}]
The result is well known. For completeness, we state the proof for $\check{\theta}$. The proof for $\tilde{\theta}$ is similar; see e.g. \cite[Corollary 4]{chernozhukov2018learning}.

    Observe that $\theta_0 = \E[m_{a}(Z; g_0)]$ for all $a$. Moreover,
\begin{align*}
    \check{\theta}-\theta_0 =~& \frac{1}{n} \sum_{k=1}^K \sum_{i\in P_k} \left(m_{\hat{a}_k}(Z_i; \hat{g}) - \E_Z[m_{\hat{a}_k}(Z; \hat{g}_k)]\right) + \frac{1}{K} \sum_{k=1}^K \left(\E_Z[m_{\hat{a}_k}(Z; \hat{g}_k)] - \E_Z[m_{\hat{a}_k}(Z; g_0)]\right)\\
    =~& \frac{1}{n} \sum_{k=1}^K \sum_{i\in P_k} \left(m_{\hat{a}_k}(Z_i; \hat{g}_k) - \E_Z[m_{\hat{a}_k}(Z; \hat{g}_k)]\right) + \frac{1}{K} \sum_{k=1}^K \E_X[(a_0(X) - \hat{a}_k(X))\, (\hat{g}_k(X) - g_0(X))].
\end{align*}
Hence by Assumption~\ref{ass:main-cond},
\begin{align*}
    \sqrt{n}\left(\check{\theta}-\theta_0\right) =~& \sqrt{n} \underbrace{\frac{1}{n} \sum_{k=1}^K \sum_{i\in P_k} \left(m_{\hat{a}_k}(Z_i; \hat{g}_k) - \E_Z[m_{\hat{a}_k}(Z; \hat{g}_k)]\right)}_{A} + o_p(1).
\end{align*}
If $\|\hat{a}_k-a_*\|_2 \to_p 0$ and $\|\hat{g}_k-g_*\|_2\to_p 0$ then we can further decompose $A$ as
\begin{align*}
    A =~& \E_n[m_{a_*}(Z; g_*)] - \E_Z[m_{a_*}(Z; g_*)] \\
    &+ \frac{1}{n} \sum_{k=1}^K \sum_{i\in P_k} \underbrace{m_{\hat{a}_k}(Z_i; \hat{g}_k) - m_{a_*}(Z_i; g_*) - \E_Z[m_{\hat{a}_k}(Z; \hat{g}_k) - m_{a_*}(Z; g_*)]}_{V_i}.
\end{align*}

To lighten notation, let $
    B := \frac{1}{n}\sum_{k=1}^K \sum_{i\in P_k} V_i =: \frac{1}{n} \sum_{k=1}^{K} B_k
$.
If $n\,\E[B^2]\to 0$, then $\sqrt{n} B \to_p 0$. The second moment of each $B_k$ is:
\begin{equation*}
    \E\left[B_k^2\right] = \sum_{i, j\in P_k} \E[V_i V_j] = \sum_{i, j\in P_k} \E[\E[V_i V_j \mid \hat{g}_k]] = \sum_{i\in P_k} \E\left[V_i^2\right] 
\end{equation*}
where in the last equality appeal to cross fitting: for any $i\neq j$, $V_i$ is independent of $V_j$ and mean zero, conditional on the nuisance $\hat{g}_k$ estimated on samples outside of fold $k$. 
Moreover, by Jensen's inequality with respect to $\frac{1}{K}\sum_{k=1}^K B_k$,
\begin{align*}
    \E[B^2] &= \E\left[\left(\frac{1}{n} \sum_{k=1}^K B_k\right)^2\right] = \frac{K^2}{n^2} \E\left[\left(\frac{1}{K} \sum_{k=1}^K B_k\right)^2\right] 
    \leq \frac{K}{n^2} \sum_{k=1}^K \E[B_k^2] = \frac{K}{n^2} \sum_{k=1}^K \sum_{i\in P_k} \E[V_i^2] = \frac{K}{n^2} \sum_{i=1}^n \E[V_i^2].
\end{align*}
Finally, observe that $\E[V_i^2] \to_p 0$, by Assumption~\ref{ass:strong-smooth} and boundedness. More elaborately,
\begin{align*}
    \E[V_i^2] \leq~& \E\left[\left(m_{\hat{a}}(Z_i; \hat{g}_k) - m_{a_*}(Z_i; g_*)\right)^2\right] \\
    \leq~& 2\,\E\left[\left(m(Z_i; \hat{g}_k) - m(Z_i; g_*)\right)^2\right] + 2\, \E[\left(\hat{a}_k(X)\,(Y - \hat{g}_k(X)) - a_*(X)\,(Y-g_*(X))\right)^2].
\end{align*}
The latter can further be bounded as
\begin{align*}
    4\E[\left(a_k(X) - a_*(X)\right)^2 (Y-g_k(X))^2] + 4\E[a_*(X)^2 (g_*(X) - g_k(X))^2] \leq 4 C\, \left(\E\left[\|\hat{a}_k - a_*\|_2^2 + \|\hat{g} - g_*\|_2^2\right]\right)
\end{align*}
if $(Y-\hat{g}_k(X))^2 \leq C$ and $a_*(X)^2 \leq C$ almost surely.
Finally, by linearity and Assumption~\ref{ass:strong-smooth},
\begin{equation*}
    \E[\left(m(Z_i; \hat{g}_k) - m(Z_i; g_*)\right)^2] = \E[\left(m(Z_i; \hat{g}_k - g_*)\right)^2] \leq M\, \E\left[\|\hat{g}_k - g_*\|_2^2\right].
\end{equation*}

In summary,
\begin{equation*}
    \E[V_i^2] \leq (2M + 4C)\, \left(\E\left[\|\hat{a}_k - a_*\|_2^2 + \|\hat{g} - g_*\|_2^2\right]\right) \to 0.
\end{equation*}
Thus as long as $K = \Theta(1)$, we have that
\begin{equation*}\label{eqn:crucial-normality}
    n\, \E[B^2] = \frac{K}{n} \sum_{i=1}^n \E[V_i^2] \leq (2M + 4C)\, K\, \E\left[\|\hat{g}-g_*\|_2^2 + \|\hat{a}-a_*\|_2^2\right] \to 0
\end{equation*}
and we can conclude
\begin{align*}
    \sqrt{n}\left(\check{\theta} - \theta_0\right) = \sqrt{n} \left(\E_n[m_{a_*}(Z; g_*)] - \E_Z[m_{a_*}(Z; g_*)]\right) + o_p(1).
\end{align*}
By the central limit theorem, the final expression is asymptotically normal with asymptotic variance $\sigma_*^2 =\Var(m_{a_*}(Z; g_*))$.
\end{proof}

\subsection{Proof of additional results}

\begin{proof}[Proof of Proposition~\ref{prop:local}]
By a localized concentration bound, $\forall a\in \mcA, f\in \mcF(r_n)$,
\begin{align*}
    \left| \Psi_n(a, f) - \Psi(a, f) \right| &= O\left(\bar{\delta} \|m(\cdot; f) - a\, f\|_2 + \bar{\delta}^2 \right)
    = O\left((M+1)\bar{\delta} \|f\|_2 + \bar{\delta}^2\right) \\
    &= O\left((M+1)\bar{\delta} r_n + \bar{\delta}^2\right) =: \epsilon_{n}.
\end{align*}
Hence
\begin{align*}
   \sup_{f \in \mcF(r_n)} \Psi(\hat{a}, f) - \epsilon_n &\leq \sup_{f \in\mcF(r_n)} \Psi_n(\hat{a}, f) 
   \leq \sup_{f \in \mcF(r_n)} \Psi(a_*, f) + \epsilon_n 
    = \inf_{a \in \mcA} \sup_{f \in \mcF(r_n)} \Psi(a, f) + \epsilon_n.
\end{align*}
We conclude that
\begin{equation*}
    \sup_{f \in \mcF(r_n)} \Psi(\hat{a}, f) \leq \inf_{a \in \mcA} \sup_{f \in \mcF(r_n)} \Psi(a, f) + 2\,\epsilon_n.
\end{equation*}
Moreover, if $a_0$ is a local Riesz representer, i.e. if $a_0$ satisfies Riesz representation for any $\hat{g}-g$, where $g\in\mcG$ is within a ball of size $r_n$ around $\hat{g}$, then
\begin{equation*}
    \inf_{a \in \mcA} \sup_{f \in \mcF(r_n)} \Psi(a, f)  = \inf_{a \in \mcA} \sup_{f \in \mcF(r_n)} \ldot{a_0 - a}{f} \leq r_n \inf_{a \in \mcA} \sup_{f\in \mcF(1)}  \ldot{a_0 - a}{f} \leq r_n \inf_{a\in \mcA} \|a_0 - a\|_2.
\end{equation*}
\end{proof}

\begin{proof}[Proof of Proposition~\ref{prop:inconsistent}]
    Observe that $\theta_0 = \E[m_{a}(Z; g_0)]$ for all $a$. Moreover,
\begin{align*}
    \check{\theta}-\theta_0 =~& \frac{1}{n} \sum_{k=1}^K \sum_{i\in P_k} \left(m_{\hat{a}_k}(Z_i; \hat{g}) - \E[m_{\hat{a}_k}(Z; \hat{g}_k)]\right) + \frac{1}{K} \sum_{k=1}^K \left(\E[m_{\hat{a}_k}(Z; \hat{g}_k)] - \E[m_{\hat{a}_k}(Z; g_0)]\right)\\
    =~& \underbrace{\frac{1}{n} \sum_{k=1}^K \sum_{i\in P_k} \left(m_{\hat{a}_k}(Z_i; \hat{g}_k) - \E[m_{\hat{a}_k}(Z; \hat{g}_k)]\right)}_{A} + \underbrace{\frac{1}{K} \sum_{k=1}^K \E[(a_0(X) - \hat{a}_k(X))\, (\hat{g}_k(X) - g_0(X))]}_{C}.
\end{align*}
By the proof of Corollary~\ref{cor:debias},
\begin{align*}
    \sqrt{n}\, A = \sqrt{n} \left(\E_n[m_{a_*}(Z; g_*)] - \E[m_{a_*}(Z; g_*)]\right) + o_p(1).
\end{align*}
Now we analyze term $C$. We will prove one of the two conditions in the ``or'' statement, when $\hat{g}_k$ has an asymptotically linear representation. The case when $\hat{a}_k$ is asymptotically linear can be proved analogously.

Let
$
    C_k := \E[(a_0(X) - \hat{a}_k(X))\, (\hat{g}_k(X) - g_0(X))] 
$.
We can then write:
\begin{align*}
    C_k =  \E[(a_*(X) - \hat{a}_k(X))\, (\hat{g}_k(X) - g_0(X))] + \E[(a_0(X) - a_*(X))\, (\hat{g}_k(X) - g_0(X))].
\end{align*}
Since
\begin{equation*}
\sqrt{|P_k|}\E[(a_*(X) - \hat{a}_k(X))\, (\hat{g}_k(X) - g_0(X))]\leq \sqrt{|P_k|} \|a_* - \hat{a}_k\|_2 \, \|\hat{g}_k - g_0\|_2 = \|a_* - \hat{a}_k\|_2\, O_p(1) = o_p(1),
\end{equation*}
we have that
\begin{align*}
    \sqrt{|P_k|} C_k =~& \sqrt{|P_k|}\E[(a_0(X) - a_*(X))\, (\hat{g}_k(X) - g_0(X))] + o_p(1)\\
    =~& \frac{1}{\sqrt{|P_k|}} \sum_{i\in P_k} \E_X[(a_0(X) - a_*(X))\,\psi(X, Z_i; g_0)] + o_p(1).
\end{align*}
Since $K=\Theta(1)$ and $n/|P_k| \to K$, we can therefore show
\begin{align*}
    \sqrt{n} C =~& \frac{\sqrt{n}}{K} \sum_{k=1}^K C_k =  \frac{\sqrt{K}}{K} \sum_{k=1}^K \sqrt{|P_k|} C_k + o(1)\\
    =~& \frac{1}{\sqrt{K}} \sum_{k=1}^K \frac{1}{\sqrt{|P_k|}} \sum_{i\in P_k} \E_X[(a_0(X) - a_*(X))\,\psi(X, Z_i; g_0)] + o_p(1)\\
    =~& \frac{1}{\sqrt{n}} \sum_{k=1}^K \sum_{i\in P_k} \E_X[(a_0(X) - a_*(X))\,\psi(X, Z_i; g_0)] + o_p(1)\\
    =~& \frac{1}{\sqrt{n}} \sum_{i\in [n]} \E_X[(a_0(X) - a_*(X))\,\psi(X, Z_i; g_0)] + o_p(1)\\
    =~& \sqrt{n} \E_n\left[\E_X[(a_0(X) - a_*(X))\,\psi(X,Z_i; g_0)]\right] + o_p(1).
\end{align*}

In summary,
\begin{align*}
    \sqrt{n}\left(\check{\theta}-\theta_0\right) =~& \sqrt{n} \left(\E_n\left[m_{a_*}(Z; g_*) + \E_X\left[(a_0(X) - a_*(X))\,\psi(X, Z_i; g_0)\right]\right] - \E[m_{a_*}(Z; g_*)]\right)  + o_p(1).
\end{align*}
By the central limit theorem, the final expression is asymptotically normal with asymptotic variance $\sigma_*^2 =\Var_{Z_i}(m_{a_*}(Z_i; g_*) + \E_X\left[(a_0(X) - a_*(X))\,\psi(X, Z_i; g_0)\right])$.
\end{proof}
\section{Computational analysis details}\label{sec:compute_proof}

\subsection{Stochastic gradient descent for neural network}

Consider the setting of Corollary~\ref{cor:nn}, which uses neural network function spaces. When $\mcF$ and $\mcA$ are represented by deep neural networks, then the optimization problem in Estimator~\ref{algo:reg-estimator} is highly non-convex. Beyond the challenge of using a non-convex function space, which also appears in problems with square losses, we face the challenge of a non-convex and non-smooth min-max loss. We describe off-the-shelf optimization methods for generative adversarial networks (GANs) that apply to our problem, and describe a closely related approximate guarantee for stochastic gradient descent. 

Similar to the optimization problem of GANs, our our estimator solves a non-convex, non-concave zero sum game, where the strategy of each player is a set of neural network parameters. A variety of recent iterative optimization algorithms for GANs inspired by zero sum game theory apply to our problem. See e.g. the ``optimistic Adam'' procedure \cite{Daskalakis2017}, which has been adapted to conditional moment models \cite{bennett2019deep,dikkala2020minimax}. The ``extra gradient'' procedures \cite{Hsieh2019,Mishchenko2019} provide further options. 

In practice, we find that a simple optimization procedure converges to the solution of Estimator~\ref{algo:reg-estimator} when using overparametrized neural networks. The procedure is to implement simultaneous gradient descent-ascent, then to obtain the average path by averaging over several iterations. We directly extend the main procedure of \cite{liao2020provably}, who prove convergence for min-max losses that are similar to our own, building on principles used to study square losses  \cite{AllenZhu2018,du2018gradient,Soltanolkotabi2019}.

    Neural networks that are sufficiently wide and randomly initialized behave like linear functions in an RKHS called the neural tangent kernel space. As long as the error of this approximation is carefully accounted for, one can invoke the analysis of sparse linear function spaces given in Appendix~\ref{sec:sparse}.
    Future work may formalize this intuition, generalizing the main result of \cite{liao2020provably}.

To facilitate optimization,
computational analysis increases the width of the neural network. Doing so deteriorates the statistical guarantee in Corollary~\ref{cor:nn}, since the critical radius grows as a function of the width; see Section~\ref{sec:intro}. 
Future work may improve the dependence on width, sharpening the results of \cite{liao2020provably},
to alleviate
this trade-off.

\subsection{Tuning regularization hyperparameters}

We provide details on tuning the regularization hyperparameters $(\lambda,\mu)$ in a theoretically justified manner. This section amounts to a summary of various computational results provided in the paper, as well as heuristics for additional hyperparameters that arise. Each approach is implemented in our publicly available replication package: 
\if1\blind
{
\url{https://colab.research.google.com/github/vsyrgkanis/adversarial_reisz/blob/master/Results.ipynb}. 
} \fi

\textbf{Neural network.} Consider the setting of Corollary~\ref{cor:nn}. Its justification appeals to Theorem~\ref{thm:reg-main-error}, which requires $\mu\geq 6\lambda\geq 12\bar{\delta}^2/B$. We take $\mu=6\lambda$ and $\lambda=10^{-4}$. In addition, we use early stopping as a form of adaptive regularization similar to \cite{bennett2019deep,dikkala2020minimax}. Specifically, we train with simultaneous optimistic ADAM, as described above, for 200 epochs, storing test functions after every few training iterations. Then we re-initiate training and use the maximum of moment violations over these finitely many stored test functions on an out-of-sample set, as a proxy for the maximum out-of-sample moment violation. The analyst does not need to train a neural network for every evaluation, and instead can use a representative set of test functions.

\textbf{Random forest.} Consider the setting of Corollary~\ref{cor:rf}. Its justification appeals to Proposition~\ref{prop:oracle}, which requires $\lambda=\mu=0$. We use similar settings to \cite{dikkala2020minimax}: each tree is linear in original variables, with a maximimum depth of two and minimum of 20 samples per leaf node. The estimator $\hat{a}$ has five trees in a forest, while the estimator $\hat{f}$ has 100 trees in a forest.

\textbf{RKHS}. Consider the setting of Corollary~\ref{cor:rkhs}. Its justification appeals to Theorem~\ref{thm:reg-main-error}, which requires $\mu\geq 6\lambda\geq 12\bar{\delta}^2/B$ and also $\lambda=\frac{C}{n}$. We take $\mu=6\lambda$ and $C=0.001$. We use the product kernel $k(x,x')=k(d,d')k(w,w')$ where $k(d,d')$ is a binary kernel and $k(w,w')$ is a radial basis function kernel with lengthscale set as the inverse of $dim(w) \times \text{var}(w)$. For the Nystr\"om approximation, we consider $S=100$ landmarks.

\textbf{Sparse linear function}. Consider the setting of Corollary~\ref{cor:a_cr}. There is one regularization parameter $\lambda=\gamma/8s$, where $s$ is the sparsity and $\gamma$ is the restricted eigenvalue. We take $\lambda=0.01$. We use the optimization routine of Proposition~\ref{prop:sparse-optimization-ell1} with $B=10$, similar to \cite{dikkala2020minimax}.

\textbf{Cross validation.} As an extension, we consider a general purpose and automated tuning procedure for $(\lambda,\mu)$. Take $\mu=6\lambda$.
Consider a grid $\Lambda$ of possible $\lambda$ values. We explicitly write regularization hyperparameters as arguments: $\hat{f}_a(\lambda)$ and $\hat{a}(\lambda,\mu)$.

\begin{algo}[Nested cross validation]
    Fix $\lambda \in \Lambda$. On a training set, fit a candidate $\hat{a}(\lambda,6\lambda)$. On a test set, partition observations into part one and part two. For each $\lambda'\in \Lambda$, calculate $\hat{f}_{\hat{a}(\lambda,6\lambda)}(\lambda')$, the optimal test function over observations in part one using regularization $\lambda'$ and the candidate $\hat{a}(\lambda,6\lambda)$. Evaluate the moment violation for $\hat{f}_{\hat{a}(\lambda,6\lambda)}(\lambda')$ on part two. Save it as $\textsc{score}(\lambda,\lambda')$. Reverse the roles of part one and part two, and save the average of the scores. 
 Repeat this exercise across $\lambda\in \Lambda$, then set $\lambda^*=\argmin_{\lambda} \max_{\lambda'} \textsc{score}(\lambda,\lambda')$, i.e. the value with the smallest out-of-sample maximum moment violation.
\end{algo}

See our earlier draft for simulations demonstrating the efficacy of this general and automated tuning procedure.

\subsection{Proof for random forest}

\begin{proof}[Proof of Proposition~\ref{prop:oracle}]
The loss function $-\ell(a, \cdot)$ is strongly convex in $f$ with respect to the $\|\cdot\|_{2,n}$ norm:
$
    -\frac{1}{2} D_{ff} \ell(a, f)[\nu, \nu] \geq \E_n[\nu(X)^2]
$.
The difference
$
    \ell(a, f) - \ell(a', f) = \E_n[(a(X)-a'(X))\cdot f(X)]
$
is an $\|a-a'\|_{2,n}$-Lipschitz function with respect to the $\ell_{2,n}$ norm by Cauchy-Schwarz inequality. Thus by \cite[Lemma 1]{syrgkanis},
$
    \|f_t - f_{t+1}\|_{2,n} \leq \|\bar{a}_{t-1} - \bar{a}_{t}\|_{2, n}.
$

By \cite[Proof of Theorem 1]{syrgkanis}, the cumulative regret of follow-the-leader is at most
$
    R(T) \leq \sum_{t=1}^T \left|\ell(a_t, f_t) - \ell(a_t, f_{t+1})\right|
$.
Since $\|a_t\|_{\infty}, \|f_t\|_{\infty} \leq 1$, each summand of the latter is upper bounded by
$
    \left|\E_n[m(Z; f_t - f_{t+1})]\right| + 3\|f_t - f_{t+1}\|_{1,n}.
$

We assume that the empirical operator $E_n[m(Z; f)]$ is bounded by $M_n$. Therefore
$
    \left|\E_n[m(Z; f_t - f_{t+1})]\right| \leq M_n \|f_t - f_{t+1}\|_{2,n}
$.
Overall,
\begin{equation*}
     \left|\ell(a_t, f_t) - \ell(a_t, f_{t+1})\right| \leq (M_n+3) \|f_t - f_{t+1}\|_{2,n} \leq (M_n+3) \|\bar{a}_{t-1} - \bar{a}_{t}\|_{2,n} \leq \frac{2\,(M_n+3)}{t}
\end{equation*}
where we use $\left|\bar{a}_{t-1}(X) - \bar{a}_{t}(X)\right|\leq \frac{2}{t}$, since $\|a\|_{\infty}\leq 1$. We conclude that
$
    R(T) \leq 2\,(M_n+3)\sum_{t=1}^T \frac{1}{t} = O(M_n\, \log(T))
$.

After $T=\Theta\left(\frac{M_n\, \log(1/\epsilon)}{\epsilon}\right)$ iterations, $f$ has regret of at most $\epsilon$. By standard results for convex-concave zero sum games, the average solutions $\bar{f}_T=\frac{1}{T}\sum_{t=1}^T f_t$ and $\bar{a}_T = \frac{1}{T} \sum_{t=1}^T a_t$ are an $\epsilon$-equilibrium.  Therefore $\bar{a}_T$ is an $\epsilon$-approximate solution to the minimax problem.    
\end{proof}

\subsection{Proof for RKHS}

\textbf{Kernel matrices and vectors.}
To begin, we formally define the kernel matrices and vectors that appear in the closed form. Their entries are given by
\begin{align*}
[K^{(1)}]_{ij}=~& k(X_i, X_j), & 
[K^{(2)}]_{ij}=~& k_m(X_i, X_j), &
[K^{(3)}]_{ij}=~& k_m(X_j, X_i), &
[K^{(4)}]_{ij}=~& k_{mm}(X_i, X_j) \\
[K_{xX}^{(1)}]_{j}=~& k(x, X_j), & 
[K_{xX}^{(2)}]_{j}=~& k_m(x, X_j), &
[K_{xX}^{(3)}]_{j}=~& k_m(X_j, x), &
[K_{xX}^{(4)}]_{j}=~& k_{mm}(x, X_j).
\end{align*}
All that remains is to define the kernels $k_m$ and $k_{mm}$, which combine the kernel $k$ with the functional $\theta:g\mapsto \mathbb{E}[m(g;Z)]$. For simplicity, let $Z=X$.

We define these additional kernels via the feature map representation. Given a kernel $k:\mathcal{X}\times\mathcal{X}\rightarrow\mathbb{R}$, its feature map is $\phi:x\mapsto k(x,\cdot)$. One can conceptualize $\phi(x)=\{\sqrt{\lambda_j}\varphi_j(x)\}$ where $\{\lambda_j\}$ and $\{\varphi_j(x)\}$ are the eigenvalues and eigenfunctions of the kernel, and $\{\varphi_j(x)\}$ is an orthonormal basis in $L_2$. Given a kernel $k$ and a functional $\theta:g\mapsto  \mathbb{E}[m(g;X)]$, we define the modified feature map $
\phi^{(m)}(x)=\{\sqrt{\lambda_j}m(x,\varphi_j)\}.
$ With these definitions, for $f\in\mcH$, $f(x)=\langle f,\phi(x) \rangle_{\mcH}$ and $m(x;f)= \langle f, \phi^{(m)}(x) \rangle_{\mathcal{H}}$. Finally,
\begin{align*}
k(x,x'):=~& \ldot{\phi(x)}{\phi(x')}_{\mcH}, &
k_{m}(x, x'):=~& \ldot{\phi(x)}{\phi^{(m)}(x')}, &
k_{mm}(x, x'):=~&\ldot{\phi^{(m)}(x)}{\phi^{(m)}(x')}.
\end{align*}
In summary, each kernel matrix and vector can be computed as linear combinations of kernel evaluations at pairs of points.

\begin{proposition}[Computing kernel matrices]\label{prop:kernel_matrices}
Let $X=(D,W)$ where $D$ is the treatment and $W$ is the covariate. Consider the average treatment effect functional $\theta:g\mapsto \mathbb{E}[(1,W)-g(0,W)]$. The induced kernels are
\begin{align*}
k(x, x') &=k((d,w),(d',w')),\quad 
k_m(x, x') =k((d,w),(1,w'))-k((d,w),(0,w')) \\
k_{mm}(x, x') &=k((1,w),(1,w'))-k((1,w),(0,w'))-k((0,w),(1,w'))+k((0,w),(0,w')).
\end{align*}
\end{proposition}
\begin{proof}
Observe that
\begin{align*}
    k_m(x, x') &=\langle \phi(x),\phi^{(m)}(x')\rangle_{\mathcal{H}}
    = \langle \phi(d,w), \phi(1,w')-\phi(0, w')\rangle_{\mathcal{H}} 
    =k((d,w),(1,w'))-k((d,w),(0,w')); \\
    k_{mm}(x, x') &=\langle \phi^{(m)}(x),\phi^{(m)}(x')\rangle_{\mathcal{H}}
    =\langle \phi(1,w)-\phi(0,w),\phi(1,w')-\phi(0,w') \rangle_{\mathcal{H}} \\
    &=k((1,w),(1,w'))-k((1,w),(0,w'))-k((0,w),(1,w'))+k((0,w),(0,w')).
\end{align*}
\end{proof}

\textbf{Maximizer closed form.}
We proceed in steps. We prove the result via a sequence of lemmas. To simplify the proofs, we introduce some additional operator notation. Given a kernel $k$, define the feature operator $\Phi$ with $i$th row $\phi(X_i)^{\top}$. Hence $K^{(1)}=\Phi\Phi^{\top}$, where $\Phi^{\top}$ is the adjoint of $\Phi$. Given a kernel and a functional, similarly define the operator $\Phi^{(m)}$ with $i$th row $\phi^{(m)}(X_i)^{\top}$. Finally define $\Psi$ as the operator with $2n$ rows that is constructed by concatenating $\Phi$ and $\Phi^{(m)}$:
$$
\Psi:=\begin{bmatrix} \Phi\\ \Phi^{(m)}\end{bmatrix},\quad K=\Psi\Psi^{\top}=\begin{bmatrix} \Phi \Phi^{\top} & \Phi (\Phi^{(m)})^{\top} \\ \Phi^{(m)}\Phi^{\top} & \Phi^{(m)}  (\Phi^{(m)})^{\top}  \end{bmatrix} =\begin{bmatrix} K^{(1)} & K^{(2)} \\ K^{(3)}  & K^{(4)} \end{bmatrix}.
$$

\begin{lemma}[Existence]\label{lemma:rep1}
There exists a coefficient $\hat{\gamma}_a\in\mathbb{R}^{2n}$ such that the maximizer $\hat{f}_a$ takes the form $\hat{f}_a=\Psi'\hat{\gamma}_a$.
\end{lemma}

\begin{proof}
Write the objective as 
$$
\mathcal{E}_{1}(f):=\frac{1}{n}\sum_{i=1}^n \langle f,\phi^{(m)}(X_i)\rangle_{\mathcal{H}}-a(X_i)\langle f,\phi(X_i)\rangle_{\mathcal{H}}-\langle f,\phi(X_i)\rangle_{\mathcal{H}}^2-\lambda\|f\|^2_{\mathcal{H}}.
$$
For an RKHS, evaluation is a continuous functional represented as the inner product with the feature map. Due to the ridge penalty, the stated objective is coercive and strongly convex with respect to $f$. Hence it has a unique maximizer $\hat{f}_a$ that obtains the maximum.

To lighten notation, we suppress the indexing of $\hat{f}_a$ by $a$ for the rest of this argument. Write $\hat{f}=\hat{f}_n+\hat{f}^{\perp}_n$ where $\hat{f}_n\in row(\Psi)$ and $\hat{f}_n^{\perp}\in null(\Psi)$. Substituting this decomposition of $\hat{f}$ into the objective, we see that
$
\mathcal{E}_{1}(\hat{f})=\mathcal{E}_{1}(\hat{f}_n)-\lambda \|\hat{f}_n^{\perp}\|^2_{\mathcal{H}}.
$
Therefore
$
\mathcal{E}_{1}(\hat{f})\leq \mathcal{E}_{1}(\hat{f}_n).
$
Since $\hat{f}$ is the unique maximizer, $\hat{f}=\hat{f}_n$.
\end{proof}

\begin{lemma}[Formula]\label{lemma:closed1}
The explicit formula for the coefficient is given by $\hat{\gamma}_a=\frac{1}{2}\Delta^{-}\left[V -U \mathbf{a}\right]$,
where
\begin{align*}
U :=~& \begin{bmatrix}K^{(1)} \\ K^{(3)} \end{bmatrix} \in \mathbb{R}^{2n \times n} &
\Delta:=~& U U' + n\lambda K\in\mathbb{R}^{2n\times 2n} &
V := \begin{bmatrix}K^{(2)} \\ K^{(4)} \end{bmatrix} \mathbf{1}_{n} \in \mathbb{R}^{2n}.
\end{align*}
\end{lemma}

\begin{proof}
Write the objective as
\begin{align*}
    \mathcal{E}_1(f)&= \frac{1}{n}\sum_{i=1}^n \langle f,\phi^{(m)}(X_i) \rangle_{\mathcal{H}} -\langle a,\phi(X_i)\rangle_{\mathcal{H}} \langle f,\phi(X_i)\rangle_{\mathcal{H}} -\langle f,\phi(X_i)\rangle_{\mathcal{H}}^2-\lambda \langle f,f \rangle_{\mathcal{H}}   \\
    &= \frac{1}{n} f' (\Phi^{(m)})'\mathbf{1}_{n} -f'\hat{T} a- f' \hat{T} f-\lambda f' f
\end{align*}
where $\hat{T}:=\frac{1}{n}\sum_{i=1}^n \phi(X_i)\otimes \phi(X_i)$. Appealing to Lemma~\ref{lemma:rep1},
\begin{align*}
    &\mathcal{E}_1(\gamma)= \frac{1}{n}\gamma'\Psi (\Phi^{(m)})'\mathbf{1}_{n} -\gamma'\Psi\hat{T} a- \gamma'\Psi \hat{T} \Psi'\gamma-\lambda \gamma'\Psi \Psi'\gamma \\
    &= \frac{1}{n} \gamma' \begin{bmatrix}K^{(2)} \\ K^{(4)} \end{bmatrix} \mathbf{1}_{n} -\frac{1}{n}\gamma'U\Phi a- \frac{1}{n}\gamma'UU'\gamma-\lambda \gamma'K\gamma
    = \frac{1}{n} \gamma' V -\frac{1}{n}\gamma'U\Phi a- \frac{1}{n}\gamma'UU'\gamma-\lambda \gamma'K\gamma.
\end{align*}
The first order condition yields
$
\frac{1}{n} V -\frac{1}{n}U\Phi a- \frac{2}{n}UU'\hat{\gamma}_a-2\lambda K\hat{\gamma}_a=0.
$
Hence
$    \hat{\gamma}_a
    =
\frac{1}{2}\left[UU'+n\lambda K\right]^{-}\left[V -U\Phi a\right]
$.
Finally, note that $\mathbf{a}=\Phi a \in\mathbb{R}^n$ with $[\mathbf{a}]_i=a(X_i)$.
\end{proof}

\begin{lemma}[Evaluation]\label{lemma:f}
To evaluate the adversarial maximizer, set $\hat{f}_a(x)=\begin{bmatrix}K^{(1)}_{xX} & K^{(2)}_{xX} \end{bmatrix}\hat{\gamma}_a$ where $[K_{xX}^{(1)}]_{j}= k(x, X_j)$ and $[K_{xX}^{(2)}]_{j}=k_m(x, X_j)$.
\end{lemma}

\begin{proof}
By Lemma~\ref{lemma:rep1}
$
    \hat{f}_a(x)=\langle \hat{f}_a, \phi(x)\rangle_{\mathcal{H}}  
    =\phi(x)'\Psi'\hat{\gamma}_a 
    =\begin{bmatrix}K^{(1)}_{xX} & K^{(2)}_{xX} \end{bmatrix}\hat{\gamma}_a.
$
\end{proof}

\begin{lemma}[Evaluation]\label{lemma:m}
To evaluate the functional applied to the adversarial maximizer, set $m(x,\hat{f}_a)=\begin{bmatrix}K^{(3)}_{xX} & K^{(4)}_{xX} \end{bmatrix} \hat{\gamma}_a$ where $[K_{xX}^{(3)}]_{j}=k_m(X_j, x)$ and $[K_{xX}^{(4)}]_{j}= k_{mm}(x, X_j)$.
\end{lemma}

\begin{proof}
By Lemma~\ref{lemma:rep1},
$
    m(x,\hat{f}_a)=\langle \hat{f}_a, \phi^{(m)}(x)\rangle_{\mathcal{H}}  
    =\phi^{(m)}(x)'\Psi'\hat{\gamma}_a 
    =\begin{bmatrix}K^{(3)}_{xX} & K^{(4)}_{xX} \end{bmatrix}\hat{\gamma}_a.
$
\end{proof}

\begin{proof}[Proof of Proposition~\ref{prop:closed1}]
    The result is immediate from the lemmas above.
\end{proof}

\textbf{Minimizer closed form.}
Once again, we prove the result via a sequence of lemmas.

\begin{lemma}[Existence]\label{lemma:rep2}
There exists a coefficient $\hat{\beta}\in\mathbb{R}^n$ such that the minimizer $\hat{a}$ takes the form $\hat{a}=\Phi'\hat{\beta}$.
\end{lemma}

\begin{proof}
Observe that by Lemmas~\ref{lemma:closed1},~\ref{lemma:f}, and~\ref{lemma:m},
\begin{align*}
    \hat{f}_a(x)&= \frac{1}{2}\begin{bmatrix}K^{(1)}_{xX} & K^{(2)}_{xX}\end{bmatrix}\Delta^{-}\left[V -U\Phi a\right],\quad 
    m(x;\hat{f}_a)=\frac{1}{2} \begin{bmatrix}K^{(3)}_{xX} & K^{(4)}_{xX} \end{bmatrix} \Delta^{-}\left[V -U\Phi a\right] \\
   \|\hat{f}_a\|_{\mathcal{H}}^2 &=\hat{\gamma}_a'\Psi \Psi'\hat{\gamma}_a
   = \frac{1}{4} \left[V -U\Phi a\right]'
   \Delta^{-}
   K  
   \Delta^{-}\left[V -U\Phi a\right].
\end{align*}
Write the objective as 
\begin{align*}
    \mathcal{E}_{2}(a)&=\frac{1}{n}\sum_{i=1}^n m(X_i;\hat{f}_a)-\langle a,\phi(X_i)\rangle_{\mathcal{H}} \hat{f}_a(X_i)- \hat{f}_a(X_i)^2-\lambda\|\hat{f}_a\|^2_{\mathcal{H}}+\mu\|a\|^2_{\mathcal{H}}
\end{align*}
where the various terms involving $\hat{f}_a$ only depend on $a$ in the form $\Phi a$. Due to the ridge penalty, the stated objective is coercive and strongly convex with respect to $a$. Hence it has a unique maximizer $\hat{a}$ that obtains the maximum.

Write $\hat{a}=\hat{a}_n+\hat{a}^{\perp}_n$ where $\hat{a}_n\in row(\Phi)$ and $\hat{a}_n^{\perp}\in null(\Phi)$. Substituting this decomposition of $\hat{a}$ into the objective, we see that
$
\mathcal{E}_{2}(\hat{a})=\mathcal{E}_{2}(\hat{a}_n)+\mu \|\hat{a}_n^{\perp}\|^2_{\mathcal{H}}.
$
Therefore
$
\mathcal{E}_{2}(\hat{a})\geq \mathcal{E}_{2}(\hat{a}_n).
$
Since $\hat{a}$ is the unique minimizer, $\hat{a}=\hat{a}_n$.
\end{proof}

\begin{lemma}[Formula]\label{lemma:closed2}
The explicit formula for the coefficient is given by $\hat{\beta}=\big(A'  \Delta^{-} A +4n\mu\cdot K^{(1)}\big)^-A'\Delta^{-}V$,
    where $A:= U K^{(1)}$.
\end{lemma}

\begin{proof}
By Lemma~\ref{lemma:rep1},
\begin{align*}
    \mathcal{E}_1(f)&=\frac{1}{n} f' (\Phi^{(m)})'\mathbf{1}_{n} -f'\hat{T} a- f' \hat{T} f-\lambda f' f 
    =\frac{1}{n} \gamma'\Psi (\Phi^{(m)})'\mathbf{1}_{n} -\gamma'\Psi\hat{T} a- \gamma'\Psi \hat{T} \Psi'\gamma-\lambda\gamma'\Psi \Psi'\gamma
\end{align*}
with first order condition
$$
0=\frac{1}{n}\Psi (\Phi^{(m)})'\mathbf{1}_{n} -\Psi\hat{T} a- 2\Psi \hat{T} \Psi'\hat{\gamma}_a-2\lambda\Psi \Psi'\hat{\gamma}_a=\frac{1}{n}\Psi (\Phi^{(m)})'\mathbf{1}_{n} -\Psi\hat{T} a- 2\Psi \hat{T} \hat{f}_a-2\lambda\Psi \hat{f}_a.
$$
Hence
$
2\Psi(\hat{T}+\lambda I )\hat{f}_a=\Psi \left[ \frac{1}{n}(\Phi^{(m)})'\mathbf{1}_{n} -\hat{T} a\right]
$
and therefore multiplying both sides by $\hat{\gamma}_a'$ gives
$
2\hat{f}_a'(\hat{T}+\lambda I )\hat{f}_a=\hat{f}_a' \left[ \frac{1}{n}(\Phi^{(m)})'\mathbf{1}_{n} -\hat{T} a\right]
$.
Thus we have
\begin{align*}
    \mathcal{E}_2(a)&=\frac{1}{n}\hat{f}_a' (\Phi^{(m)})' \mathbf{1}_{n} -\hat{f}_a'\hat{T} a- \hat{f}_a' \hat{T} \hat{f}_a-\lambda \hat{f}_a' \hat{f}_a+\mu a'a 
    =2 \hat{f}_a' (\hat{T}+\lambda I) \hat{f}_a - \hat{f}_a' \hat{T} \hat{f}_a-\lambda \hat{f}_a' \hat{f}_a+\mu a'a \\
    &= \hat{f}_a' (\hat{T}+\lambda I) \hat{f}_a +\mu a'a  
    = \hat{f}_a' \hat{T} \hat{f}_a + \lambda \hat{f}_a'\hat{f}_a +\mu a'a.
\end{align*}
By Lemma~\ref{lemma:rep1},
\begin{align*}
    {\cal E}_2(a)&= \frac{1}{n}\hat{\gamma}_a' \Psi \Phi' \Phi \Psi' \hat{\gamma}_a + \lambda \hat{\gamma}_a' \Psi \Psi' \hat{\gamma}_a + \mu a'a
    = \frac{1}{n}\hat{\gamma}_a' U U' \hat{\gamma}_a + \lambda \hat{\gamma}_a' K \hat{\gamma}_a + \mu a'a
    = \frac{1}{n}\hat{\gamma}_a' \Delta \hat{\gamma}_a + \mu a'a.
\end{align*}
By Lemma~\ref{lemma:closed1},
\begin{align*}
    {\cal E}_2(a)&= \frac{1}{4n}(V - U\Phi a)' \Delta^- \Delta \Delta^-(V - U\Phi a) + \mu a'a
    = \frac{1}{4n}(V - U\Phi a)' \Delta^-(V - U\Phi a) + \mu a'a.
\end{align*}
By Lemma~\ref{lemma:rep2},
\begin{align*}
    &\mathcal{E}_2(\beta)= \frac{1}{4n}(V - U\Phi \Phi' \beta)' \Delta^-(V - U\Phi \Phi'\beta) + \mu \beta'\Phi \Phi' \beta \\
    &= \frac{1}{4n}(V - UK^{(1)} \beta)' \Delta^-(V - UK^{(1)}\beta) + \mu \beta' K^{(1)} \beta
    = \frac{1}{4n}(V - A \beta)' \Delta^-(V - A\beta) + \mu \beta' K^{(1)} \beta.
\end{align*}
The first order condition then becomes
\begin{align*}
    - \frac{1}{2n} A' \Delta^- (V - A\beta) + 2\mu K^{(1)} \beta = 0 \Leftrightarrow 
     (A'\Delta^- A  + 4n\mu K^{1}) \beta = A'\Delta^- V.
\end{align*}
Solving for $\beta$ yields the desired solution.
\end{proof}

\begin{lemma}[Evaluation]
To evaluate the minimizer, set $\hat{a}(x) = K_{xX}^{(1)} \hat{\beta}$.
\end{lemma}
\begin{proof}
By Lemma~\ref{lemma:rep2}
$
    \hat{a}(x)=\langle \hat{a}, \phi(x)\rangle_{\mathcal{H}}  
    =\phi(x)'\Phi'\hat{\beta} 
    =K^{(1)}_{xX}\hat{\beta}
$.
\end{proof}

\begin{proof}[Proof of Proposition~\ref{prop:closed2}]
    The result is immediate from the lemmas above.
\end{proof}

\textbf{Nystr\"om approximation.}
Computation of kernel methods may be demanding due to the inversions of matrices that scale with $n$ such as $\Delta\in\mathbb{R}^{2n\times 2n}$. One solution is Nystr\"om approximation. We now provide alternative expressions for $(\hat{f}_a,\hat{a})$ that lend themselves to Nystr\"om approximation, then describe the procedure.

\begin{lemma}[Maximizer sufficient statistics]\label{lemma:n1}
The adversarial maximizer may be expressed as $\hat{f}_a=\frac{1}{2}(\hat{T}+\lambda I)^{-1}(\hat{\mu}^{(m)}-\hat{T}a)$ where $\hat{\mu}^{(m)}:=\frac{1}{n}\sum_{i=1}^n \phi^{(m)}(X_i)$ and $\hat{T}:=\frac{1}{n}\sum_{i=1}^n \phi(X_i)\otimes \phi(X_i)$.
\end{lemma}

\begin{proof}
Recall the loss
\begin{align*}
    \mathcal{E}_1(f)=\frac{1}{n} f' (\Phi^{(m)})'\mathbf{1}_{n} -f'\hat{T} a- f' \hat{T} f-\lambda f' f 
    = f'(\hat{\mu}^{(m)}-\hat{T} a)- f'( \hat{T}+\lambda I) f.
\end{align*}
Informally, we see that the first order condition must be
$
2(\hat{T}+\lambda I)\hat{f}_a=\hat{\mu}^{(m)}  -\hat{T} a.
$
See \cite[Proof of Proposition 2]{de2005risk} for the formal derivation, which incurs additional notation. Rearranging yields the desired result.
\end{proof}

\begin{lemma}[Minimizer sufficient statistics]\label{lemma:n2}
The minimizer may be expressed as $\hat{a}=[\hat{T}(\hat{T}+\lambda I)^{-1}\hat{T}+4\mu I]^{-1}\hat{T}(\hat{T}+\lambda I)^{-1}\hat{\mu}^{(m)}$.
\end{lemma}

\begin{proof}
By Lemma~\ref{lemma:n1},
\begin{align*}
     \mathcal{E}_2(a)&=\frac{1}{n}\hat{f}_a' (\Phi^{(m)})' \mathbf{1}_{n} -\hat{f}_a'\hat{T} a- \hat{f}_a' \hat{T} \hat{f}_a-\lambda \hat{f}_a' \hat{f}_a+\mu a'a 
     =\hat{f}_a' (\hat{\mu}^{(m)}-\hat{T} a)- \hat{f}_a' (\hat{T}+\lambda I) \hat{f}_a+\mu a'a \\
     &=\hat{f}_a' [2(\hat{T}+\lambda I)\hat{f}_a]- \hat{f}_a' (\hat{T}+\lambda I) \hat{f}_a+\mu a'a  
     =\hat{f}_a' (\hat{T}+\lambda I) \hat{f}_a+\mu a'a \\
     &=\frac{1}{4}[(\hat{T}+\lambda I)^{-1}(\hat{\mu}^{(m)}-\hat{T}a)]' (\hat{T}+\lambda I) [(\hat{T}+\lambda I)^{-1}(\hat{\mu}^{(m)}-\hat{T}a)]+\mu a'a.
\end{align*}
Informally, we see that the first order condition must be
$$
0=\frac{1}{4}[-2\hat{T}(\hat{T}+\lambda I)^{-1}](\hat{T}+\lambda I)[(\hat{T}+\lambda I)^{-1}(\hat{\mu}^{(m)}-\hat{T}a)]+2\mu a.
$$
See \cite[Proof of Proposition 2]{de2005risk} for the formal way of deriving the first order condition, which incurs additional notation. Simplifying,
$
0=\hat{T}(\hat{T}+\lambda I)^{-1}(\hat{T}a-\hat{\mu}^{(m)})+4\mu a.
$
Rearranging yields the desired result.
\end{proof}

The closed form solution in Lemma~\ref{lemma:n2} is in terms of the sufficient statistics
$
\hat{T}=\frac{1}{n}\sum_{i=1}^n \phi(X_i)\otimes \phi(X_i)$ and $\hat{\mu}^{(m)}=\frac{1}{n}\sum_{i=1}^n \phi^{(m)}(X_i)
$.
Nystr\"om approximation is a way to approximate these sufficient statistics. In particular, the Nystr\"om approach uses the substitution
$
\phi(x)\mapsto \tilde{\phi}(x)= (K_{\mathcal{S}\mathcal{S}}^{(1)})^{-\frac{1}{2}}K^{(1)}_{\mathcal{S}x}
$, where
$\mathcal{S}$ is a subset of $S=|\mathcal{S}|\ll n$ observations called landmarks. $K^{(1)}_{\mathcal{S}\mathcal{S}}\in\mathbb{R}^{S\times S}$ is defined such that $[K^{(1)}_{\mathcal{S}\mathcal{S}}]_{ij}=k(X_i,X_j)$ for $i,j\in\mathcal{S}$. Similarly, $K^{(1)}_{\mathcal{S}x}\in\mathbb{R}^S$ is defined such that $[K^{(1)}_{\mathcal{S}x}]_i=k(X_i,x)$ for $i\in\mathcal{S}$. We conduct a similar substitution for $\phi^{(m)}$. For example, for ATE,
$
\phi^{(m)}(x)=\phi(1,w)-\phi(0,w) \mapsto \tilde{\phi}(1,w)-\tilde{\phi}(0,w)=\tilde{\phi}^{(m)}(x)
$
where $\tilde{\phi}(x)$ is defined above. In summary, the approximate sufficient statistics are 
$
\tilde{T}=\frac{1}{n}\sum_{i=1}^n \tilde{\phi}(X_i)\otimes \tilde{\phi}(X_i) \in\mathbb{R}^{S\times S}$ and $\tilde{\mu}^{(m)}=\frac{1}{n}\sum_{i=1}^n \tilde{\phi}^{(m)}(X_i) \in\mathbb{R}^S.
$

\section{Simulated and real data details}\label{sec:sim_detail}

\subsection{Additional results: Baseline design}

A baseline design showcases how eliminating sample splitting may improve precision---a simple point with practical consequences for applied statistics, that underscores the importance of Theorems~\ref{thm:debias-nocross-stability} and~\ref{thm:debias-nocross}.

In a baseline setting with $dim(W)=10$, every variation of our estimator achieves nominal coverage when the sample size is sufficiently large. We document performance in 100 simulations for sample sizes $n\in\{100,200,500,1000,2000\}$, which are representative for empirical social scientific research.

\begin{table}[H]
    \centering
    \begin{tabular}{|c|c|cc|ccccc|}
\hline 
 \multirow{2}{*}{$n$}  & Estimator & \multicolumn{2}{|c|}{Prop. score} & \multicolumn{5}{|c|}{Adversarial} \\
 & Function space   & logistic & R.F. &  sparse & RKHS & Nystrom & R.F. & N.N. \\  
    \hline 
       &  coverage & 95 & 95 & 97 & 96 & 91 & 93 & 53  \\
 $100$     &  bias & -1 & 5 & 1 & 2 & 3 & 8 & -11  \\
         &  length & 138 & 99 & 118 & 113 & 93 & 106 & 34  \\
        \hline 
       &  coverage & 95 & 92 &  95 & 96 & 85 & 92 & 57  \\
 $200$     &  bias & 2 & 2 &  0 & 4 & 0 & 4 & -6  \\
         &  length & 80 & 65 &  77 & 98 & 63 & 68 & 33  \\
        \hline 
       &  coverage & 95 & 95 &  97 & 98 & 93 & 97 & 82  \\
 $500$     &  bias & -1 & 0 & -3 & 1 & -2 & 0 & -5  \\
         &  length & 48 & 42 &  46 & 64 & 41 & 44 & 34  \\
        \hline 
        &  coverage & 95 & 87 &  91 & 99 & 88 & 86 & 90  \\
 $1000$     &  bias & 0 & 3 &  1 & 1 & 3 & 3 & 0  \\
          &  length & 34 & 30 &  32 & 42 & 29 & 30 & 32  \\
        \hline 
        &  coverage & 96 & 92 &  95 & 97 & 91 & 89 & 94  \\
 $2000$     &  bias & 0 & 2 &  1 & 1 & 2 & 2 & 1  \\
          &  length & 23 & 21 &  22 & 28 & 21 & 21 & 23  \\
        \hline 
    \end{tabular}
    \caption{Simple design, sample splitting $\{dim(W)=10\}$. Values are multiplied by $10^2$.}
    \label{tab:base}
\end{table}

\begin{table}[H]
    \centering
    \begin{tabular}{|c|c|cc|ccccc|}
    \hline 
 \multirow{2}{*}{$n$}  & Estimator & \multicolumn{2}{|c|}{Prop. score} &  \multicolumn{5}{|c|}{Adversarial} \\
 & Function space   & logistic & R.F. &  sparse & RKHS & Nystrom & R.F. & N.N. \\  
    \hline 
       &  coverage & 92 & - &  94 & 95 & 91 & 92 & 50  \\
 $100$     &  bias & 0 & - &  1 & 0 & 1 & 7 & -8  \\
         &  length & 94 & - &  94  & 96 & 86 & 88 & 32 \\
        \hline 
       &  coverage & 92 & - &  95 & 94 & 88 & 93 & 58  \\
 $200$     &  bias & 1 & - & -1 & 1 & 0 & 2 & -6  \\
         &  length & 68 & -  & 68 & 73 & 58 & 65 & 34 \\
        \hline 
       &  coverage & 92 & -  & 91 & 92 & 87 & 90 & 77  \\
 $500$     &  bias & -2 & -  & -3 & -3 & -2 & -2 & -6  \\
         &  length & 41 & - & 41 & 43 & 37 & 40 & 33 \\
        \hline 
        &  coverage & 86 & - & 85 & 97 & 83 & 85 & 84  \\
 $1000$     &  bias & 0 & -  & 0 & 0 & 1 & 1 & 0  \\
          &  length & 27 & -  & 26 & 32 & 25 & 26 & 28 \\
        \hline 
        &  coverage & 94 & - & 92 & 93 & 90 & 92 & 93  \\
 $2000$     &  bias & -1 & -  & 0 & 0 & 0 & 1 & 0  \\
          &  length & 20 & -  & 19 & 20 & 19 & 19 & 20 \\
        \hline 
    \end{tabular}
    \caption{Simple design, no sample splitting $\{dim(W)=10\}$. Values are multiplied by $10^2$.}
    \label{tab:base_no}
\end{table}

Tables~\ref{tab:base} and~\ref{tab:base_no} present results. We find that every version of our adversarial estimator (i-v) achieves nominal coverage with and without sample splitting, as long as the sample size is large enough. In particular, (i-iv) achieve nominal coverage with $n=100$, but (v) requires $n=2000$. The propensity score estimators (vi-vii) also achieve nominal coverage. Comparing entries across tables, we see that eliminating sample splitting always reduces the confidence interval length. For example, with $n=100$, (i) has length 1.18 with sample splitting and length 0.94 without sample splitting; eliminating sample splitting reduces the confidence interval length by 20\%. Another general trend is that our estimators have the shortest confidence intervals among those implemented---tied with the parametric estimator (vi), and shorter than (vii). For example, with $n=100$, (i) has length 0.94 without sample splitting while the lengths of (vi-vii) are at best 0.94 and 0.99, respectively.

\subsection{Simulation designs}

We generate one draw from the high dimensional setting as follows. Define the vector $\beta\in\R^p$ such that $\beta_j=j^{-2}$ and the matrix $\Sigma\in \R^{p\times p}$ such that $\Sigma_{ii}=1$ and $\Sigma_{ij}=0.5\cdot 1_{|i-j|=1}$ for $i\neq j$. We draw the covariates $W\sim N(0,\Sigma)$, the treatment $D\sim \text{Bernoulli}(0.05+0.5\Lambda(W^\top\beta))$ where $\Lambda(\cdot)$ is the logistic function, and then calculate the outcome $Y=2.2D+1.2W^\top\beta +DW_1+\epsilon$ where $\epsilon\sim N(0,1)$. To convey a high dimensional setting, we take $n=100$ and $p=100$.

We generate one draw from the highly nonlinear setting as follows. We draw the covariates $W\sim N(0,\Sigma)$, the treatment $D\sim \text{Bernoulli}(0.1+0.8\cdot 1_{W_1>0})$, and then calculate $Y=2.2D+1.2\cdot 1_{W_1>0} +DW_1+\epsilon$ where $\epsilon\sim N(0,1)$. We take $n=1000$ and $p=10$. 

To generate one draw from the baseline design, we modify the high dimensional design to have $p=10$. We vary $n\in\{100,200,500,1000,2000\}$.

\subsection{Heterogeneous effects by political environment}

We follow variable definitions of \cite{karlan2007does}. The outcome $Y$ is dollars donated in Figure~\ref{fig:dollars}, and an indicator of whether the household donated in Figure~\ref{fig:whether}. The treatment $D$ indicates receiving a 1:1 matching grant as part of a direct mail solicitation. For simplicity, we exclude other treatment arms, namely 2:1 and 3:1 matching grants. A ``red'' unit of geography is one in George W. Bush won more than half of the votes in the 2004 presidential election. The raw covariates $W \in\mathbb{R}^{15}$ are political environment, previous contributions, race, age, household size, income, home-ownership, education, and urban status. Specifically, the initial covariate $W_1$ indicates whether the political environment is a red state. 

We exclude 4147 observations with extreme propensity scores based on these covariates, yielding $n=21712$ observations for analysis. Specifically, we dropped observations with propensity scores outside of $[0.1,0.9]$ based on simple logistic models for the treatment propensity score and red state propensity score. For tractability with this large sample size, we approximate our adversarial RKHS estimator using our adversarial Nystr\"om estimator with 1000 landmarks.

\begin{proposition}[Riesz representation and mean square continuity]
    To lighten notation, define $g_0(D,A,V)=\mathbb{E}(Y|D,A,V)$ and
$$
\mathbb{E}\{m(Z;g)\}=\mathbb{E}[\{g(1,1,V)-g(0,1,V)\}-\{g(1,0,V)-g(0,0,V)\}].
$$
Then $
\mathbb{E}\{m(Z;g)\}=\mathbb{E}\{a_0(D,A,V)g(D,A,V)\}
$
where 
\begin{align*}
a_0(D,A,V)
&=\frac{1_{A=1}}{\mathbb{P}(A=1|V)}\left\{\frac{1_{D=1}}{\mathbb{P}(D=1|A=1,V)}-\frac{1_{D=0}}{\mathbb{P}(D=0|A=1,V)}\right\}\\
&\quad -\frac{1_{A=0}}{\mathbb{P}(A=0|V)}\left\{\frac{1_{D=1}}{\mathbb{P}(D=1|A=0,V)}-\frac{1_{D=0}}{\mathbb{P}(D=0|A=0,V)}\right\}\\
&=\left\{\frac{1_{A=1}}{\mathbb{P}(A=1|V)}-\frac{1_{A=0}}{\mathbb{P}(A=0|V)}\right\}\left\{\frac{1_{D=1}}{\mathbb{P}(D=1|A,V)}-\frac{1_{D=0}}{\mathbb{P}(D=0|A,V)}\right\}
\end{align*}
and the functional is mean square continuous when $\mathbb{P}(D=1|A,V)$ and $\mathbb{P}(A=1|V)$ are bounded away from zero and one almost surely.
\end{proposition}

    \begin{proof}
        We directly generalize the argument for average treatment effect in Proposition~\ref{prop:continuity}.
    \end{proof}

\section{Sparse linear case details}

\subsection{Fast rate}

\begin{proof}[Proof of Corollary~\ref{cor:sparse-linear}]
    The critical radius $\delta_n$ is of order $O\left(\sqrt{\frac{s\log(p\,n)}{n}}\right)$. The $\epsilon$-covering number of such a function class is of order $N_n(\epsilon; \mcF)=O\left(\binom{p}{s} \left(\frac{b}{\epsilon}\right)^{s}\right)\leq O\left(\left(\frac{p\,b}{\epsilon}\right)^s\right)$, since it suffices to choose the support of the coefficients and then place a uniform $\epsilon$-grid on the support. Thus~\eqref{eqn:metric-entropy-critical} is satisfied for $\delta=O\left(\sqrt{\frac{s\log(p\, b)\,\log(n)}{n}}\right)$. Moreover, observe that if $m(Z;f)$ is $L$-Lipschitz in $f$ with respect to the $\ell_{\infty}$ norm, then the covering number of $m\circ \mcF$ is also of the same order. Finally, we appeal to Corollary~\ref{cor:main-error}.
\end{proof}

To simplify notation, we let $\hat{a}=\check{a}$ in this subsection.

\begin{proof}[Proof of Proposition~\ref{prop:reg-main-error-2}]
By the definition of $\hat{a}$:
$
    0\leq \sup_{f} \Psi_n(\hat{a}, f) \leq \sup_{f} \Psi_n(a_0, f) + \lambda \left(\|a_0\|_{\mcA} - \|\hat{a}\|_{\mcA}\right).
$
Let
$
\delta_{n, \zeta}=\max_{i}\left\{\mcR(\mcF^i) + \mcR(m\circ \mcF^i)\right\} + c_0 \sqrt{\frac{\log(c_1/\zeta)}{n}}
$
for some universal constants $c_0,c_1$. By \cite[Theorems~26.5 and 26.9]{shalev2014understanding}, since $\mcF^i$ is a symmetric class and since $\|a_0\|_{\infty} \leq 1$, with probability $1-\zeta$,
$
    \forall f\in \mcF^i: \left|\Psi_n(a_0, f) - \Psi(a_0, f)\right| \leq \delta_{n,\zeta}.
$
Since $\Psi(a_0, f)=0$ for all $f\in \mcF$, with probability $1-\zeta$,
$$
0\leq \sup_{f} \Psi_n(a_0, f) + \lambda \left(\|a_0\|_{\mcA} - \|\hat{a}\|_{\mcA}\right)  \leq \delta_{n,\zeta} + \lambda \left(\|a_0\|_{\mcA} - \|\hat{a}\|_{\mcA}\right)
$$
which implies
$
    \|\hat{a}\|_{\mcA} \leq \|a_0\|_{\mcA} + \delta_{n,\zeta}/\lambda.
$

Let $B_{n,\lambda,\zeta} = (\|a_0\|_{\mcH} + \delta_{n,\zeta}/\lambda)^2$, $\mcA_B\cdot \mcF^i := \{a\cdot f: a\in \mcA_B, f\in \mcF^i\}$ and
$$
\epsilon_{n,\lambda, \zeta}=\max_{i}\left\{\mcR(\mcA_{B_{n,\lambda,\zeta}}\cdot \mcF^i) + \mcR(m\circ \mcF^i)\right\} + c_0 \sqrt{\frac{\log(c_1/\zeta)}{n}}
$$
for some universal constants $c_0,c_1$. Then again by \cite[Theorems~26.5 and 26.9]{shalev2014understanding},
$
    \forall a\in \mcA_{B_{n,\lambda,\zeta}}, f\in \mcF_U^i: \;\left|\Psi_n(a, f) - \Psi(a, f)\right| \leq \epsilon_{n,\lambda, \zeta}.
$
By a union bound over the $d$ function classes composing $\mcF$, we have that with probability $1-2\zeta$
$
    \sup_{f\in \mcF} \Psi_n(a_0, f) \leq \sup_{f\in \mcF} \Psi(a_0, f) + \delta_{n,\zeta/d} = \delta_{n,\zeta/d}
$
and
$
    \sup_{f\in \mcF} \Psi_n(\hat{a}, f) \geq \sup_{f\in \mcF} \Psi(\hat{a}, f) - \epsilon_{n,\lambda, \zeta/d}.
$

If $\|\hat{a}-a_0\|_2\leq \delta_{n,\zeta}$, the result follows immediately. Thus we consider the case when $\|\hat{a}-a_0\|_2\geq \delta_{n,\zeta}$. Since, by assumption, for any $a\in \mcA_{B}$ with $\|a-a_0\|\geq \delta_{n,\zeta}$ it holds that $\frac{a_0-a}{\|a_0-a\|_2}\in \spanF_{\kappa}(\mcF)$, we have $\frac{a_0 -\hat{a}}{\|a_0 - \hat{a}\|_2}=\sum_{i=1}^p w_i f_i$, with $p<\infty$, $\|w\|_1\leq \kappa$ and $f_i\in \mcF$. Thus
\begin{align*}
    \sup_{f\in \mcF} \Psi(\hat{a}, f) \geq~& \frac{1}{\kappa} \sum_{i=1}^p w_i \Psi(\hat{a}, f_i) = \frac{1}{\kappa} \Psi\left(\hat{a}, \sum_i w_i f_i\right)
    = \frac{1}{\kappa} \frac{1}{\|\hat{a}-a_0\|_2}\Psi(\hat{a}, a_0 - \hat{a})\\
    =~& \frac{1}{\kappa} \frac{1}{\|\hat{a}-a_0\|_2}\E[(a_0(X) - \hat{a}(X))^2]
    = \frac{1}{\kappa} \|\hat{a}-a_0\|_2.
\end{align*}
Combining results, with probability $1-2\zeta$, 
$
    \|\hat{a}-a_0\|_2 \leq \kappa\, \left(\{\epsilon_{n, \lambda, \zeta/d} + \delta_{n,\zeta/d} + \lambda \left(\|a_0\|_{\mcA} - \|\hat{a}\|_{\mcA}\right)\right\}.
$
Moreover, since functions in $\mcA$ and $\mcF$ are bounded in $[-b,b]$, we have that the function $a\cdot f$ is $b$-Lipschitz with respect to the vector of functions $(a, f)$. Thus we can apply a vector version of the contraction inequality \cite{maurer2016vector} to get that
$
\mcR(\mcA_{B_{n,\lambda, z}}\cdot \mcF^i) \leq 2\, \left\{\mcR(\mcA_{B_{n,\lambda, z}}) + \mcR(\mcF^i)\right\}.
$
Finally, since $\mcA$ is star-convex, we have that
$
    \mcR(\mcA_{B_{n,\lambda, z}}) \leq \sqrt{B_{n,\lambda, z}}\,\mcR(\mcA_1),
$
leading to the final bound of
\begin{multline*}
    \|\hat{a}-a_0\|_2 \leq \kappa \left[ 2\left\{\|a_0\|_{\mcA} + \delta_{n,\zeta}/\lambda\right\} \mcR(\mcA_1) + 2\, \max_{i} \left\{\mcR(\mcF^i) + \mcR(m\circ \mcF^i)\right\} \right]\\ + \kappa \left(c_0\sqrt{\frac{\log(c_1\, d/\zeta)}{n}} + \lambda \left(\|a_0\|_{\mcA}-\|\hat{a}\|_{\mcA}\right)\right).
\end{multline*}
Since $\lambda \geq \delta_{n,\zeta}$, we arrive at the desired result.
\end{proof}

\begin{proof}[Proof of Corollary~\ref{cor:sparse-linear-reg-ell1}] We will apply Proposition~\ref{prop:reg-main-error-2} with $\mcA=\{\ldot{\theta}{\cdot}: \|\theta\|\leq B\}$ and norm $\|\ldot{\theta}{\cdot}\|_{\mcA}=\|\theta\|_1$. By standard results on the Rademacher complexity of linear function classes \cite[Lemma~26.11]{shalev2014understanding}, 
$\mcR(\mcA_B)\leq B\sqrt{\frac{2\log(2\, p)}{n}}\max_{x\in \mcX} \|x\|_{\infty}$ and $\mcR(\mcF^i), \mcR(m\circ \mcF^i)\leq \sqrt{\frac{2\log(2)}{n}}\max_{x\in \mcX} \|x\|_{\infty}$ . The latter holds since each $\mcF^i$, and therefore also $m\circ \mcF^i$, contains only two elements; then invoke Masart's lemma.
Thus we can apply Proposition~\ref{prop:reg-main-error-2} with 
$$
\delta_{n,\zeta}:=2\max_{i} \left\{\mcR(\mcF^i) + \mcR(m\circ \mcF^i)\right\} + c_0 \sqrt{\frac{\log(c_1\, p/\zeta)}{n}} = c_2 \sqrt{\frac{\log(c_1\, p/\zeta)}{n}}
$$
for some universal constant $c_2$. Hence, 
we derive from Proposition~\ref{prop:reg-main-error-2} that with probability $1-\zeta$,
$\|\hat{a}\|_{\mcA} \leq \|a_0\|_{\mcA} + \delta_{n,\zeta}/\lambda$ which implies 
$
\|\hat{\theta}\|_1 \leq \|\theta_0\|_1 + \delta_{n,\zeta}/\lambda.
$

We will now verify the span condition of Proposition~\ref{prop:reg-main-error-2} and derive the constant $\kappa$. Consider any $\hat{a}=\ldot{\hat{\theta}}{\cdot}\in \mcA_{B_{n,\lambda,\zeta}}$ and let $\nu=\hat{\theta} - \theta_0$. Then
\begin{align*}
    \delta_{n,\zeta}/\lambda + \|\theta_0\|_1 \geq \|\hat{\theta}\|_1 =\|\theta_0 + \nu\|_1 = \|\theta_0 + \nu_S\|_1+\|\nu_{S^c}\|_1 \geq \|\theta_0\|_1 - \|\nu_S\|_1 + \|\nu_{S^c}\|_1.
\end{align*}
Hence
$
    \|\nu_{S^c}\|_1\leq \|\nu_S\|_1 + \delta_{n,\zeta}/\lambda
$
and $\nu$ lies in the restricted cone for which the restricted eigenvalue of $V$ holds. 
Since $|S|=s$,
\begin{equation*}
    \|\nu\|_1 \leq 2 \|\nu_S\|_1 + \delta_{n,\zeta}/\lambda \leq 2\sqrt{s} \|\nu_S\|_2  + \delta_{n,\zeta}/\lambda \leq 2\sqrt{s}\|\nu\|_2  + \delta_{n,\zeta}/\lambda \leq 2 \sqrt{\frac{s}{\gamma} \nu^\top  V \nu}  + \delta_{n,\zeta}/\lambda.
\end{equation*}
Moreover, observe that
$
    \|\hat{a}-a_0\|_2 = \sqrt{\E[ \ldot{\nu}{x}^2 ]} = \sqrt{\nu^\top  V \nu} 
$. Thus we have
$
    \frac{\hat{a}(x)-a_0(x)}{\|\hat{a}-a_0\|_2} = \sum_{i=1}^p \frac{\nu_i}{\sqrt{\nu^\top V\nu}} x_i.
$
As a consequence, for any $\hat{a}\in \mcA_{B_{n,\lambda,\zeta}}$, we can write $\frac{\hat{a}-a_0}{\|\hat{a}-a_0\|_2}$ as $\sum_{i=1}^p w_i f_i$, with $f_i\in \mcF$ and
\begin{equation*}
    \|w\|_1 = \frac{\|\nu\|_1}{\sqrt{\nu^\top V \nu}} \leq 2\sqrt{\frac{s}{\gamma}} + \frac{\delta_{n,\zeta}}{\lambda} \frac{1}{\|\hat{a}-a_0\|_2}.
\end{equation*}
In summary, $\frac{\hat{a}-a_0}{\|\hat{a}-a_0\|_2} \in \spanF_{\kappa}(\mcF)$ for $\kappa=2\sqrt{\frac{s}{\gamma}} + \frac{\delta_{n,\zeta}}{\lambda} \frac{1}{\|\hat{a}-a_0\|_2}$. 

Next, by the triangle inequality,
\begin{equation*}
    \|a_0\|_{\mcA} - \|\hat{a}\|_{\mcA} = \|\theta_0\|_1 - \|\hat{\theta}\|_1 \leq \|\theta_0-\hat{\theta}\|_1 = \|\nu\|_1 \leq 2 \sqrt{\frac{s}{\gamma} \nu^\top  V \nu}  + \delta_{n,\zeta}/\lambda.
\end{equation*}

Therefore by Proposition~\ref{prop:reg-main-error-2}, as long as $\lambda \geq \delta_{n,\zeta}$
\begin{align*}
    \|\hat{a}-a_0\|_2 \leq~& \left(2\sqrt{\frac{s}{\gamma}} + \frac{\delta_{n,\zeta}}{\lambda} \frac{1}{\|\hat{a}-a_0\|_2}\right)\cdot \left(2 (\|\theta_0\|_{1}+1) \sqrt{\frac{\log(2p)}{n}} + \delta_{n,\zeta} + \lambda \sqrt{\frac{s}{\gamma}} \|\hat{a}-a_0\|_2\right).
\end{align*} 
The right hand side is upper bounded by the sum of the following four terms:
\begin{align*}
    Q_1 :=~& 2\sqrt{\frac{s}{\gamma}} \left(2(\|\theta_0\|_1+1) \sqrt{\frac{\log(2p)}{n}} + \delta_{n,\zeta}\right)\\
    Q_2 :=~& \left(\frac{\delta_{n,\zeta}}{\lambda} \frac{1}{\|\hat{a}-a_0\|_2}\right)\left(2 (\|\theta_0\|_1+1) \sqrt{\frac{\log(2p)}{n}} + \delta_{n,\zeta} \right)\\
    Q_3 :=~& 2 \lambda \frac{s}{\gamma} \|\hat{a}-a_0\|_2,\quad 
    Q_4 := \delta_{n,\zeta} \sqrt{\frac{s}{\gamma}}.
\end{align*}
If $\|\hat{a}-a_0\|_2 \geq \sqrt{\frac{s}{\gamma}} \delta_{n,\zeta}$ and $\lambda \leq \frac{\gamma}{8s}$, then
\begin{align*}
    Q_2 \leq~& 8 \frac{1}{\lambda}\sqrt{\frac{\gamma}{s}}\left(2 (\|\theta_0\|_1+1) \sqrt{\frac{\log(2p)}{n}} + \delta_{n,\zeta} \right),\quad 
    Q_3 \leq \frac{1}{4} \|\hat{a}-a_0\|_2.
\end{align*}
Thus bringing $Q_3$ on the LHS and dividing by $3/4$, we have
\begin{equation*}
    \|\hat{a}-a_0\|_2 \leq \frac{4}{3} (Q_1 + Q_2 + Q_4) = \frac{4}{3}\max\left\{\sqrt{\frac{s}{\gamma}}, \frac{1}{\lambda} \sqrt{\frac{\gamma}{s}}\right\} \left(20\, (\|\theta_0\|_1+1) \sqrt{\frac{\log(2p)}{n}} + 11 \delta_{n,\zeta}\right).
\end{equation*}
On the other hand if $\|\hat{a}-a_0\|_2\leq \sqrt{\frac{s}{\gamma}} \delta_{n,\zeta}$, then the latter inequality trivially holds. Thus it holds in both cases.
\end{proof}

\subsection{Semiparametric inference}

\begin{proof}[Proof of Corollary~\ref{cor:cone}]
    Note that
    $$
   \sqrt{\tilde{B}}\geq \|\theta\|_1=\|\theta_* + \nu\|_1 = \|(\theta_* + \nu)_T\|_1 + \|\nu_{T^c}\|_1 \geq \|\theta_*\|_1 - \|\nu_T\|_1 + \|\nu_{T^c}\|_1.
    $$
    So if 
    $
\sqrt{\tilde{B}}=\kappa+\|\theta_*\|_1
    $
    then $\|\nu_{T^c}\|_1 \leq \|\nu_T\|_1 + \kappa$ as desired. What remains to characterize is 
        $$
    \|\nu\|_1\leq \|\theta\|_1+\|\theta_*\|_1\leq \sqrt{\tilde{B}}+\|\theta_*\|_1=\kappa+2\|\theta_*\|_1 =\sqrt{B}.
    $$
\end{proof}

\begin{proof}[Proof of Proposition~\ref{prop:cone}]
    Let $\mcF$ be the subset of $(\mcA-a_*)_B$ with $\nu'\E[\phi(X)\phi(X)']\nu\leq \delta^2$. We first show that $\mcF \subset \mcG$ where
    $
    \mcG=\left\{\nu'\phi(X): \nu\in \R^p, \|\nu\|^2_1\leq B, \|\nu\|_1 \leq \frac{(\mu+1)\sqrt{s}}{\sqrt{\gamma}} \delta + \kappa\right\}.
    $
    In particular, since $ \nu' (\gamma I )\nu\leq \nu'\E[\phi(X)\phi(X)']\nu$
\begin{align*}
    \|\nu\|_1 &= \|\nu_{T}\|_1 + \|\nu_{T^c}\|_1 
    \leq (\mu + 1) \|\nu_T\|_1 + \kappa 
    \leq (\mu+1)\sqrt{s}\|\nu_T\|_2 + \kappa \\
    &\leq (\mu + 1)\sqrt{s} \|\nu\|_2 + \kappa
    \leq \frac{(\mu+1)\sqrt{s}}{\sqrt{\gamma}} \sqrt{\nu'\E[\phi(X)\phi(X)']\nu} + \kappa.
\end{align*}
Next we characterize the empirical metric entropy $H_n(\mcG,\epsilon)=\log\left[N\{\epsilon; \mcG;\ell_2\}\right]$. By \cite[Theorem 3]{zhang2002covering}, 
$$
H_n( \{\nu'\phi(X): \nu\in \R^p, \|\nu\|_1\leq a_1, \|\phi(X)\|_{\infty}\leq a_2 \} ,\epsilon) =O\left\{\frac{a_1^2a_2^2}{\epsilon^2}\log(p)\right\}
$$
which in our case implies 
$
H_n(\mcG,\epsilon)
=O\left[\left\{\frac{(\mu+1)\sqrt{s}}{\sqrt{\gamma}} \delta + \kappa\right\}^2\frac{1}{\epsilon^2}\log(p)\right]
.$
Note that the same holds replacing $\left\{\frac{(\mu+1)\sqrt{s}}{\sqrt{\gamma}} \delta + \kappa\right\}^2$ with $B$. 

Finally, we turn to~\eqref{eqn:metric-entropy-critical}:
$
 \int_{\frac{\delta^2}{2b}}^{\delta} \sqrt{\frac{\log\left[N\{\epsilon; \mcG;\ell_2\}\right]}{n}} d\epsilon \leq  \frac{\delta^2}{64b}.
$
The left hand side simplifies, for some universal constant $c_0$, as
\begin{align*}
     &\int_{\frac{\delta^2}{2b}}^{\delta} \sqrt{\frac{\log\left[N\{\epsilon; \mcG;\ell_2\}\right]}{n}} d\epsilon 
     \leq c_0 \left\{\frac{\mu \sqrt{s}}{\sqrt{\gamma}} \delta + \kappa\right\} \sqrt{\frac{\log(p)}{n}}  \int_{\frac{\delta^2}{2b}}^{\delta} \frac{1}{\epsilon} d\epsilon \\
     &= c_0 \left\{\frac{\mu \sqrt{s}}{\sqrt{\gamma}} \delta + \kappa\right\} \sqrt{\frac{\log(p)}{n}}  \log(2b/\delta) 
     =c_0 \frac{\mu \sqrt{s}}{\sqrt{\gamma}} \delta \sqrt{\frac{\log(p)}{n}}  \log(2b/\delta) + c_0  \kappa \sqrt{\frac{\log(p)}{n}}  \log(2b/\delta).
\end{align*}
The former term is bounded by $\delta^2/64b$ when
$
c_0 b\frac{\mu \sqrt{s}}{\sqrt{\gamma}}\sqrt{\frac{\log(p)}{n}}  \log(2b) \leq \frac{\delta}{\log(1/\delta)}.
$
The latter term is bounded by $\delta^2/64b$ when
$
\left\{c_0 b \kappa \sqrt{\frac{\log(p)}{n}}  \log(2b)\right\}^{1/2} \leq \frac{\delta}{\sqrt{ \log(1/\delta)}}.
$
An analogous argument gives the condition
$
\left\{c_0 b \sqrt{B} \sqrt{\frac{\log(p)}{n}}  \log(2b)\right\}^{1/2} \leq \frac{\delta}{\sqrt{ \log(1/\delta)}}.
$
Hence it suffices to take
$$
\delta_n=c_1  \cdot \log(n) \log(B) B^{1/2} \cdot \max\left\{ \mu \sqrt{\frac{s\log(p)}{n\gamma}}, \kappa^{1/2} \left(\frac{\log(p)}{n}\right)^{1/4}, \left(\frac{\log(p)}{n}\right)^{1/4} \right\}.
$$ 
The middle term dominates the final term.
\end{proof}

\begin{proof}[Proof of Corollary~\ref{cor:g_cr}]
    Since $\kappa=0$, we ignore the second term in Proposition~\ref{prop:cone}.
\end{proof}

\begin{proof}[Proof of Corollary~\ref{cor:m_cr}]
    The proof is similar to the proof of Corollary~\ref{cor:g_cr}. The key observation is that, for this function space, $m(z;g-g_*)=(\theta-\theta_*)(m\circ \phi)(x)=\nu(x)'\psi(x).$ 
\end{proof}

\begin{proof}[Proof of Corollary~\ref{cor:a_cr}]
 By Corollary~\ref{cor:sparse-linear-reg-ell1} and $\lambda=\gamma/8s$, with probability $1-\zeta$, 
$$
    \|\hat{\theta}\|_1 \leq \|\theta_*\|_1 + c_2 \frac{s}{\gamma} \sqrt{\frac{\log(p/\zeta)}{n}}=\sqrt{\tilde{B}}=\|\theta_*\|_1+\kappa.
$$
Hence by Corollary~\ref{cor:cone} we consider the restricted cone with $\mu=1$, $\kappa=c_2 \frac{s}{\gamma} \sqrt{\frac{\log(p/\zeta)}{n}}$, and $\sqrt{B}=\kappa+2\|\theta_*\|_1$. By Proposition~\ref{prop:cone},
 \begin{align*}
     \delta_n&=c_3  \cdot \log(n) \log(B) B^{1/2} \cdot \max\left\{ \sqrt{\frac{s\log(p)}{n\gamma}}, \sqrt{\frac{s}{\gamma}}\left(\frac{\log(p/\zeta)}{n}\right)^{1/4} \left(\frac{\log(p)}{n}\right)^{1/4}\right\} \\
     &\leq c_3  \cdot \log(n) \log(B) B^{1/2} \cdot \sqrt{\frac{s\log(p/\zeta)}{n\gamma}}.
 \end{align*}
 Finally, note that if $s=o\left\{\sqrt{\frac{n}{\log(p)}}\right\}$ then for $n$ sufficiently large,
 $
\sqrt{B}=\kappa+2\|\theta_*\|_1=c_2 \frac{s}{\gamma} \sqrt{\frac{\log(p/\zeta)}{n}}+2\|\theta_*\|_1 \leq 3\|\theta_*\|_1.
 $
\end{proof}

\subsection{Computational analysis}

\begin{lemma}[OFTRL; c.f. Proposition 7 of \cite{syrgkanis2015fast}]
Consider an online linear optimization algorithm over a convex strategy space $S$ and consider the OFTRL algorithm with a $1$-strongly convex regularizer with respect to some norm $\|\cdot\|$ on space $S$:
$
    f_t = \argmin_{f \in S} f^\top \left(\sum_{\tau\leq t} \ell_{\tau} + \ell_t\right) + \frac{1}{\eta} R(f).
$
Let $\|\cdot\|_*$ denote the dual norm of $\|\cdot\|$ and $R=\sup_{f\in S} R(f) - \inf_{f\in S} R(f)$. Then for any $f^*\in S$:
\begin{equation*}
    \sum_{t=1}^T (f_t-f^*)^\top \ell_t \leq \frac{R}{\eta} + \eta \sum_{t=1}^T \|\ell_t - \ell_{t-1}\|_* - \frac{1}{4\eta} \sum_{t=1}^T \|f_t - f_{t-1}\|^2.
\end{equation*}
\end{lemma}

\begin{proof}
\cite[Proposition 7]{syrgkanis2015fast} holds verbatim for any convex strategy space $S$ and not necessarily the simplex.
\end{proof}

\begin{lemma}[Approximate equilibrium; c.f. Lemma 4 of \cite{Rakhlin2013} and Theorem~25 of \cite{syrgkanis2015fast}]\label{lemma:appendix-minimax}
Consider a minimax objective $\min_{\theta\in \Theta} \max_{w\in W} \ell(\theta, w)$. Assume that $\Theta, W$ are convex sets, $\ell(\theta, w)$ is convex in $\theta$ for every $w$, and $\ell(\theta, w)$ is concave in $\theta$ for any $w$. Let $\|\cdot\|_\Theta$ and $\|\cdot\|_W$ be arbitrary norms in the corresponding spaces. Moreover, suppose that the following Lipschitzness properties are satisfied:
\begin{align*}
    \forall \theta\in \Theta, w, w'\in W: \left\|\nabla_{\theta}\ell(\theta, w)  - \nabla_{\theta}\ell(\theta, w')\right\|_{\Theta, *} \leq L \|w-w'\|_W\\
    \forall w\in W, \theta, \theta'\in \Theta: \left\|\nabla_{w}\ell(\theta, w)  - \nabla_{w}\ell(\theta', w)\right\|_{W, *} \leq L \|\theta-\theta'\|_\Theta
\end{align*}
where $\|\cdot\|_{\Theta, *}$ and $\|\cdot\|_{W, *}$ correspond to the dual norms of $\|\cdot\|_{\Theta}, \|\cdot\|_W$. Consider the algorithm where at each iteration each player updates their strategy based on:
\begin{align*}
    \theta_{t+1} =~& \argmin_{\theta\in \Theta} \theta^\top \left(\sum_{\tau\leq t} \nabla_{\theta}\ell(\theta_\tau, w_\tau) + \nabla_{\theta} \ell(\theta_t, w_t)\right) + \frac{1}{\eta} R_{\min}(\theta)\\
    w_{t+1} =~& \argmax_{w\in W} w^T \left(\sum_{\tau \leq t} \nabla_{w} \ell(\theta_\tau, w_\tau) + \nabla_w \ell(\theta_t, w_t)\right) - \frac{1}{\eta} R_{\max}(w)
\end{align*}
such that $R_{\min}$ is $1$-strongly convex in the set $\Theta$ with respect to norm $\|\cdot\|_\Theta$ and $R_{\max}$ is $1$-strongly convex in the set $W$ with respect to norm $\|\cdot\|_W$ and with any step-size $\eta \leq \frac{1}{4L}$. Then the parameters $\bar{\theta} = \frac{1}{T} \sum_{t=1}^T \theta_t$ and $\bar{w}=\frac{1}{T}\sum_{t=1}^T w_t$ correspond to an $\frac{2 R_*}{\eta \cdot T}$-approximate equilibrium and hence $\bar{\theta}$ is a $\frac{4 R_*}{\eta T}$-approximate solution to the minimax objective, where
\begin{equation*}
    R_* := \max\left\{ \sup_{\theta\in \Theta} R_{\min}(\theta) - \inf_{\theta\in \Theta} R_{\min}(\theta), \sup_{w\in W} R_{\max}(w)-\inf_{w\in W} R_{\max}(w)\right\}.
\end{equation*}
\end{lemma}
\begin{proof}
The lemma is essentially a re-statement of \cite[Theorem~25]{syrgkanis2015fast}, specialized to the case of the OFTRL algorithm and to the case of a two-player convex-concave zero-sum game. That result, in turn, adapts \cite[Lemma~4]{Rakhlin2013}. If the sum of regrets of players is at most $\epsilon$, then the pair of average solutions corresponds to an $\epsilon$-equilibrium \cite{FREUND199979}.
\end{proof}

\begin{proof}[Proof of Proposition~\ref{prop:sparse-optimization-ell1}]
    Let $R_E(x)=\sum_{i=1}^{2p} x_i \log(x_i)$. For the space $\Theta:=\{\rho\in \R^{2p}: \rho \geq 0, \|\rho\|_1\leq B\}$, the entropic regularizer is $\frac{1}{B}$-strongly convex with respect to the $\ell_1$ norm and hence we can set $R_{\min}(\rho)=B\, R_{E}(\rho)$. Similarly, for the space $W:=\{w\in \R^{2p}: w\geq 0, \|w\|_1=1\}$, the entropic regularizer is $1$-strongly convex with respect to the $\ell_1$ norm and thus we can set $R_{\max}(w)=R_E(w)$. For these regularizers, the update rules can be verified to have the closed form solution provided in Proposition~\ref{prop:sparse-optimization-ell1} by writing the Lagrangian of each OFTRL optimization problem and invoking strong duality. 
  Next, we verify the Lipschitz conditions. Since the dual of the $\ell_1$ norm is the $\ell_{\infty}$ norm,  $\nabla_{\rho}\ell(\rho, w) = \E_n[VV^\top] w + \lambda$ so
\begin{align*}
    \left\|\nabla_{\rho}\ell(\rho, w) - \nabla_{\rho}\ell(\rho, w')\right\|_{\infty} =\|\E_n[VV^\top] (w-w')\|_{\infty} \leq \|\E_n[VV^\top]\|_{\infty} \|w-w'\|_1\\
    \left\|\nabla_{w}\ell(\rho, w) - \nabla_{w}\ell(\rho', w)\right\|_{\infty} =\|\E_n[VV^\top] (\rho-\rho')\|_{\infty} \leq \|\E_n[VV^\top]\|_{\infty} \|\rho-\rho'\|_1.
\end{align*}
Therefore $L=\|\E_n[VV^\top]\|_{\infty}$. 
Since
\begin{align*}
    \sup_{\rho\in \Theta} B\, R_{E}(\rho) - \inf_{\rho\in \Theta} B\, R_E(\rho) =~& B^2 \log(B\vee 1) + B \log(2p),\quad 
    \sup_{w\in W} R_{E}(w) - \inf_{w\in W} R_E(w) = \log(2p)
\end{align*}
we can take $R_*=B^2 \log(B\vee 1) + (B+1) \log(2p)$. If $\eta = \frac{1}{4\|\E_n[VV^\top]\|_{\infty}}$, then after $T$ iterations, $\bar{\theta}=\bar{\rho}^+-\bar{\rho}^-$ is an $\epsilon(T)$-approximate solution to the minimax problem, with \begin{equation*}
\epsilon(T)=16\|\E_n[VV^\top]\|_{\infty} \frac{4B^2 \log(B\vee 1) + (B+1) \log(2p)}{T}.
\end{equation*}
Finally appeal to Lemma~\ref{lemma:appendix-minimax}.
\end{proof}

\newpage

\end{document}